\documentclass[11pt]{article}
\usepackage[utf8]{inputenc}
\usepackage{amsmath}
\usepackage{amsfonts}
\usepackage{amssymb}
\usepackage{amsthm}
\usepackage{mathtools}
\usepackage{tikz}
\usepackage{microtype}

\usepackage{geometry}
\usepackage{sectsty}
\usepackage[normalem]{ulem}

\usepackage{nohyperref}
\usepackage{url}

\sectionfont{\sffamily\bfseries\upshape\large}
\subsectionfont{\sffamily\bfseries\upshape\normalsize}
\subsubsectionfont{\sffamily\mdseries\upshape\normalsize}
\makeatletter
\renewcommand\@seccntformat[1]{\csname the#1\endcsname.\quad}
\makeatother

\makeatletter
\def\@maketitle{%
  \begin{center}%
  \let \footnote \thanks
    {\large \@title \par}%
    {\normalsize
      \begin{tabular}[t]{c}%
        \@author
      \end{tabular}\par}%
    {\small \@date}%
  \end{center}%
}
\makeatother

\usepackage{natbib} 

\theoremstyle{definition}

\newtheorem{theorem}{Theorem}

\newtheorem*{corollary*}{Corollary}
\newtheorem{lemma}{Lemma}
\newtheorem{definition}{Definition}
\newtheorem{example}{Example}

\newcommand{\randvar}[1]{#1^\ast}
\newcommand{\randvardesc}{asterisk }

\newcommand{\citeappendices}{}
\newcommand{\citeappendicesalt}{}

\title{\bf Simulation-Based Calibration Checking for Bayesian Computation:\\
The Choice of Test Quantities Shapes Sensitivity\footnote{We thank Garud Iyengar and Henry Lam for helpful discussions on theory and proofs (Appendix A) and David Yao for bringing attention to stochastic ordering which motivated delving into families of test quantities. We thank Feras Saad for alerting us that a previous version of this paper did not correctly reflect their contributions. This work was supported by the ELIXIR CZ research infrastructure project (Ministry of Youth, Education and Sports of the Czech Republic, Grant No: LM2023055), including access to computing and storage facilities; the Deutsche Forschungsgemeinschaft (DFG, German Research Foundation) under Germany’s Excellence Strategy --- EXC-2075 - 390740016 (the Stuttgart Cluster of Excellence SimTech); the U.S. National Science Foundation, National Institutes of Health, and Office of Naval Research; and the Academy of Finland Flagship programme: Finnish Center for Artificial Intelligence.}
\vspace{.1in}}
\author{Martin Modr\'ak\footnote{Institute of Microbiology of the Czech Academy of Sciences}, 
Angie H. Moon\footnote{Massachusetts Institute of Technology}, 
Shinyoung Kim\footnote{Department of Computer Science, Kookmin University}, 
Paul B\"{u}rkner\footnote{Department of Statistics, Technical University of Dortmund}, \\
Niko Huurre\footnote{Unaffiliated}, 
Kate\v{r}ina Faltejskov\'{a}\footnote{Institute of Organic Chemistry
and Biochemistry of the Czech Academy of Sciences},\,
Andrew Gelman\footnote{Department of Statistics and Political Science, Columbia University},\, 
Aki Vehtari\footnote{Department of Computer Science, Aalto University}\vspace{.1in}}
\date{19 Oct 2023}

\begin{document}

\maketitle

\begin{abstract}
Simulation-based calibration checking (SBC) is a practical method to validate computationally-derived posterior distributions or their approximations. In this paper, we introduce a new variant of SBC to alleviate several known problems. Our variant allows the user to in principle detect any possible issue with the posterior, while previously reported implementations could never detect large classes of problems including when the posterior is equal to the prior. This is made possible by including additional data-dependent test quantities when running SBC. We argue and demonstrate that the joint likelihood of the data is an especially useful test quantity. Some other types of test quantities and their theoretical and practical benefits are also investigated.
We provide theoretical analysis of SBC, thereby providing a more complete understanding of the underlying statistical mechanisms. We also bring attention to a relatively common mistake in the literature and clarify the difference between SBC and checks based on the data-averaged posterior. We support our recommendations with numerical case studies on a multivariate normal example and a case study in implementing an ordered simplex data type for use with Hamiltonian Monte Carlo. The SBC variant introduced in this paper is implemented in the \texttt{SBC} R package.
\end{abstract}

\section{Introduction}

Simulation-based calibration checking (SBC; \citealt{talts_sbc}) is a method to validate Bayesian computation, extending ideas from \citet{cook_gelman_rubin}.\footnote{The term in the literature is ``simulation-based calibration''; here we have added the word ``checking'' to emphasize that these methods do not themselves produce calibration; rather, they measure departure from calibration.}
While SBC is primarily intended for validating sampling algorithms such as MCMC, it can be used for validating any method implementing or approximating Bayesian inference. Published applications include variational inference \citep{yao_yes_2018} and neural posterior approximations \citep{radev2023jana}.  

Throughout this paper we assume an implicit and fixed Bayesian statistical model $\pi$ with data space $Y$ and parameter space $\Theta$. For $y \in Y, \theta \in \Theta$ the model implies the following joint, marginal, and posterior distributions:
\begin{gather*}
    \pi_\text{joint}(y, \theta) = \pi_\text{obs}(y | \theta) \pi_\text{prior}(\theta)\\
    \pi_\text{marg}\left(y \right) = \int_\Theta \mathrm{d} \theta \: \pi_{\text{obs}}(y | \theta) \pi_\text{prior}(\theta)\\
    \pi_\text{post}(\theta | y) = \frac{\pi_\text{obs}(y | \theta) \pi_\text{prior}(\theta)}{\pi_\text{marg}\left(y \right)}.
\end{gather*}

Typically, the posterior distribution $\pi_\text{post}$ is the target of inference but is impossible to evaluate directly. While many computational approaches exist for sampling from the posterior or its approximations, they may fail to provide a correct answer. Problems can arise from errors in how the algorithm or the statistical model are encoded or from inherent inability of the computational method to correctly handle a given model with a given dataset. 

\subsection{Self-consistency of Bayesian models}

To discover problems with computation, several classes of checks can be derived from self-consistency properties of statistical models. One such property concerns the data-averaged posterior \citep{geweke_getting_2004}:
\begin{equation}
\pi_\text{prior}(\theta) = \int_Y \mathrm{d} y \int_\Theta \mathrm{d}\tilde\theta  \: \pi_\text{post}(\theta |y) \pi_\text{obs}(y | \tilde\theta) \pi_\text{prior}(\tilde \theta).  \label{eq:data_averaged_posterior_introduction} 
\end{equation}

SBC relies on a different property that involves the joint distribution of prior and posterior samples from the same model \citep{cook_gelman_rubin}:
\begin{equation}
\pi_\text{SBC}(y, \theta, \tilde\theta) = \pi_\text{prior}(\tilde\theta) \pi_\text{obs}(y | \tilde\theta) \pi_\text{post}(\theta | y).
\label{eq:sbc_joint_distribution}
\end{equation}
Since $\pi_\text{obs}(y | \tilde\theta) \pi_\text{prior}(\tilde\theta) = \pi_\text{marg}\left(y \right)\pi_\text{post}(\tilde\theta | y)$, this implies,
\begin{equation}
\pi_\text{SBC}(y, \theta, \tilde\theta) = \pi_\text{marg}(y) \pi_\text{post}(\theta | y) \pi_\text{post}(\tilde\theta | y).      \label{eq:sbc_joint_2}
\end{equation}
Equation (\ref{eq:sbc_joint_2}) immediately shows that conditional on a specific data $y \in Y$, the distributions of $\theta$ and $\tilde\theta$ in Equations (\ref{eq:sbc_joint_distribution}) and (\ref{eq:sbc_joint_2}) are identical. In general, SBC-like checks are sensitive to different deviations from the correct posterior than checks based on the data-averaged posterior (see Section~\ref{sec:examples_summary} for more details). The two families of checks coincide when $Y$ has just a single element as in this case both reduce to directly comparing two distributions.

SBC and related methods employ two different implementations of the same statistical model and check if the results have the same distribution conditional on data. The first step is to define a \emph{generator} capable of directly simulating draws from $\pi_\text{prior}(\tilde\theta)$ and $\pi_\text{obs}(y | \tilde\theta)$, and the second step is to define a \emph{probabilistic program} that, in combination with a given \emph{posterior approximation algorithm}, samples from the posterior distribution $\pi_\text{post}(\theta | y)$. Each simulation from the generator yields,
\begin{align}
  \randvar{\tilde\theta} &\sim \pi_\text{prior}(\tilde\theta)    \notag\\
  \randvar{y} &\sim \pi_\text{obs}(y | \randvar{\tilde\theta})    \notag\\
  \theta_1, \dots \theta_M &\sim \pi_\text{post}(\theta | \randvar{y}),
  \label{eq:sbc_setup}
\end{align}
where $M$ is the number of posterior draws sampled. Where confusion is possible we use \randvardesc  to mark a random variable. We run many such simulations and then inspect the realized distributions of $\theta_1, \ldots, \theta_M$ and $\randvar{\tilde\theta}$ conditional on $\randvar{y}$. Specific calibration checking methods differ in how exactly they test the conditional equality of the two distributions. 

\subsection{Proposed SBC variant}

SBC has been believed to be insensitive to some classes of mismatches, and as described in \citep{talts_sbc} would not work for discrete variables. To remove those limitations, we argue for the following variant of the SBC check: First, project the potentially high-dimensional parameter and data space into a scalar \emph{test quantity} $f: \Theta \times Y \rightarrow \mathbb{R}$. Second, compute the rank of the prior draw in the posterior conditional on $y$. Specifically, we take the number of posterior sample draws where the test quantity is lower than in the prior draw, and, if there are any ties, choosing the rank randomly among the tied positions: 
\begin{align*}
   N_{\mathtt{less}} &:= \sum_{m=1}^M \mathbb{I} \left[f(\theta_m, y) < f(\tilde \theta, y) \right] \\
   N_{\mathtt{equals}} &:= \sum_{m=1}^M \mathbb{I} \left[f(\theta_m, y) = f(\tilde \theta, y) \right] \\
   K &\sim \mathrm{uniform}(0,  N_{\mathtt{equals}})\\
   N_\mathtt{total} &:= N_{\mathtt{less}} + K,
\end{align*}  
where $\mathbb{I}[P]$ denotes the indicator function for predicate $P$. The procedure simplifies if there are no ties, which will be true for most practical test quantities over models with continuous parameter space. When no ties occur, we have $N_\mathtt{total} = N_{\mathtt{less}}$. Then, if the probabilistic program and the generator implement the same probabilistic model, we have
\begin{equation}
    N_{\mathtt{total}} \sim \mathrm{uniform}(0, M).    
    \label{eq:sbc_sample}
\end{equation}
See Theorems~3 and 4 for a formal statement and proof. As a result, once we obtain a set of draws from empirical distribution of $N_{\mathtt{total}}$ via multiple simulations, we can perform a test for uniformity. The process is then repeated for all test quantities we want to consider. If we are using MCMC to sample from $\pi_\text{post}$, the posterior sample typically needs to be thinned to ensure that $\theta_1, \dots, \theta_M$ are  approximately independent \citep{talts_sbc, sailynoja_graphical_2021}. The overall SBC process is illustrated in Figure~\ref{fig:sbc_schema}.

\begin{figure}
    \centering
    \usetikzlibrary{positioning}

    \definecolor{priorcolor2}{HTML}{DBEEF3}
    \definecolor{priorcolor1}{HTML}{31859B}
    \definecolor{programcolor2}{HTML}{F2DCDB}
    \definecolor{programcolor1}{HTML}{953734}
    \definecolor{rankcolor2}{HTML}{FDEADA}
    \definecolor{rankcolor1}{HTML}{E36C09}
    \definecolor{unifcolor2}{HTML}{EBF1DD}
    \definecolor{unifcolor1}{HTML}{76923C}
    \resizebox{\textwidth}{!}{%
    \begin{tikzpicture}
    
    \tikzstyle{invisible} = [outer sep=0,inner sep=0,minimum size=0]
    
    \filldraw [color=priorcolor1,fill=priorcolor2] (0.55,3.1) rectangle (2.45,-2.9);
    \filldraw [color=programcolor1, fill=programcolor2] (3,3.1) rectangle (5.2,-2.9);
    \filldraw [color=rankcolor1, fill=rankcolor2] (5.7,3.1) rectangle (11.1,-2.9);
    \node at (1.45,3.6) {Generator};
    \node at (4,3.6) {\begin{tabular}{c} Probabilistic \\ program \end{tabular}};
    \node at (8.2,3.6) {Test quantity};

    \node[invisible] (prioranchor) {};
    \node (priordraw2) [right = 5mm of prioranchor] {$\tilde\theta^{(2)}$};
    \node (priordraw1) [above  = 12mm of priordraw2] {$\tilde\theta^{(1)}$};
    \node (priordrawdots) [below = 4mm of priordraw2.center] {$\vdots$};
    \node (priordrawS) [below = 10mm of priordrawdots] {$\tilde\theta^{(S)}$};
    
    \foreach \i in {1,2,S} {
      \node (datadraw\i) [right = 12mm of priordraw\i.west] {$y^{(\i)}$};
      \node (postdraw\i) [right = 12mm of datadraw\i.west] {$\theta^{(\i)}_1, \dots, \theta^{(\i)}_M$};
      \node[invisible] (priordatajoinhelper\i) [above = 9mm of datadraw\i.west] {};
      \node[invisible] (priordatajoin\i) [right = 10mm of priordatajoinhelper\i] {};
      \node (fpost\i) [right = 27mm of postdraw\i.west] {$f\left(\theta^{(\i)}_1, y^{(\i)}\right), \dots, f\left(\theta^{(\i)}_M, y^{(\i)}\right)$};
      \node (fprior\i) [right = 43mm of priordatajoin\i] {$f\left(\tilde\theta^{(\i)}, y^{(\i)}\right)$};
      \node (rank\i) [right = 60mm of fpost\i.west] {$N^{(\i)}_\text{total}$};
    }
    
    \foreach \g in {datadraw,postdraw,fpost,rank} {
        \node ({\g}dots) [below= 4mm of \g2.center] {$\vdots$};
    }

    \node[rectangle,  draw=unifcolor1, fill=unifcolor2] (uniformity) [right = 5mm of rank2] {\begin{tabular}{c} Uniformity \\ test \end{tabular}};

    \foreach \i in {1,2,S} {
      \draw[->] (prioranchor) -- (priordraw\i);
      \draw[->] (priordraw\i) -- (datadraw\i);
  
      \draw (priordraw\i) -- (priordatajoin\i);
      \draw (datadraw\i) -- (priordatajoin\i);
      \draw[->] (priordatajoin\i) -- (fprior\i);

      \draw[->] (datadraw\i) -- (postdraw\i);
      \draw[->] (postdraw\i) -- (fpost\i);
      \draw[->] (fprior\i.east) .. controls ++(15mm, 0mm) .. (rank\i);

      \draw[->] (datadraw\i) .. controls ++(20mm, 8mm) .. (fpost\i.north west);
      \draw[->] (fpost\i) -- (rank\i);
    }
    
     \draw[->] (rank1) -- (uniformity.north west);
     \draw[->] (rank2) -- (uniformity.west);
     \draw[->] (rankS) -- (uniformity.south west);

    \end{tikzpicture}
    }%
    \caption{\em Schematic representation of SBC with $S$ simulations. The generator is responsible for sampling from the prior distribution $\tilde\theta \sim \pi_\text{prior}(\tilde\theta)$ and from the observation model $y\sim \pi_\text{obs}(y \mid \tilde\theta)$. The draws from the observation model are then treated as input for the probabilistic program and the associated algorithm which takes $M$ posterior draws $\theta_1, \dots \theta_M$. Each test quantity projects the prior draw and the posterior draws (potentially using data) onto the real line, letting us compute a single rank ($N_\text{total}$). Finally, deviations from discrete uniform distribution are assessed numerically or visually.}
    \label{fig:sbc_schema}
\end{figure}
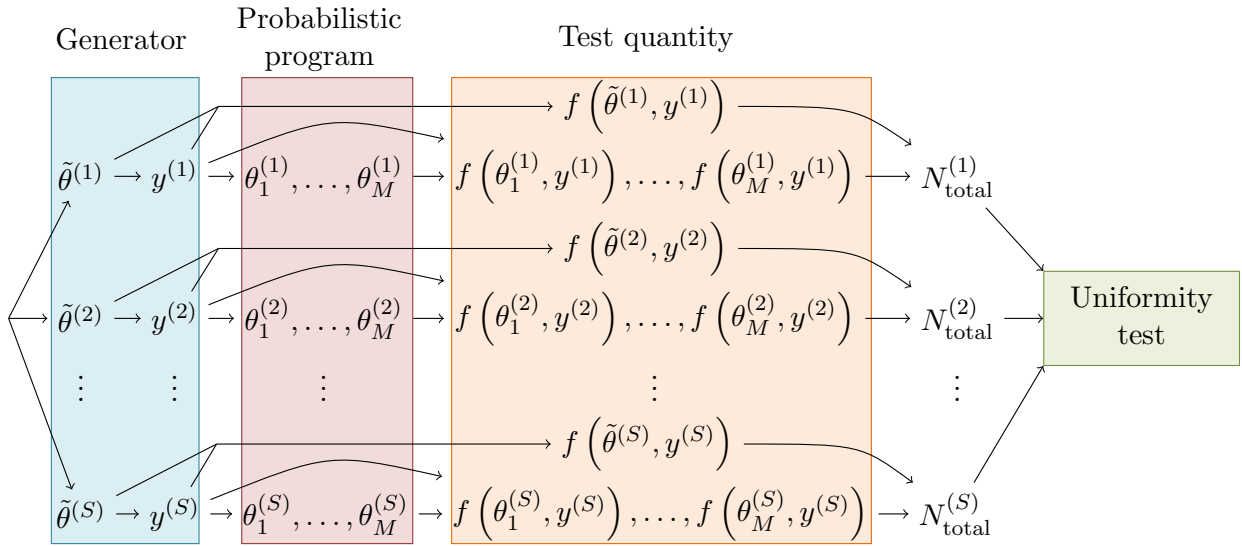

While it is possible to use numerical tests for uniformity with SBC, we generally prefer to use visualisations of the rank distribution as they are more informative than numerical summaries and discourage dichotomous thinking. Most prominent are rank histograms and plots of empirical cumulative distribution functions \citep{sailynoja_graphical_2021}.

Our proposed SBC variant improves upon the way SBC has been previously reported and used in two major ways:
\begin{itemize}
    \item We let test quantities depend on both data and parameters, while previous work only considered quantities that depend on the parameters. In practice, these test quantities were almost exclusively just the individual parameters themselves.
    \item Previous formulations of SBC required uniformity of $N_\mathtt{less}$. However, even if the probabilistic program is exactly correct, $N_\mathtt{less}$ will not be uniform if $\mbox{Pr}(N_\mathtt{equals} > 0) > 0$, that is, if ties can occur. With our improved SBC procedure, we can handle test quantities that have distributions with point masses and thus ties between $f(\tilde\theta, y)$ and $(f(\theta_1, y), \ldots, f(\theta_M, y))$. Resolving ties lets us use SBC for models with discrete parameters as well as in some other special cases, such as when a theoretically strictly positive test quantity suffers underflow and some prior/posterior sample draws are numerically zero. Random tie-braking has previously been used for checking that data-averaged posterior equals prior \eqref{eq:data_averaged_posterior_introduction} over discrete parameter spaces \citep{saad_family_2019}.
\end{itemize}

\subsection{Practical considerations}

SBC will be satisfied if the generator, probabilistic program, and posterior approximation algorithm are in harmony: The generator and the probabilistic program should correspond to the same data-generating process. At the same time the posterior approximation algorithm (including the associated tuning parameters) provides samples that have at most a negligible difference from the correct posterior for the probabilistic program, given the data simulated from the prior.  Failure indicates that at least one of the components is mismatched to the others. However, by itself, SBC cannot determine where exactly the problem lies. As a result, two broad uses of SBC arise:

\begin{itemize}
\item We have code to simulate data and a probabilistic program we trust, and the goal is to check that an algorithm correctly samples from the posterior, or
\item We have an algorithm that we trust is correct and trustworthy code to simulate data, and the goal is to check that we correctly implemented our probabilistic program.
\end{itemize}

In practice, those classes overlap and mix: we are rarely completely certain of the correctness of any algorithm, generator, or probabilistic program. Additionally, SBC as a simulation method has no way to inform us about a discrepancy between the process that generated real data and the assumptions of our statistical model. For reliable inference, SBC thus needs to be combined with other elements of Bayesian workflow that can detect model misspecification, such as posterior predictive checks or analysis of residuals \citep{gabry_visualization_2019,workflow_preprint, kay_residuals}. 

\subsection{Importance of test quantities}

It has been generally believed that methods based on Equation~\eqref{eq:sbc_joint_distribution}, including SBC, are never sensitive to some classes of mismatches between the generator and the probabilistic program---most notably that it is impossible to detect if the probabilistic program samples from the prior distribution and ignores the information in the data (e.g., Equation (1.3) of \citealt{lee_calibration_2019}; Appendix M.2 of  \citealt{pmlr-v130-lueckmann21a}; \citealp{schad_sbc_bf, pmlr-v161-zhao21b, ramesh_gatsbi_2022, cockayne_testing_2022}).

In this paper we show that the choice of test quantities greatly influences the usefulness and sensitivity of SBC. We show that using test quantities that depend on data makes it possible to detect any conceivable mismatch between the generator and the probabilistic program. Thus, we demonstrate that the belief in inherent limitations of SBC has relied on overly restrictive and sometimes plainly incorrect assumptions. We discuss useful classes of test quantities that have not been used so far and provide characterization of possible remaining undetected failures. We provide simulation studies as well as theoretical analysis of SBC to support our findings. We hope that our theoretical framework can serve as a basis for a better understanding of the properties of SBC and related methods. All of the techniques discussed are implemented in the \texttt{SBC} R package \citep{SBC_package}.

The rest of the paper is structured as follows: Section~\ref{sec:related} discusses related work, Section~\ref{sec:summary_theory} summarizes the theoretical results we derived, Sections~\ref{sec:numerical} and \ref{sec:real_world} show results of simulation and real-world case studies, and Section~\ref{sec:conclusions} discusses the results and our recommendations for practical use of SBC. 

\section{Related work}
\label{sec:related}

Prior contributions to validation of Bayesian computation can be roughly split into works that focus on the data-averaged posterior, those that focus on the SBC property, and other relevant works that do not directly invoke any self-consistency property.

\subsection{Data-averaged posterior}

The idea of using simulations via a generator to verify Bayesian computation can be traced back to \citet{geweke_getting_2004} who compared the moments of the prior and the data-averaged posterior distributions for multiple test quantities. That paper proposed to integrate a transition kernel for an MCMC sampler targeting $\pi_\text{post}$ into a scheme that samples $\pi_\text{joint}$ directly. This lets one obtain the data-averaged posterior from a single run of this sampler, potentially reducing the computational cost but increasing implementation burden. Geweke's formalism allows the test quantities to depend on data, although all the examples actually shown only depend on parameters. Comparing the mean vector and covariance matrix of the prior distribution and the data-averaged posterior distribution is also discussed by \citet{yu_assessment_2021}, who use repeated fits to build the data-averaged posterior. 

\citet{saad_family_2019} proposed a check for identity of two potentially high-dimensional discrete distributions by inspecting ranks generated by different total orderings over the parameter space. Their work is relevant in four ways: (1) it can be used to assess the data-averaged posterior criterion \eqref{eq:data_averaged_posterior_introduction} for discrete domains, (2) in close analogy to the use of test quantities in this paper, they focus on different orderings of the underlying domain and their different power to detect discrepancies, (3) they propose breaking ties in ordering uniformly at random in the same way we do, and (4) they prove some results that are analogous to or special cases of some of our theoretical results.

\subsection{SBC-like checks}

The identity of prior and posterior distributions conditional on a specific dataset as a tool to check computation was proposed by \citet{cook_gelman_rubin} and further refined by \citet{talts_sbc} who introduced SBC as it is currently used. Specific variants of SBC have been proposed for variational inference \citep{yao_yes_2018}, Bayes factors \citep{schad_sbc_bf}, and Gaussian processes \citep{mcleod_validating_2021}. SBC has also been used to validate likelihood-free inference methods including neural posterior approximators with normalizing flows \citep{radev2020bayesflow, radev2021bayesflow} and an SBC variant for checking joint calibration of such methods has been proposed and used in \citep{radev2023jana}.

\citet{gandy_unit_2021} proposed a procedure similar to SBC that can work with shorter sequences of Markov transitions than a full fit, reducing computational cost. This is, however, less relevant for algorithms that need a nontrivial warmup phase to adapt to the specific posterior (e.g., the adaptive Hamiltonian Monte Carlo sampler implemented in Stan; \citealp{stan_jss}). This is because warmup is a fixed cost that occurs during every model fit even if fewer post-warmup draws are needed. 

\citet{prangle_abc} proposed an SBC-like procedure for approximate Bayesian computation (ABC). They note that the possibility that the probabilistic program simply samples from the prior distribution cannot be ignored in this context and resolve this issue by separately inspecting ranks for some subsets of the simulated datasets. SBC is closely related to the \emph{coverage property} discussed by \citeauthor{prangle_abc}: when using $M$ posterior sample draws, SBC can be understood as checking for all posterior intervals of width $\alpha \in \left\{\frac{1}{M}, \dots, \frac{M - 1}{M}\right\}$ that the probability the interval contains the original simulated value of the test quantity is $\alpha$.

A broader framework for calibration of learning procedures has been proposed by \citet{cockayne_testing_2022}. There, Bayesian inference is just one example of procedures where calibration can be empirically verified with an SBC-like check. They distinguish between \emph{strong calibration} which corresponds to passing SBC (specifically continuous SBC as defined in Appendix A; \citeappendicesalt) for all measurable test quantities and \emph{weak calibration} which corresponds to having a correct data-averaged posterior (Equation~\ref{eq:data_averaged_posterior_introduction}). They however only consider test quantities that depend only on parameters. 

\subsection{Miscellaneous}

The problem of diagnosing and understanding computational issues is transformed by \citet{pmlr-v151-rendsburg22a}. Their approach tries to find a prior distribution that would make the probabilistic program and algorithm exactly match the generator.

Both \citet{grosse2016} and \citet{domke_easy_2021} proposed to use fits to multiple generated datasets to estimate the symmetrized KL-divergence between a distributional approximation to the correct posterior (e.g., Laplace or variational inference) and the true posterior. \citet{CusumanoTowner2017} described a method to compute the symmetrized KL-divergence between a gold standard posterior and an approximate posterior. 

\section{Theoretical results}
\label{sec:summary_theory}

The formalism required for SBC is relatively heavy on definitions and syntax, so for all results in this section we also provide plain English-language summaries. Proofs, expanded definitions (as required for proofs) and some additional discussion can be found in Appendix A \citeappendices. Some of the results for stochastic rank statistics (SRS) by \cite{saad_family_2019} can be understood as special cases of some of our theorems. Specifically, SRS assumes that the parameter space $\Theta$ is finite or countable and the goal is to directly compare two distributions, which is the same as assuming the data space $Y$ has just a single element.  We will refer to this special case as the \emph{SRS assumption}.

\begin{definition}[Posterior family, test quantity] A \emph{posterior family} $\phi$ assigns a normalized posterior density to each possible $y \in Y$. That is, a posterior family is a function $\phi: \Theta \times Y  \rightarrow  \mathbb{R^{+}}$ such that $\forall y: \int \mathrm{d}\theta \:\phi(\theta | y) = 1.$
For each $y$, we will denote the implied distribution over $\Theta$ as $\phi_y$. 
A \emph{test quantity} is any measurable function $f: \Theta \times Y  \to  \mathbb{R}$ 
\end{definition}

\begin{definition}[Sample rank CDF, sample Q, sample SBC]
Given a test quantity $f$, $M \in \mathbf{N}$ and a posterior family $\phi$. If $\theta_1, \dots, \theta_M \sim \phi_y$ we can define the following random variables:
\begin{align*}
   N_{\phi,f,\tilde\theta,y}^{\mathtt{less}} &:= \sum_{m=1}^M \mathbb{I} \left[f(\theta_m, y) < f(\tilde \theta, y) \right] \\
   N_{\phi,f,\tilde\theta,y}^{\mathtt{equals}} &:= \sum_{m=1}^M \mathbb{I} \left[f(\theta_m, y) = f(\tilde \theta, y) \right] \\
    K_{\phi,f,\tilde\theta,y} &\sim \mathrm{uniform}\left(0,  N_{\phi,f,\tilde\theta,y}^{\mathtt{equals}}\right) \\
    N_{\phi,f,\tilde\theta,y}^\mathtt{total} &:= N_{\phi,f,\tilde\theta,y}^{\mathtt{less}} + K_{\phi,f,\tilde\theta,y}.
\end{align*}    
The \emph{$M$-sample $Q$} is:
\begin{equation*}
  Q_{\phi,f}(i | y) := \int_\Theta \mathrm{d}\tilde\theta \: \pi_\text{post}(\tilde\theta|y)  \mbox{Pr} \left( N_{\phi,f,\tilde\theta,y}^\mathtt{total} \leq i \right)
\end{equation*}
We then say that \emph{$\phi$ passes $M$-sample SBC w.r.t.\ $f$} if, $\forall i \in 0, \dots, M - 1$,
\begin{equation*}
\int_Y \: \mathrm{d}y \: Q_{\phi,f}(i |y) \pi_\text{marg}(y) = \frac{i+1}{M + 1}.    
\end{equation*}
\end{definition}

This definition does not match immediately with the procedure we actually use to run SBC in practice but is more convenient for further analysis and is equivalent:

\begin{theorem}[Procedural definition of sample SBC]
\label{th:procedural_sbc}
A posterior family $\phi$ passes $M$-sample SBC w.r.t.\ $f$ if and only if given $\randvar{\tilde\theta} \sim \pi(\tilde\theta), \randvar{y} \sim \pi_\text{obs}(y | \randvar{\tilde\theta}), N^\mathtt{total} = N_{\phi,f,\randvar{\tilde\theta},\randvar{y}}^\mathtt{total}$ we have $N^\mathtt{total} \sim \mathrm{uniform}(0, M)$.    
\end{theorem}

Next, we define an idealized, continuous version of SBC that will be more amenable to theoretical analysis:

\begin{definition}[Continuous rank CDF, continuous $q$, continuous SBC]
We first define fitted CDF: $C^\pi_{\phi,f}: \bar{\mathbb{R}}\times Y\to [0,1]$, $C^\pi_{\phi,f}(s | y) := \int_{\Theta}\mathrm{d}\theta\,\mathbb{I}\left[f\left(\theta, y \right) \leq s\right]\phi\left(\theta| y\right)$ and fitted tie probability: $D^\pi_{\phi,f}: \bar{\mathbb{R}}\times Y\to [0,1]$, 
    $D^\pi_{\phi,f}(s | y) := \int_\Theta \mathrm{d}\theta \: \phi(\theta | y) \mathbb{I}\left[ f(\theta, y) = s \right]$

We then define the \emph{continuous } $q: [0,1] \times Y \to [0,1]$ as
\begin{equation*}
    q_{\phi,f}(x|y) := 
\int_\Theta \mathrm{d} \tilde\theta \: \pi_\text{post}(\tilde\theta | y) \mbox{Pr}\left(C_{\phi,f} \left(f \left(\left.\tilde\theta, y \right) \right| y \right) - U D_{\phi,f} \left(f \left(\left.\tilde\theta, y \right) \right| y \right) \leq x\right),
\end{equation*}
assuming $U$ is a random variable distributed uniformly over the $[0, 1]$ interval.

Finally, \emph{$\phi$ passes continuous SBC w.r.t.\ $f$} if $
\forall x \in [0, 1]: \int_Y \mathrm{d}y \: q_{\phi,f}(x|y) \pi_\text{marg}(y)  = x$.
\end{definition}

\subsection{Correctness}

With the definitions ready, we first establish that if a probabilistic program achieves uniform distribution of ranks in sample SBC for a given test quantity as $M \to \infty$, then it will satisfy continuous SBC as well.

\begin{theorem}[Sample SBC implies continuous SBC] 
\label{th:sample_implies_continous}
\ 
\begin{enumerate}
    \item For any fixed $y \in Y$ if as the number of sample draws $M \to \infty$ we have $\forall i \in \{0, \dots, M \}: Q_{\phi,f}(i | y) \to \frac{i + 1}{M + 1}$ then $\forall x \in [0, 1]: q_{\phi,f}(x | y) = x$.
    \item If as $M \to \infty$ we have $\forall i \in \{0, \dots, M \}: \int_Y \mathrm{d}y \, Q_{\phi,f}(i | y) \pi_\text{marg}(y) \to \frac{i + 1}{M + 1}$ then $\phi$ passes continuous SBC for $f$.
\end{enumerate}

\end{theorem}

Theorem~3 then shows that if a probabilistic program passes continuous SBC for a  given test quantity, it will pass sample SBC for all $M$. We then show that passing continuous SBC (and thus our SBC variant) is a necessary condition for the correctness of posterior estimation (Theorem~4). That is, the correct posterior will always produce uniformly distributed ranks, including for test quantities that may have ties (see also Examples~5 and 6 in Appendix B \citeappendicesalt). A special case of Theorem 4 under the SRS assumption was proven as Theorem 3.1 of \cite{saad_family_2019}.

\begin{theorem}[Continuous SBC implies sample SBC] 
\label{th:continuous_implies_sample}
For all $M \in \mathbf{N}$:
\begin{enumerate}
    \item For any $y \in Y$, if $\forall x \in [0,1]: q_{\phi,f}(x|y) = x$ then $\forall i \in \{0, \dots, M - 1 \}: Q_{\phi,f}(i | y) = \frac{i + 1}{M + 1}$.
    \item If $\phi$ passes continuous SBC w.r.t.\ $f$, then $\phi$ passes $M$-sample SBC w.r.t.\ $f$. 
\end{enumerate}

\end{theorem}

\begin{theorem}[Correct posterior and $q$] 
\label{th:SBC_correct}
For any $y \in Y$, if $\forall \theta \in \Theta: \phi(\theta | y) = \pi_\text{post}(\theta | y)$ then for any test quantity $f$ we have $\forall x \in [0, 1]: q_{\phi, f}(x | y) = x$.
\end{theorem}

\subsection{Characterization of SBC failures}
Still, many incorrect posteriors will also pass SBC for any given test quantity, so in Theorem~5 we characterize those situations. 

\begin{theorem}[Characterization of SBC failures]
\label{th:characterization_failures}
For all $y \in Y$ and $s \in \mathbb{R}: \\ \int_{\Theta}\mathrm{d}\theta\,\mathbb{I}\left[f\left(\theta, y \right) \leq s\right]\phi\left(\theta| y\right) = \int_{\Theta}\mathrm{d}\theta\,\mathbb{I}\left[f\left(\theta, y \right) \leq s\right]\pi_{\text{post}}(\theta | y)$ if and only if \\ $\forall x \in [0, 1]: q_{\phi,f}(x | y) = x$.
\end{theorem}

Not only does the correct posterior yield a uniform distribution of ranks when averaging over the whole data space $Y$, but the ranks are uniformly distributed even when we only consider simulations that yielded data in some $\bar{Y} \subset Y$. The reverse implication also holds: when the ranks are uniformly distributed for all subsets of the data space $\bar{Y} \subset Y$, then the implied posterior distribution of the test quantity under investigation has to be exactly correct. In other words, whenever SBC ``fails'' and the implied posterior distribution of a given test quantity is incorrect although the rank distribution is uniform, we can find a subset of the data space, where the ranks are non-uniform. It just so happens that all the deviations in various subsets cancel each other out perfectly.

An obvious application of Theorem~5 is that we could partition our simulations based on some features of the data space and investigate uniformity separately for each part, similarly to the procedure suggested by \citet{prangle_abc}. This however quickly runs into issues of multiple testing due to the lower number of simulations in each part. It is thus in our experience not practical except for the special case that interest lies only in some subset of the data space, so that the SBC checks can focus only on that data space of interest. This is a form of rejection sampling and can be practically useful if it is easy to formulate a criterion that constrains plausible real data sets but hard to construct a defensible prior distribution that would enforce this criterion implicitly. For example, prior information can be available on the plausible variance of an outcome across the whole population, which may be hard to express as a prior on coefficients associated with predictors (but see the approaches for linear models discussed in \citealp{zhang_bayesian_2022} and \citealp{aguilar2022intuitive}). 

\subsection{Data-dependent test quantities}
\label{sec:data_data_dependent_quantities}
The characterization of SBC failures discussed above  provides intuition why test quantities that depend on data are useful: If SBC passes for a test quantity $f$, but the posterior is in fact incorrect, we can always pick a test quantity $g$ that combines $f$ with some aspect of the data and ensures that the discrepancies in various parts of data space add up instead of canceling out. For example, we could have over-abundance of low ranks and under-abundance of high ranks in $Y_1 \subset Y$ and a matching under-abundance of low ranks and over-abundance of high ranks in $Y_2 \subset Y$. Setting 
$$
g(\theta, y) = \begin{cases}-f(\theta, y) & y \in Y_1 \\ f(\theta, y) & \text{ otherwise} \end{cases}
$$ 
will ensure over-abundance of high ranks in both $Y_1$ and $Y_2$. Since such a test quantity uses all the simulations, we do not lose power from reduced number of simulations.

An even stronger reason to use data-dependent test quantities is that they make SBC in some sense complete: If there is any difference between the correct posterior and the posterior implemented by the probabilistic program, there will exist a data-dependent test quantity that fails SBC. In fact, we can construct a specific test quantity that detects the failures, which is the ratio of the correct posterior density to the posterior density actually implemented by the probabilistic program. 

\begin{theorem}[Density ratio]
\label{th:density_ratio}
For any posterior family $\phi$, take $g\left(\theta,y\right)=\frac{\pi_{\mathrm{post}}\left(\theta\mid y\right)}{\phi\left(\theta|y\right)}$. Then $\phi$ passes continuous SBC w.r.t.\ $g$ if and only if $\pi_{\mathrm{post}}$ and $\phi$ are equal except for a set of measure $0$:

\begin{equation*}
\int_Y \mathrm{d}y \int_\Theta \mathrm{d}\theta \: \pi_\text{joint}(y, \theta) \mathbb{I}\left[ \pi_\text{post}(\theta | y) \neq \phi(\theta | y)\right] = 0. 
\end{equation*}
\end{theorem}

Here $g$ is not a practical test quantity, as it (a) depends on the specific probabilistic program we implemented and (b) requires that we already have the correct posterior density. However our empirical results in this paper, and our experience with using SBC in model development more generally, shows that the model likelihood $\pi_\text{obs}(y | \theta)$ is frequently useful as a general-purpose test quantity. This makes sense intuitively, as the likelihood is an important contributor to the density ratio. In their Theorem 3.1, \cite{saad_family_2019} proved an analogous result  under the SRS assumption, although relying on a different test quantity. Under the SRS assumption, they also show that the \emph{difference} of the two densities will fail $M$-sample SBC for all $M > 1$ (their Theorem 3.6) and has maximum power against discrepancies (their Theorem 3.7).

\subsection{Ignoring data}
We generalize the result that probabilistic programs sampling from the prior distribution will pass SBC against all test quantities that do not depend on data. 

\begin{theorem}[Incomplete use of data] 
\label{th:incomplete_use_data}
Assume a model $\pi$ with observation space $Y$ and parameter space $\Theta$, a space $Y^\prime$, and a measurable function $t: Y \rightarrow Y^\prime$. Denote the set $t^{-1}(y^\prime) = \{y \in Y: t(y) = y^\prime\}$. Consider the model $\pi^\prime$ with parameter space $\Theta$ and observation space $Y^\prime$ such that for all $\theta \in \Theta, y^\prime \in Y^\prime$:
\begin{align*}
\pi^\prime_\text{prior}(\theta) &= \pi_\text{prior}(\theta) \\
\pi^\prime_\text{obs}(y^\prime |\theta) &= \int_{t^{-1}(y^\prime)} \mathrm{d}y \: \pi_\text{obs}(y | \theta).
\end{align*}

Assume a test quantity $f^\prime: Y^\prime \times \Theta \rightarrow \mathbb{R}$. 
If we have a posterior family $\phi^\prime$ on $Y^\prime, \Theta$ such that $\phi^\prime$ passes continuous SBC w.r.t.\ $f^\prime$ and set test quantity $f: Y \times \Theta \rightarrow \mathbb{R}, f(\theta, y) = f^\prime(\theta, t(y))$ and posterior family $\phi$ on $\Theta, Y$ such that $\phi(\theta | y) = \phi^\prime(\theta | t(y))$ then $\phi$ passes continuous SBC w.r.t.\ $f$.
\end{theorem}

Here, the choice of $t$ lets us choose which aspects of the data are ignored, if $\forall y \in Y: t(y) = 1$, we recover the case where all data are ignored:  $\pi^\prime_\text{post}(\theta | y) = \pi_\text{prior}(\theta)$ and thus $\phi(\theta |y) = \pi_\text{prior}(\theta)$ will pass SBC w.r.t.\ $f$. If $t$ is a bijection, no information is lost. Other choices of $t$ then let us interpolate between those two extremes, for example ignoring just a subset of the data points, treating some data points as censored, rounding all data to integers.

\subsection{Detailed analysis of simple models and test quantities}
\label{sec:examples_summary}
Appendix B \citeappendices \ provides full theoretical analysis of SBC for simple models and test quantities where we can actually characterize all possible posterior distributions that will satisfy SBC. This is aimed at providing intuition on what SBC actually does and also serves as counterexamples to some claims.
In some literature (e.g., \citealp{lee_calibration_2019}, \citealp{pmlr-v130-lueckmann21a}, \citealp{schad_sbc_bf}, \citealp{grinsztajn_bayesian_2021}, \citealp{ramesh_gatsbi_2022},  \citealp{saad_family_2019}), it is assumed that SBC is based on the data-averaged posterior \eqref{eq:data_averaged_posterior_introduction}.
We show that this is incorrect: Example~2 not only explicitly constructs posterior distributions that will satisfy \eqref{eq:data_averaged_posterior_introduction} for some test quantity while not passing SBC, but also posterior distributions that pass SBC while not satisfying \eqref{eq:data_averaged_posterior_introduction}. One possibly more general lesson is that SBC is most naturally understood as enforcing constraints on the quantile function of the test quantity while having a correct data-averaged posterior is most naturally seen as constraint on the density of the test quantity.

This implies there might be some gains from using both the data-averaged posterior and SBC when verifying the correctness of Bayesian computation. We however suspect that the additional practical benefit of using the data-averaged posterior is small in the sense that the incorrect posteriors that pass SBC but are ruled out by Equation \eqref{eq:data_averaged_posterior_introduction} are mostly contrived and unlikely to be a result of a computational problem or an inadvertent mistake. Lemma 2.19 of \citet{cockayne_testing_2022} proves that if a posterior passes SBC for \emph{all possible} test quantities that do not depend on data, it will have the correct data-averaged posterior for all test quantities that do not depend on data, so SBC is stronger at least in the limit of using infinitely many test quantities. We leave a more thorough examination of the relationship between data-averaged posterior and SBC as future work. 

Additionally, we show the behavior of SBC when ties are present, whether induced by a test quantity (Example~5) or by discrete parameter space (Example~6). Discrete parameter spaces may induce additional structure on the space of posterior families passing SBC.

\subsection{Monotonic transformations of test quantities}
\label{sec:monotonic}
Finally, transforming a test quantity by a strictly monotonic function produces equivalent SBC results:
\begin{theorem}[Monotonic transformations]
\label{th:monotonous_transformations}
Assume test quantities $f,g$ and a set of measurable functions $h_y: \mathbb{R} \to \mathbb{R}$ such that $\forall y \in Y ,\theta \in \Theta: f(\theta, y) = h_y(g(\theta_1, y))$ and a posterior family $\phi$. If either for all $y \in Y: h_y$ is strictly increasing or  for all $y \in Y: h_y$ is strictly decreasing then 1) $\phi$ passes continuous SBC w.r.t.\ $f$ if and only if $\phi$ passes continuous SBC w.r.t.\ $g$ and 2) $\phi$ passes $M$-sample SBC w.r.t.\ $f$ if and only if $\phi$ passes $M$-sample SBC w.r.t.\ $g$.    
\end{theorem}

The result cannot be easily strengthened as many non-monotonic transformations lead to different, non-equivalent SBC checks. Example~3 shows that flipping the ordering of values only for some subset of the data space yields a different SBC check. Example~4 shows that we can also obtain a different check if we combine a test quantity with a non-monotonic bijection, and Example~5 shows the same for the case when a whole range of values is projected onto a single point. In all those examples, the transformed test quantities rule out some sets of posteriors that pass SBC for the original quantity, but there are also sets of posteriors not passing SBC for the original quantity but passing SBC for the transformed quantity.

\section{Numerical case studies}
\label{sec:numerical}

The theoretical analysis in previous section primarily deals with the behavior of SBC in the limit of both infinitely many posterior draws per fit and infinitely many simulations. Here, we further support the results by numerical experiments which let us understand not only whether a certain problem is detectable at all but also how much computational effort is required for SBC to detect the problem.

\subsection{Setup}

To illustrate some of the properties of various types of test quantities, we use a simple multivariate normal model,
\begin{align}
\mathbf{\mu} &\sim \mbox{MVN}(0, \mathbf{\Sigma}) \notag \\ 
\mathbf{y}_1, \ldots, \mathbf{y}_n &\sim \mbox{MVN}(\mathbf{\mu}, \mathbf{\Sigma}) \notag\\
\mathbf{\Sigma} &= \left(\begin{matrix}
 1 & 0.8 \\
 0.8 & 1 \\
\end{matrix}\right),
\end{align}
%
%
where the two-element vector $\mu$ is the target of inference and $\mathbf{y}_1, \ldots, \mathbf{y}_n$ are observed.
Introducing $\bar{\mathbf{y}} = \frac{1}{n}\sum_{i = 1}^{n} \mathbf{y}_i$, the correct analytic posterior is $\mbox{MVN}\left(\frac{N\bar{\mathbf{y}}}{n + 1}, \frac{1}{n + 1}\mathbf{\Sigma}\right)$.  Unless mentioned otherwise we will use $n = 3$.
In most previous use cases of SBC, the only test quantities used would have been the parameters themselves, that is, the elements of $\mathbf{\mu}$ in the above example. Below, we also check a host of derived quantities: the sum, difference, and product of the $\mathbf{\mu}$ elements, the joint likelihood of all the data, and pointwise likelihoods for the first two data points.

To quantify the discrepancy between an observed distribution of posterior ranks and the uniform distribution, we take the likelihood of observing the most extreme point on the empirical CDF if the rank distribution was indeed uniform:
\begin{equation}
\gamma = 2 \min_{i\in \{1, \dots, M+1\}}\left(\min\{\text{Bin}(R_i | S, z_i), 1 - \text{Bin}(R_i - 1 | S, z_i)\}\right).
\end{equation}
Here, $M$ is the number of draws in the sample obtained from the posterior, $S$ is the number of simulations (and thus the number of observed ranks), $z_i = \frac{i}{M + 1}$ is the expected proportion of observed ranks smaller than $i$, $R_i$ is the observed count of ranks smaller than $i$, and $\text{Bin}(R | S, p)$ is the CDF of the binomial distribution with $S$ trials and probability of success $p$ evaluated at $R$. This metric was introduced in a paper by \citet{sailynoja_graphical_2021}, where we can also find computational methods to evaluate the distribution of $\gamma$ under uniform distribution of ranks for given $M$ and $S$. Our primary metric of interest would then be $\log\frac{\gamma}{\bar{\gamma}}$, where $\bar{\gamma}$ is the 5th percentile of the null distribution. That is, if you adopt a hypothesis-testing framework, then $\log\frac{\gamma}{\bar{\gamma}} < 0$ implies a rejection of the hypothesis of uniform distribution at the 5\% level. Having $\log\frac{\gamma}{\bar{\gamma}} < 0$ also corresponds to situations where visual checks of the ECDF plots would show problems (for a single test quantity). This diagnostic is typically more sensitive than the Kolmogorov-Smirnoff or $\chi^2$ test.

\begin{figure}
    \centering
    \includegraphics[width=\textwidth]{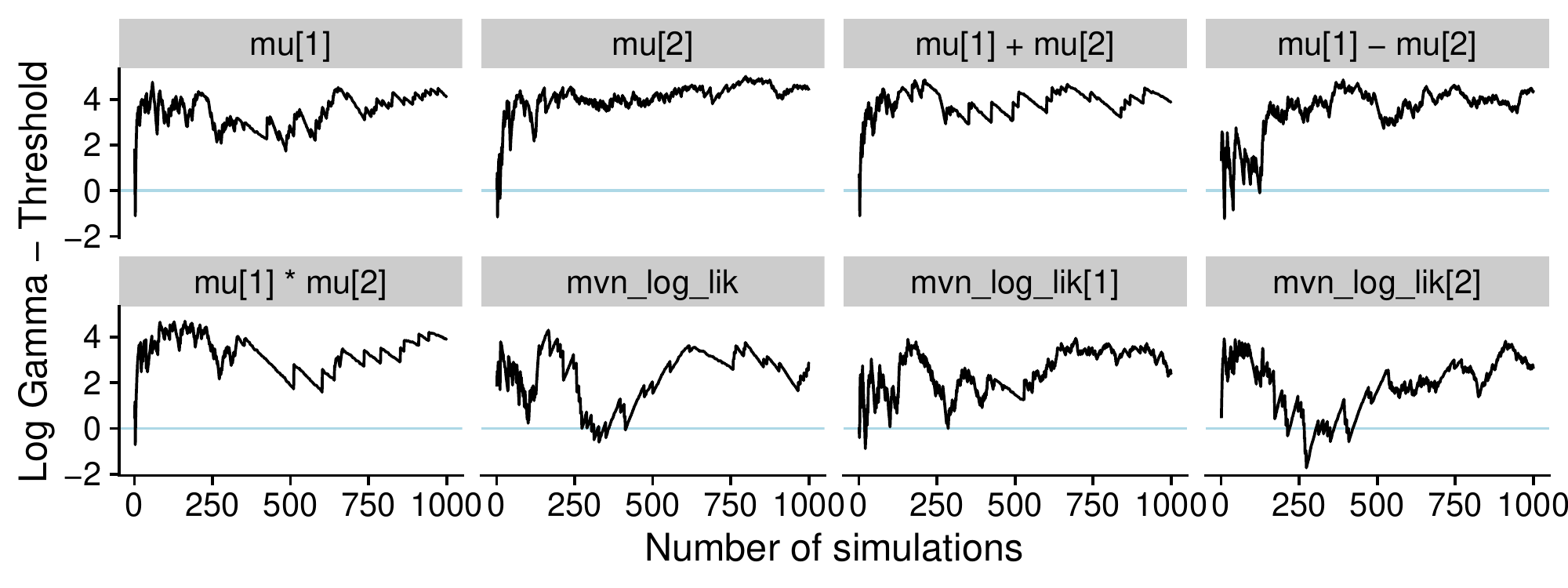}
    \vspace{-.3in}
    \caption{\em Case study 1: Evolution of the difference between the gamma statistic and threshold ($\log \bar{\gamma}$) for rejecting uniformity at 5\% for the correct posterior. \texttt{mvn\_log\_lik[1]} and \texttt{mvn\_log\_lik[2]} are the pointwise likelihoods $\pi(\mathbf{y}_1|\mu)$ and $\pi(\mathbf{y}_2|\mu)$ respectively, while \texttt{mvn\_log\_lik} is the joint likelihood. As expected when using a 5\% level for rejection, false positives (values below the threshold) do happen, but they tend to correspond to only small discrepancies.}
    \label{fig:hist_correct}
\end{figure} 

\subsection{Correct posterior - Case study 1}

Figure~\ref{fig:hist_correct} shows how the $\gamma$ statistic evolves in a fairly typical SBC run as we add more simulations using a probabilistic program that samples from the correct posterior. There is some variability, but most of the time all quantities would indicate uniformity and if they indicate some non-uniformity, the discrepancies tend to be small so we are unlikely to reject this model as incorrect.

\subsection{Ignoring data - Case studies 2--4}

\begin{figure}
    \centering
    \includegraphics[width=\textwidth]{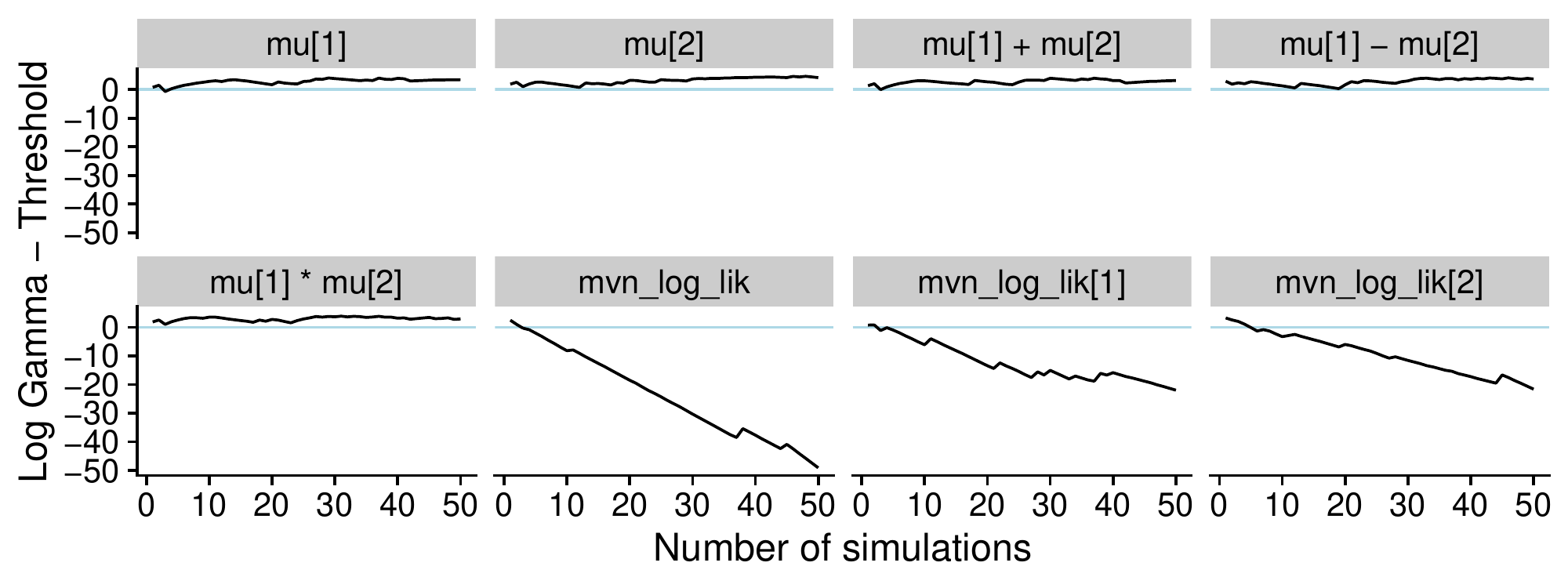}
    \vspace{-.3in}
    \caption{\em Case study 2: Evolution of the difference between the gamma statistic and threshold for rejecting uniformity at 5\% for an incorrect posterior that equals the prior. Note how quickly large discrepancies accumulate for the likelihood-based quantities, despite the horizontal axis being zoomed to show only first 50 simulations.}
    \label{fig:hist_prior_only}
\end{figure}

For comparison, case study 2 (Figure~\ref{fig:hist_prior_only}) shows the evolution of the same quantities for a typical run with an incorrect posterior that is completely equal to the prior. All quantities that do not depend on data pass SBC, barring small short-term deviations as seen for the correct posterior. But all the likelihood-based quantities start showing big discrepancies after just a handful of simulations. While the overall distribution of ranks for the parameters themselves is uniform, when we look separately at data with large average $y$ and low average $y$, the ranks are strongly non-uniform in both regions (Figure~\ref{fig:rank_hist_prior_only_split}).

In case study 3, we observe similar behaviour for the posterior that ignores only the first data point; see Figure~\ref{fig:hist_one_missing}.  The biggest difference is that that now the pointwise likelihood for the second data point---which was not ignored---passes SBC, while the joint likelihood as well as the pointwise likelihood for the first ignored data point show problems. Additionally, the pointwise likelihood for the ignored data point now shows bigger discrepancy than the joint likelihood. For both quantities, the discrepancy is smaller and requires about $S = 20$ simulations to reliably uncover, because ignoring a single data point produces a posterior that is closer to the correct one than when ignoring all the data. For case study 4 we increase the number of data points to $n = 20$ (Figure~\ref{fig:hist_one_missing_20}), ignoring just a single data point produces a posterior that is  close to correct and even after 1000 simulations, the discrepancy for the joint likelihood is  small. The pointwise likelihood for the first (ignored) data point still detects the problem relatively quickly.

More generally, if the model (partially) ignores data, then adding a test quantity that involves both data and parameters can detect this failure. Specifically adding the joint log-likelihood of the data as a derived quantity seems to be a useful default. If only a small part of the data is missing, using the joint likelihood in SBC will turn it into a problem of precision. Missing just a single datapoint in a large dataset (e.g., an off-by-one error in the probabilistic program) may change the posterior only slightly and be undetectable with realistic computational effort.

\begin{figure}
    \centering
    \includegraphics[width=\textwidth]{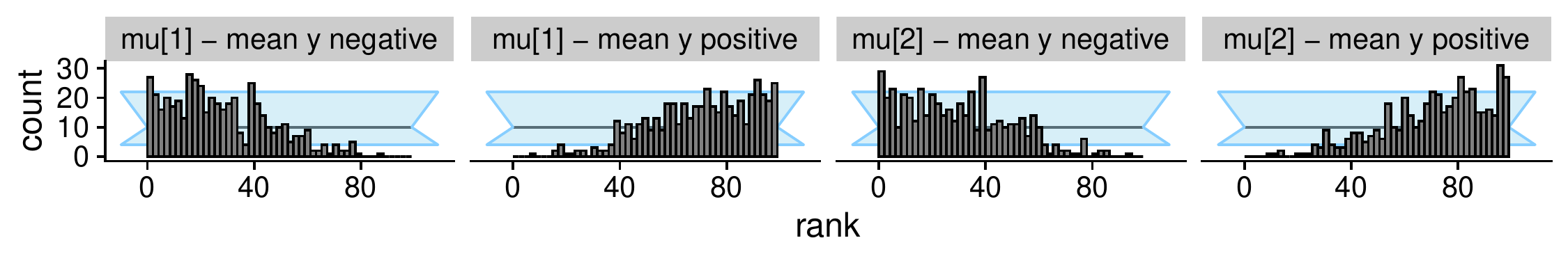}
    \vspace{-.3in}
    \caption{\em Case study 2: Rank distribution for the elements of $\mu$ split by the average value of the corresponding $y$ elements for the incorrect posterior that is completely equal to the prior. The distributions for the two cases exactly compensate to make the overall distribution uniform. The gray horizontal line represents exact uniform distribution and the blue areas represent an approximate 95\% prediction interval for the observed ranks, assuming uniform rank distribution.}
    \label{fig:rank_hist_prior_only_split}
\end{figure}

\begin{figure}
    \centering
    \includegraphics[width=\textwidth]{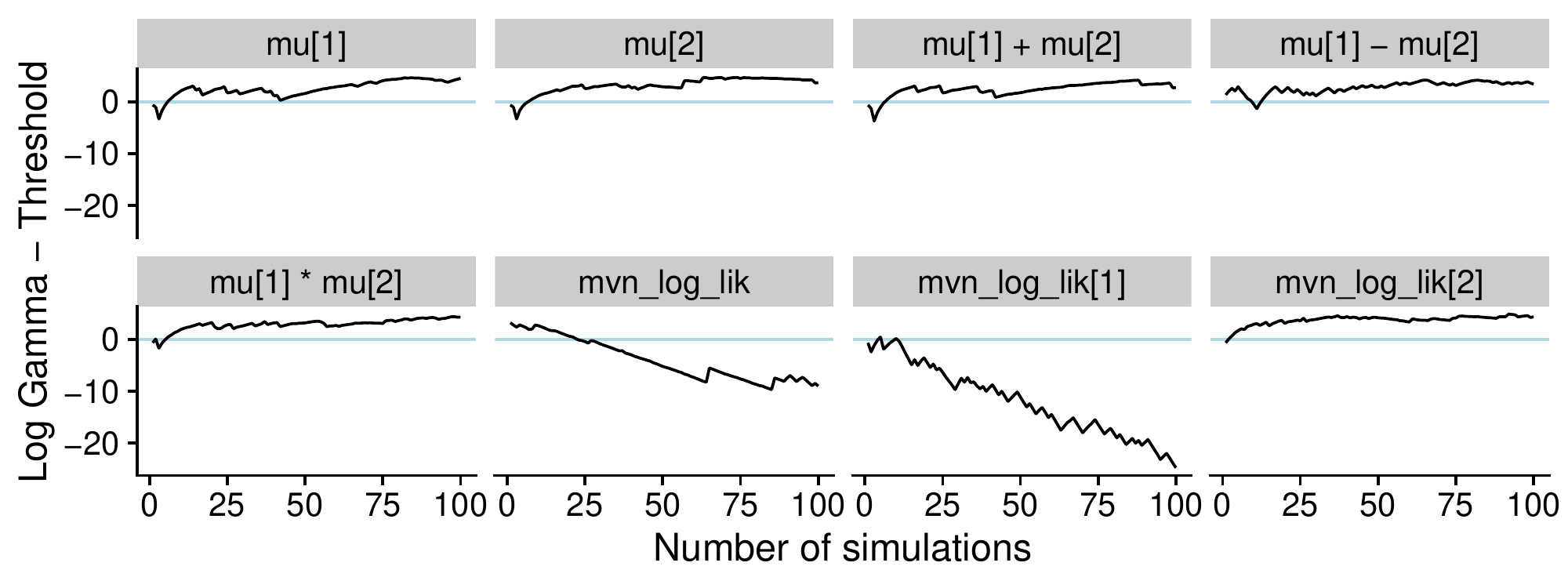}
    \vspace{-.3in}
    \caption{\em Case study 3: Evolution of the difference between the gamma statistic and threshold for rejecting uniformity at 5\% for an incorrect posterior that ignores the first datapoint among a small data set ($n = 3$).}
    \label{fig:hist_one_missing}
\end{figure} 

\begin{figure}
    \centering 
    \includegraphics[width=\textwidth]{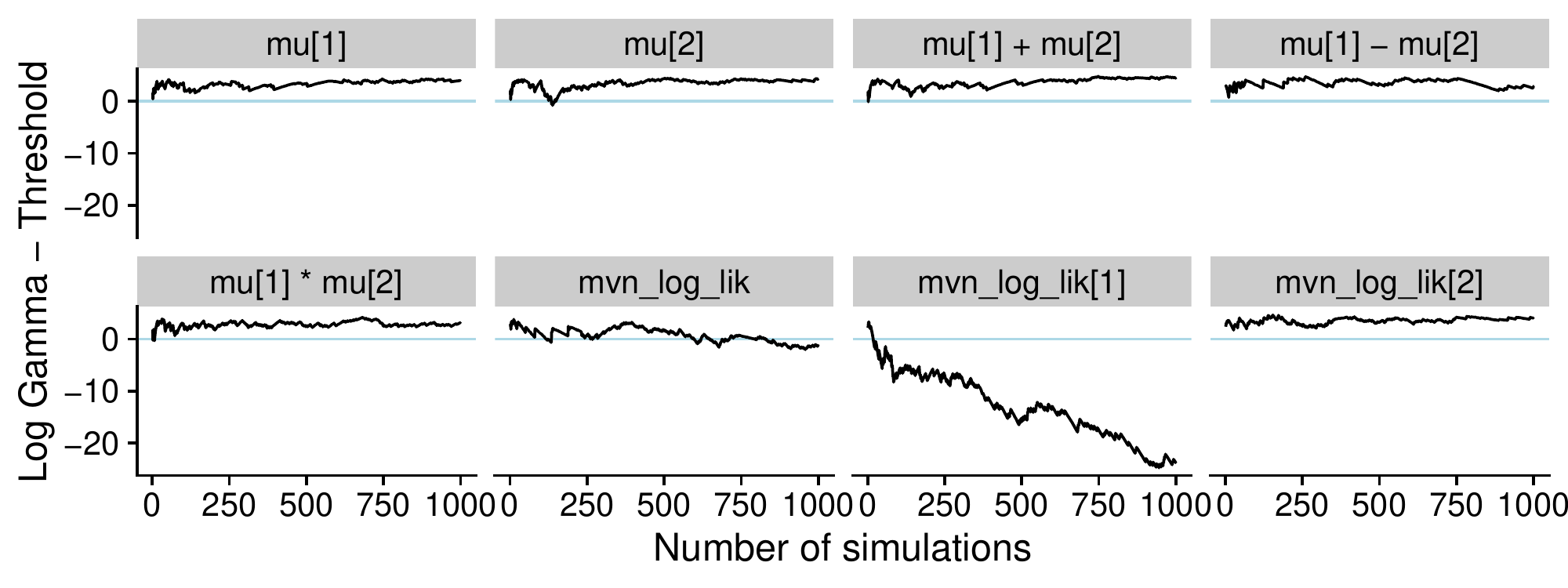}
    \vspace{-.3in}
    \caption{\em Case study 4: Evolution of the difference between the gamma statistic and threshold for rejecting uniformity at 5\% for an incorrect posterior that ignores the first datapoint among a larger dataset ($n = 20$).}
    \label{fig:hist_one_missing_20}
\end{figure} 

\subsection{Incorrect correlations - Case study 5}

Suppose we have an incorrect posterior that has the correct marginal distributions for both parameters, i.e., sampling is done from independent univariate normal distributions, $\mu_i \mid \mathbf{y}_1, \ldots, \mathbf{y}_n \sim N\left(\frac{n\bar{\mathbf{y}}_i}{n + 1}, \frac{1}{n + 1}\mathbf{\Sigma}_{i,i}\right)$. The evolution of the discrepancy as simulations are added is shown in Figure~\ref{fig:hist_corr}. If the test quantities are the univariate parameters,  SBC passes without any indication of problems, while the likelihood-based quantities as well as the difference, product, and sum of the variables show problems relatively quickly. The joint likelihood is the first to show serious issues.

\begin{figure}
    \centering
    \includegraphics[width=\textwidth]{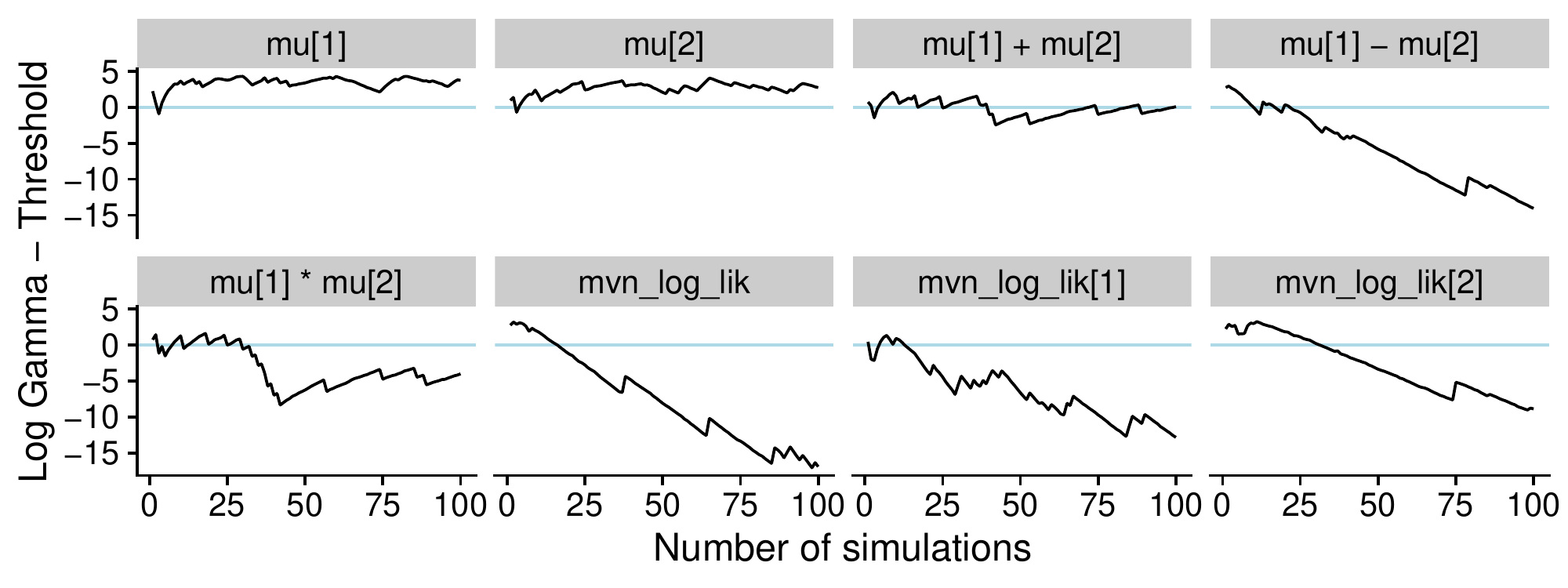}
    \vspace{-.3in}
    \caption{\em Case study 5: Evolution of the difference between the gamma statistic and threshold for rejecting uniformity at 5\% for incorrect posterior that has wrong correlation structure.}
    \label{fig:hist_corr}
\end{figure} 

If the inference  does not represent correlations in the posterior correctly, this should as well manifest in an SBC failure for some function of the parameters. This can be directly targeted by using products (``interactions'') of model parameters, but the log-likelihood once again seems to be generally useful as a highly nonlinear function of all model parameters.

\subsection{Less plausible problems - Case study 6}

In this subsection our results get less practical and more theoretical. The (partially) unused data case may easily arise in practice due to a bug in the probabilistic program such as an indexing bug or a deficient overall approach. For example, an approximate Bayesian computation algorithm may not learn from the data at all and just stick to the prior \citep{prangle_abc}. Incorrect correlations or more general higher-order structure of the posterior may also easily arise due to a problem with an approximate inference algorithm. For example, mean-field variational inference will never recover any correlations by design. Beyond those examples, we have found it hard to find incorrect probabilistic programs that would satisfy the SBC identity and could plausibly arise from unintentional mistakes in program code or problems with an algorithm. We see this as anecdotal evidence that SBC augmented with a few well-chosen test quantities that probe usage of data and higher order posterior structure such as the likelihood can robustly detect these kinds of mistakes. That said, for specific models, wide sets of artificial counterexamples that incorrectly pass SBC can be constructed.

In case study 6, we show a specific case of a more general class of setups where we can create an incorrect posterior approximation that produces overabundance of low ranks for datasets with average of $\mathbf{y}$ positive and compensates by producing overabundance of high ranks for other datasets. If this is done right, the test quantity will pass SBC. The distribution of the ranks conditional on the average of $\mathbf{y}$ for one such setup is shown in Figure~\ref{fig:rank_hist_non_mon_split}---here we transform draws from the correct posterior distribution by first applying the correct CDF, manipulating the results to achieve the desired shape of ranks and then transform back via the quantile function. See the associated code for more details. As seen in Figure~\ref{fig:hist_non_mon}, when averaging over all datasets, SBC indeed passes for the univariate parameter test quantities, but if we instead look at, say, the absolute value of $\mu$ (as well as some other non-monotonic transformations of $\mu$), we immediately see problems as now some of the previously low ranks flip to high ranks and the discrepancies accumulate instead of canceling each other. In this particular case, the problem is also eventually picked up by the product of the $\mu$ values and with enough simulations even by the joint likelihood, but there is no guarantee this will always happen. In general, non-monotonic transformations can discover incorrect posteriors that would be otherwise hidden when looking at the original variables. Still, the practical relevance of non-monotonic transforms in SBC is, in our view, likely limited, as it required careful work to construct posteriors that manifested this behaviour. We were unable to find even remotely plausible scenarios where an issue with Bayesian computation was best discovered by using a non-monotonic transformation of another test quantity.

\begin{figure}
    \centering
    \includegraphics[width=\textwidth]{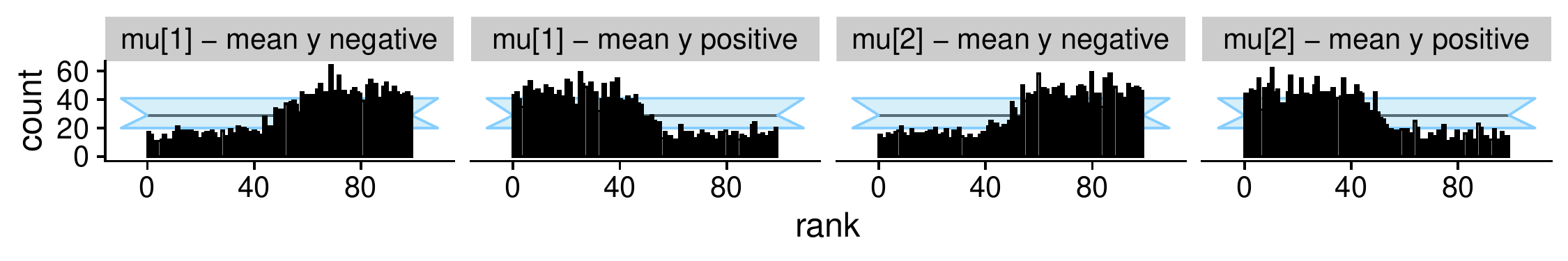}
    \vspace{-.3in}
    \caption{\em Case study 6: Rank distribution for the elements of $\mu$ split by the average value of the corresponding $y$ elements for the incorrect posterior that satisfies SBC for individual parameters. The distributions for the two cases exactly compensate to make the overall distribution uniform. The gray horizontal line represents exact uniform distribution and the blue areas represent an approximate 95\% prediction interval for the observed ranks, assuming uniform rank distribution}
    \label{fig:rank_hist_non_mon_split}
\end{figure} 

\begin{figure}
    \centering
    \includegraphics[width=\textwidth]{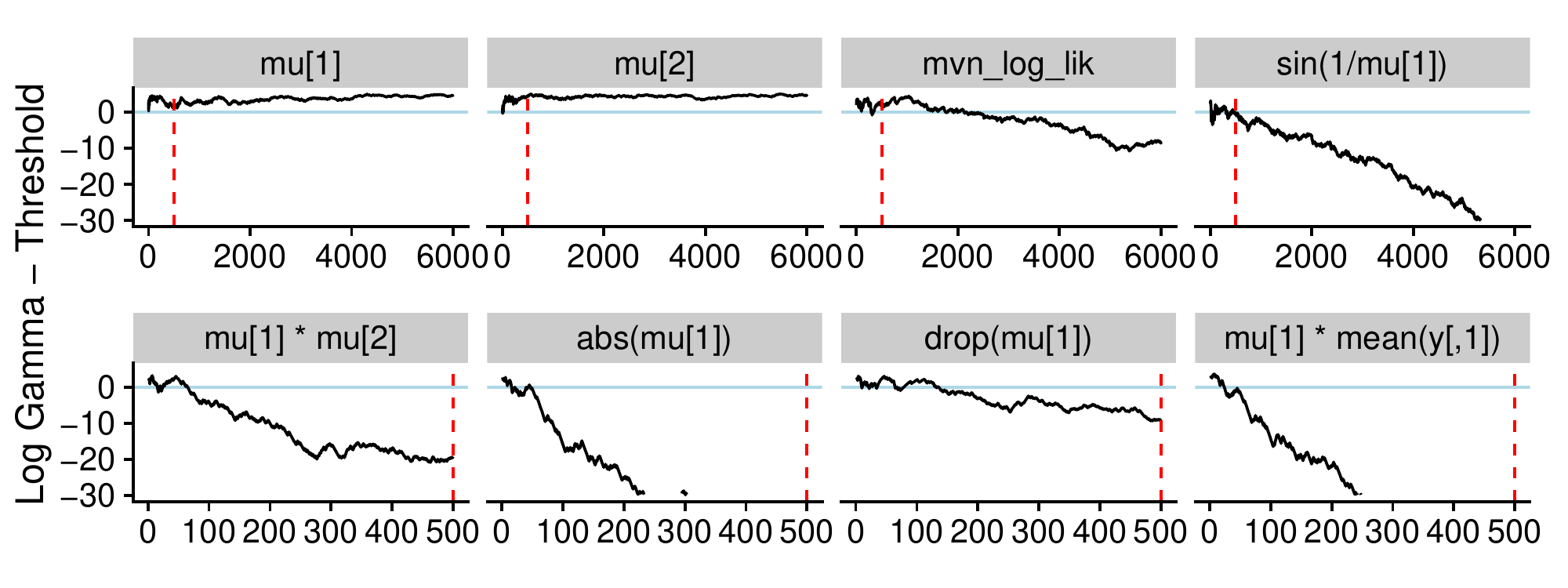}
    \vspace{-.3in}
    \caption{\em Case study 6: Evolution of the difference between the gamma statistic and threshold for rejecting uniformity at 5\% for incorrect posterior that satisfies SBC for individual parameters.
    Note the different horizontal axis between top row (quantities that detect the problem slowly or not at all) and bottom row (quantities that detect the problem quickly). The vertical red dashed line marks 500 simulations. We only show quantities derived from the first element of $\mu$; the situation is analogous for the second element. The \texttt{drop(mu[1])} quantity is defined as  $\mu_1 \text { if } \mu_1 < 1$ and as $\mu_1 - 5 \text{ otherwise}$.}
    \label{fig:hist_non_mon}
\end{figure} 

\subsection{Small discrepancies - Case study 7}

A final case study considers small discrepancies in the posterior. To be specific, we introduce a small bias in the posterior drawn from $\mbox{normal}(0, 0.3)$ independently for each simulation and element of $\mu$. The resulting SBC history is shown in Figure~\ref{fig:hist_small_change}. While all of the monitored quantities will eventually show the problem, the likelihood-based quantities and the difference of $\mu$ do that noticeably sooner than others. This demonstrates that derived quantities can somewhat improve precision of SBC: small changes in the univariate marginals can result in big (and thus easy to detect) changes for some test quantities combining the univariate marginals with data and other parameters.

\begin{figure}
    \centering
    \includegraphics[width=\textwidth]{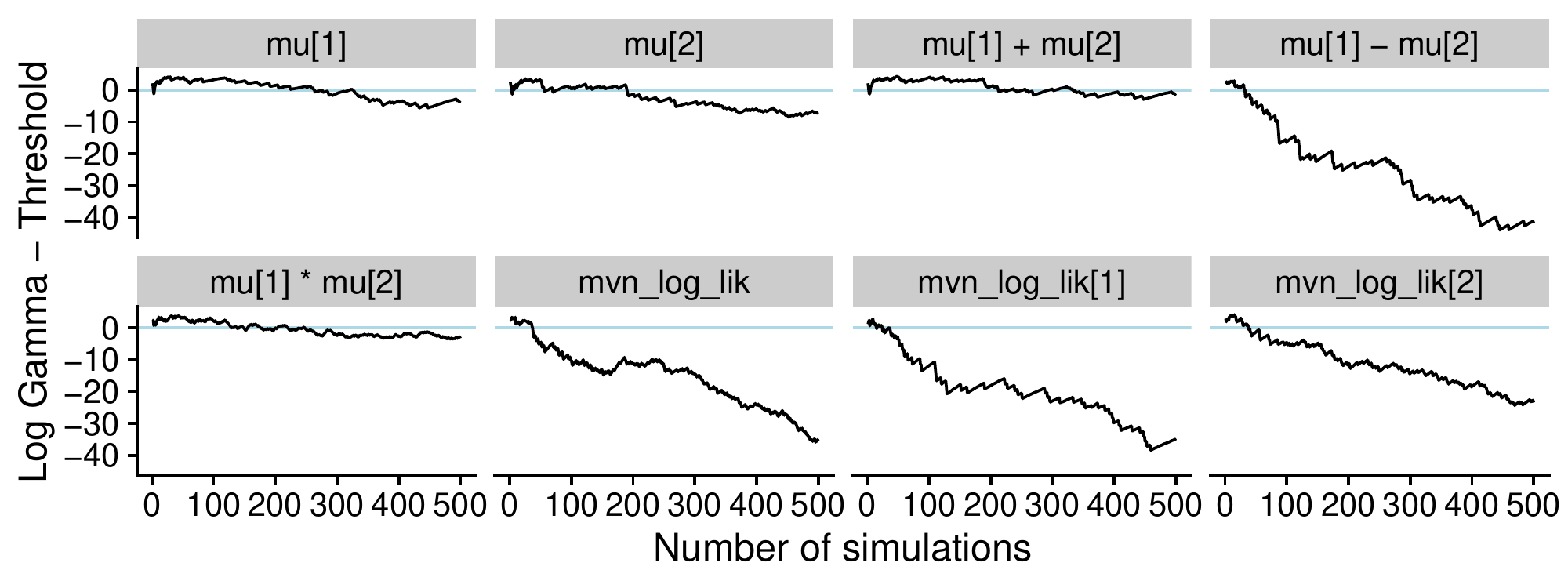}
    \vspace{-.3in}
    \caption{\em Case study 7: Evolution of the difference between the gamma statistic and threshold for rejecting uniformity at 5\% for incorrect posterior that introduces a small bias for each parameter.}
    \label{fig:hist_small_change}
\end{figure} 

\section{Real-world case study}
\label{sec:real_world}

We present a case study adapted from an actual user discussion on forums of the Stan probabilistic programming language. Our goal is to use Stan and its Hamiltonian Monte Carlo implementation to sample from a distribution over an ordered $K$-dimensional simplex which is then to be used as a component in a larger model:\footnote{The discussion can be found at \url{https://discourse.mc-stan.org/t/ordered-simplex-constraint-transform/24102}. We thank Sean Pinkney, Bob Carpenter and Ben Goodrich for contributing to the discussion and suggesting solutions.}

\begin{equation*}
 \text{OrdSimplex}_K = \{\mathbf{x} \in \mathbb{R}^K | 0 < x_1 < \ldots < x_K < 1, \sum_{i=1}^K x_i = 1 \}.
\end{equation*}

To do that, we need to construct an ordered simplex from primitive data types available in Stan and compute the logarithm of the Jacobian determinant of the transformation (up to a constant). 

\subsection{Proposed implementations}

We consider three variants. Mimicking the fallibility of methods proposed by real statisticians, not all of the following derivations will be correct. A reader interested in little mathematical puzzles may try to pin down any errors. In the next subsection, we will then show how to use SBC to discover the error(s) without the painstaking attention to detail required for checking the math. We then also remedy the error(s).

The first variant will be called \texttt{min}. 
Here, we start with an unordered bounded vector $\mathbf{u} \in [0,1]^{K - 1}$ (which is a primitive in Stan). The minimal element of the simplex needs to satisfy $x_1 < \frac{1}{K}$, so we set $x_1 = \frac{u_1}{K}$. Given $x_1$, if we set $\mathbf{x}^\prime \in \mathbb{R}^{K-1}, x^\prime_i = \frac{x_{i + 1} - x_1}{1 - Kx_1} = \frac{x_{i + 1} - x_1}{1 - u_1}$, then $\mathbf{x}^\prime \in \text{OrdSimplex}_{K - 1}$, giving us a recursive formula for the transformation, which we can unroll as:
\begin{align*}
b_1 &= 0, \ r_1 = 1 \\
\text{for } 1 \leq i < K &: \ x_i =  b_i + r_i \frac{u_i}{K + 1 - i}, \ b_{i + 1} = x_i, \ r_{i + 1} = r_{i}( 1 - u_i) \\
x_k &= b_k + r_k = 1 - \sum_{i=1}^{K - 1}x_i.
\end{align*}
Here $r_i$ can be understood as tracking the remaining amount to be distributed to ensure $x_i$ sum to $1$ if all the following elements will be at least $b_i$. For $1 \leq i < K$, we have $\frac{\partial x_i}{\partial u_i} = \frac{r_i}{K + 1 - i}$, and when also $i \leq j < K$ then $\frac{\partial x_i}{\partial u_j} = 0$, so the Jacobian matrix is triangular and the Jacobian determinant is thus 
\begin{equation}
\det \mathbf{J} = \prod_{i=1}^{K-1} \frac{r_i}{K + 1 - i}.    
\end{equation}
The second variant, called \texttt{softmax} starts with a positive ordered vector $\mathbf{v} \in (0, +\infty)^{K-1}$, $v_1 < \ldots < v_{K-1}$ (also a primitive in Stan). We then prepend $0$ to the vector and normalize it with the softmax function:\footnote{One could also base the normalization on the arithmetic sum of the elements, but this results in problematic geometry of the posterior and the sampler has trouble converging.}
\begin{equation*}
s = 1 + \sum_{i=1}^{K-1} \exp (v_i), \ x_1 = \frac{1}{s}, \ x_k = \frac{\exp (v_{k - 1})}{s}.
\end{equation*}
For $k > 1, 1 \leq j \leq K - 1, j \neq k - 1$ the partial derivatives are:
\begin{align*}
\frac{\partial x_k}{\partial v_{k - 1}} &= \frac{\exp (v_{k-1})}{s} - \frac{\exp(2v_{k-1})}{s^2} = \frac{\exp (v_{k-1})(s - \exp v_{k-1})}{s^2} \\
\frac{\partial x_k}{\partial v_{j}} &= -\frac{\exp (v_{k -1 } + v_j)}{s^2}.
\end{align*}
We notice the repeated elements and define a $K-1$ dimensional diagonal matrix $\mathbf{D}$, where $\mathbf{D}_{i, i} = \frac{\exp (y_i)}{s^2}$. We can now express the Jacobian matrix as
$$
\mathbf{J} = \left(
\mathbf D
\begin{pmatrix}
- \exp (v_1) & \dots &  - \exp (v_{K-1}) \\
\vdots & \ddots & \vdots \\
- \exp (v_1) & \dots &  - \exp (v_{K-1}) \\
\end{pmatrix} + s\mathbf{I}_{K-1} \right).
$$
We now define a $K-1$ dimensional column vector $\mathbf{c}, c_k = -\exp (v_k)$ and a row vector $\mathbf{r}, r_k = 1$ and obtain $\mathbf{J} = \mathbf{D}\left(\mathbf{cr} + s\mathbf{I}_{K-1}\right)$.
By the matrix determinant lemma, $\det(\mathbf{cr} + \mathbf{X}) =\det(\mathbf{X})( 1 + \mathbf{r}\mathbf{X}^{-1}\mathbf{c})$, for any invertible matrix $\mathbf{X}$. Since $\mathbf{rc} = \sum_{i=1}^{K - 1}(-\exp v_i) = 1 - s$, we have:
$$
\det( \mathbf{cr} + s\mathbf{I}_{K-1}) = \left(1 + \frac{1}{s}\mathbf{rc} \right) s^{K-1} = \left(1 + \frac{1 - s}{s} \right) s^{K-1} = s^{K - 2}. \\
$$
Since $\det (\mathbf{D}) = \frac{\exp (\sum_{i=1}^{K - 1} y_i)}{s^{2(K - 1)}}$, we finally have
\begin{equation}
\det (\mathbf{J}) = \det (\mathbf{D}) \det (\mathbf{cr} + s\mathbf{I}_{K-1}) = \frac{\exp (\sum_{i=1}^{K - 1} v_i)}{s^{K - 1}}.
\label{eq:jacobian_bad}
\end{equation}

As a different approach, if we are willing to restrict our priors over the ordered simplex to Dirichlet distributions, we may employ the fact that if $\mathbf{w} \in (0, +\infty)^K, w_i \sim \Gamma(\alpha_i, 1)$ then $\frac{\mathbf{w}}{\sum_{i=1}^K w_i }\sim \text{Dirichlet}(\mathbf{\alpha})$. So if we start with $\mathbf{w}$ positive ordered (a primitive in Stan), then $\mathbf{x} = \frac{\mathbf{w}}{\sum_{i=1}^K w_i }$ will be Dirichlet distributed over $\text{OrdSimplex}_K$ and no Jacobian adjustment is required. A downside of this approach is that the mapping is many-to-one and in models where $\mathbf{x}$ is tightly constrained by data, the implied geometry on $\mathbf{w}$ will likely pose difficulty for most samplers. This variant will be referred to as \texttt{gamma}.

At this point the interested reader is welcome to try to find issues with any of the above approaches. 

\subsection{Testing with SBC}
Whether the reader managed to find errors or not, we can use SBC to test all approaches. To run SBC we embed the ordered simplex into a simple model:
\begin{align}
\mathbf{x} &\in \text{OrdSimplex}_4, \pi(\mathbf{x}) \propto \text{Dirichlet(2, 2, 2, 2)} \notag \\
\mathbf{y} &\sim \text{Multinomial(10, x)}. \label{eq:ord_simplex_model}
\end{align}
Implementing the simulator code is straightforward:  due to symmetry, we can sample $\mathbf{x}$ simply by ordering a sample from the unordered Dirichlet distribution. Both \texttt{min} and \texttt{gamma} variant show no problems in SBC and are indeed correct, but \texttt{softmax} exhibits issues. Figure~\ref{fig:hist_real_world} shows the evolution of the discrepancies. The problems are most quickly picked up by the first element of $\mathbf{x}$ and the log Dirichlet prior density. Although the problem is found relatively quickly with SBC, the bias in the inferences would likely not be noticed in an informal assessment of the model:  the results are not completely wrong, just somewhat biased. The source of the issue is an off-by-one error in the exponent for $s$ in equation \eqref{eq:jacobian_bad}; the correct Jacobian determinant is
$$
\det (\mathbf{J}) = \det (\mathbf{D}) \det (\mathbf{cr} + s\mathbf{I}_{K-1}) =   \frac{\exp \sum_{i=1}^{K - 1} v_i}{s^{K}}.
$$
Indeed, if we correct the Jacobian, SBC passes. 

\begin{figure}
    \centering
    \includegraphics[width=\textwidth]{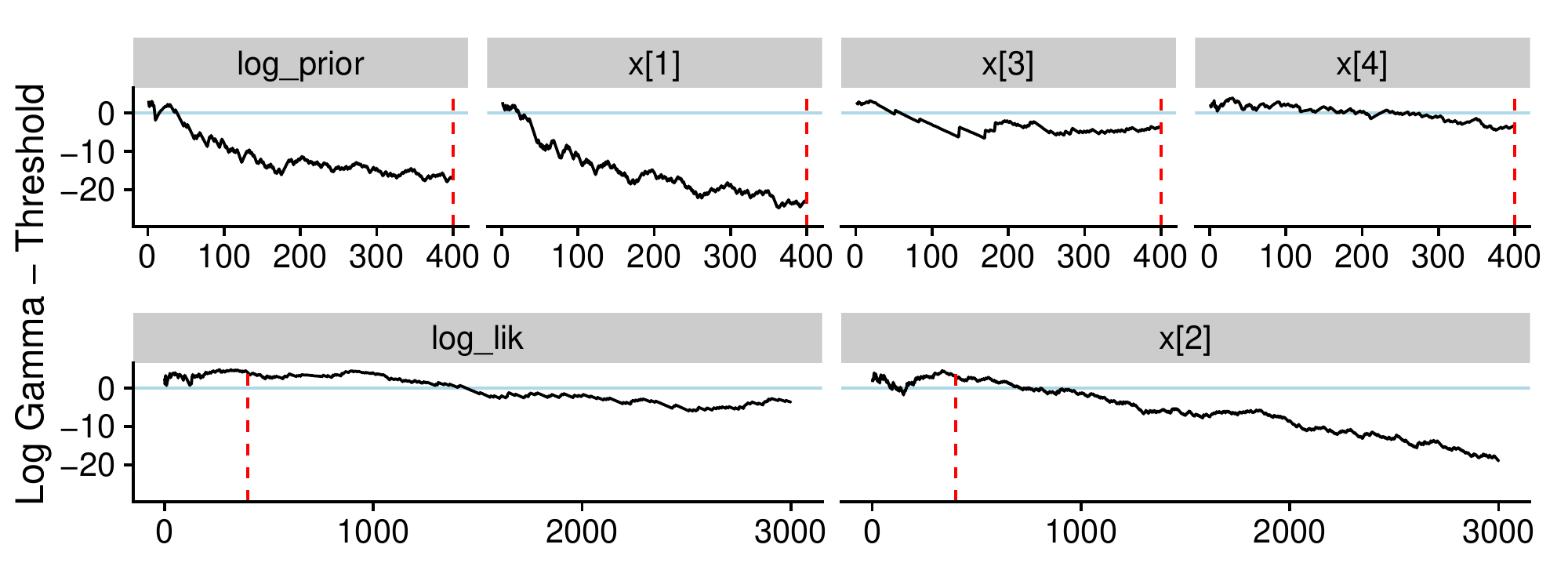}
    \vspace{-.3in}
    \caption{\em Evolution of the difference between the gamma statistic and threshold for rejecting uniformity at 5\% for the incorrectly implemented \texttt{softmax} variant of an ordered simplex model. \texttt{log\_lik} is the multinomial log likelihood of the data and \texttt{log\_prior} is the log density of the prior Dirichlet distribution. Note the different horizontal axis between top row (quantities that detect the problem quickly) and bottom row (quantities that detect the
problem slowly). The vertical red dashed line marks 400 simulations.}
    \label{fig:hist_real_world}
\end{figure}

\subsection{Remarks}

The previous section showed the type of modeling problem where SBC is in our view the most useful: deriving and implementing the probabilistic program is relatively involved and offers plenty of opportunities for error, but building a simulator is straightforward. Jacobian adjustments for changes of variables are also in our experience one of the most confusing concepts to Stan users and SBC offers a good way to check if one's reasoning is correct. 

The examples in this section introduce several non-obvious conceptual questions: For \texttt{min} and \texttt{softmax} we compute the Jacobian only considering $K - 1$ elements of the ordered simplex, even when the Dirichlet prior then acts on all elements.
Is that correct? For \texttt{gamma}, will ordering $\mathbf{w}$ imply the correct ordered simplex distribution? Running SBC is then a useful (although not completely definitive) check that our reasoning is correct.

In this example, the log likelihood did not reveal the error quickly, showing that it is not a panacea, especially in cases where the problem lies with the prior. The log prior density seems potentially useful in this case as it shows the problem as quickly as the most problematic individual parameter.
Another lesson is that SBC is useful not only for testing a full model but also for testing components of a model in isolation, akin to unit tests in software engineering. Additionally, by running SBC, we get a simulation study for free: in the specific setup described by Equation~\ref{eq:ord_simplex_model}, \texttt{min} is the most efficient in terms of effective sample size per second, followed by \texttt{gamma}. The correct version of \texttt{softmax} performs worst. The \texttt{softmax} variant also fails to converge in 6 of the 1000 simulations, while the other two are slightly more stable (convergence problems in 3 and 1 of the simulations respectively). Finally, our posterior uncertainty is large and the data do not really provide a lot of information about the parameter values. See the rendered output of the supplementary code for details.

\section{Conclusions}
\label{sec:conclusions}

\subsection{Choosing test quantities for SBC}
We have found that enriching the repertoire of test quantities used in SBC provides both qualitative and quantitative improvements to the ability of SBC to detect problems in Bayesian computation. For practical use of SBC in everyday model and algorithm development, we recommend to use by default the individual model parameters as test quantities as well as the joint likelihood of the data and potentially a small number of other quantities.

Individual parameters are recommended as they are always immediately available and are able to diagnose a large number of problems with a posterior approximation. Also, the parameters are themselves often of primary interest for inference, so it is desirable to check that their uncertainty is correctly calibrated. 

The joint likelihood is a highly useful quantity to detect the types of problems discussed in Section~\ref{sec:numerical} (especially ignoring data and incorrect correlations). In all of the cases presented in our simulations, the joint likelihood was able to detect the discrepancies and in many cases it was even able to detect them with the fewest simulations among all considered quantities. While, for some specific problems, we could find quantities that are more sensitive than the joint likelihood, none other was useful in all cases. Section~\ref{sec:data_data_dependent_quantities} provides theoretical justification for why we could expect this to hold frequently and not only in the examples we discussed. We think this generality makes the joint likelihood a good default quantity to monitor in SBC. If not using all the data correctly is a potential issue (e.g., because the code handling the data is particularly complex), then adding selected likelihoods for subsets of the data might also be sensible. 

As shown in Section~\ref{sec:real_world}, knowing where a potential problem lies can let us design more sensitive problem-specific checks (e.g., when we are not sure our prior density is correct, the log prior can be highly useful).
It also makes sense to add test quantities tailored to the specific inferential goals we have built the model for (e.g., some specific model predictions). These quantities often let us implicitly check the correctness of parameter correlations or other dependency structures and safeguard the user against problems that they care about the most. 
If correlations or other dependencies in the posterior are directly of interest, then pairwise products or differences of the model parameters can also be sensible test quantities. 

\subsection{Limitations}

Although we have shown that SBC can in principle diagnose any problem, limitations for practical use remain. For nontrivial models, adding a finite number of test quantities cannot guard against all possible ways the SBC identity may be satisfied by an incorrect posterior. However, as we check more quantities, the potential counterexamples become contrived, hard to construct, and unlikely to be the result of an inadvertent bug in model or algorithm code. At the same time, adding more test quantities increases the risk of false SBC failures simply due to the number of tests performed (if no corrections for multiple comparisons are made for the SBC checks) or it may reduce the overall power of the check (if corrections for multiple comparisons are made), so choosing test quantities carefully remains important.

This problem could potentially be alleviated by improving our understanding of the expected dependency structure of different test quantities' uniformity checks, letting us correct for multiple comparisons without loosing that much power. However, even similar test quantities can lead to in principle different SBC checks (see Section~\ref{sec:monotonic}). So any practical measure of dependency or orthogonality between test quantities would need to reflect not only existence of a difference, but also its magnitude. We leave that as future work. In practice, we have seen similarity in the degree of uniformity violation between different test quantities using the same inputs, making the need for multiple comparison correction less urgent.

Moreover, there are practical limitations imposed by the fact that we always have only limited computational resources for SBC: We can produce only a limited number of simulated datasets to fit the model on and only a limited number of posterior draws per fitted model. Both contribute to the stringency and precision of the uniformity test we can perform. The difference between continuous SBC and any practical implementation of sample SBC arises due to (a) approximating $q_{\phi,f}(x | y)$ by $Q_{\phi,f}(\lfloor xM \rfloor \,|\,y)$, and (b) using finite number of simulations to assess uniformity of $N_{\mathtt{total}}$. In both cases, the underlying difference can be understood as estimating a CDF by an empirical CDF and should therefore have similar rate of decrease with more draws. This suggests that for a given computational budget a user is likely to obtain the highest sensitivity using the same order of magnitude of simulated datasets as posterior draws per dataset. However, in practice most algorithms incur a substantial cost in a warmup phase, before any samples can be extracted. We also want to assess that our fitting algorithm has converged for each dataset, which typically requires the equivalent of at least 100 independent posterior samples (as measured by effective sample size) to do that (e.g., to get a low $\hat R$ statistic, as discussed by \citealt{improvedrhat}). It is thus hard to get a speedup by reducing the number of posterior draws. Unless we can afford to run many thousands of simulations, we are also unlikely to benefit substantially from getting more than this minimal number of draws.

Additional test quantities do not help much with precision problems---if the posterior is close to correct, the test quantities will also be close to correct. Although in some cases, some test quantities can slightly increase the sensitivity of the check by combining multiple parameters, so small imprecisions in each of the parameters can get compounded (once again the nonlinearity of the likelihood seems to be at least sometimes useful in this regard). 

\subsection{Implications for non-SBC checks}

As a contribution to the broader discussion about validation of Bayesian computation, we show that SBC and the data-averaged posterior provide different checks, despite being repeatedly conflated in the literature (see Section \ref{sec:examples_summary}). We leave a more detailed comparison of SBC and data-averaged posterior as future work, although there are some tentative arguments to believe that SBC provides stricter checks. 

SBC is not the only approach to validating Bayesian computation that relies on choosing specific test quantities---test quantities are fundamental to the methods of \citet{geweke_getting_2004}, \citet{prangle_abc}, \citet{gandy_unit_2021}, and \citet{cockayne_testing_2022}. We suspect that many of the considerations regarding their choice for SBC are applicable also in these other approaches.

\bibliographystyle{ba}
\bibliography{main.bib}

\setcounter{theorem}{0}
\setcounter{definition}{0}

\section*{Supplementary Material}
Appendix A contains mathematical theory and proofs, and Appendix B contains examples of simple models where we can fully characterize the space of posteriors that satisfy simulation-based calibration checking with respect to several test quantities. Code for simulations and all the figures in sections \ref{sec:numerical} and \ref{sec:real_world} can be found at \url{https://github.com/martinmodrak/sbc_test_quantities_paper}. All code output and associated commentary can also be viewed at  \url{https://martinmodrak.github.io/sbc_test_quantities_paper/}

\section*{Appendix A: Formalized theory and proofs}

We will denote the integral of $f(x)$ w.r.t.\ $x$ over domain $X$ as $\int_X \mathrm{d}x \: f(x)$. $\mathbb{I}[P]$ is the indicator function for a given predicate $P$. When a function can be understood as describing a conditional probability distribution, we will use $|$ to separate the function arguments we condition on. This is only to assist comprehension and has the same semantic meaning as using a comma.

In all cases, we assume an underlying statistical model $\pi$ which decomposes into a prior and observational model. Given a data space $Y$ and a parameter space $\Theta$, then for $y \in Y, \theta \in \Theta$ the model implies the following joint, marginal and posterior distributions:
\begin{gather*}
    \pi_\text{joint}(y, \theta) = \pi_\text{obs}(y | \theta) \pi_\text{prior}(\theta)\\
    \pi_\text{marg}\left(y \right) = \int_\Theta \mathrm{d} \theta \: \pi_{\text{obs}}(y | \theta) \pi_\text{prior}(\theta)\\
    \pi_\text{post}(\theta | y) = \frac{\pi_\text{obs}(y | \theta) \pi_\text{prior}(\theta)}{\pi_\text{marg}\left(y \right)}.
\end{gather*}
Unless noted otherwise, all definitions and proofs implicitly assume a single model $\pi$ is given.

\begin{definition}[Posterior family] Given a data space $Y$ and a parameter space $\Theta$, a \emph{posterior family} $\phi$ assigns a normalized posterior density to each possible $y \in Y$. I.e. posterior family is a function $\phi : \Theta \times Y  \rightarrow  \mathbb{R^{+}}$ such that
\begin{equation*}
\forall y: \int \mathrm{d}\theta \:\phi(\theta | y) = 1.
\end{equation*}
For each $y$, we will denote the implied distribution over $\Theta$ as $\phi_y$. 

\end{definition}

\begin{definition}[Test quantity]
     Given a data space $Y$ and a parameter space $\Theta$ a \emph{test quantity} is any measurable function $f :\Theta \times Y  \to  \mathbb{R} $.
\end{definition}

\begin{definition}[Sample rank CDF, sample Q, sample SBC]
\label{def:sample_SBC}
Given a data space $Y$, a parameter space $\Theta$, a test quantity $f$, $M \in \mathbf{N}$ and a posterior family $\phi$. For any $y \in Y$, if $\theta_1, \dots, \theta_M \sim \phi_y$ we can define the following random variables:
\begin{align*}
   N_{\phi,f,\tilde\theta,y}^{\mathtt{less}} &:= \sum_{m=1}^M \mathbb{I} \left[f(\theta_m, y) < f(\tilde \theta, y) \right] \\
   N_{\phi,f,\tilde\theta,y}^{\mathtt{equals}} &:= \sum_{m=1}^M \mathbb{I} \left[f(\theta_m, y) = f(\tilde \theta, y) \right] \\
    K_{\phi,f,\tilde\theta,y} &\sim \mathrm{uniform}\left(0,  N_{\phi,f,\tilde\theta,y}^{\mathtt{equals}}\right) \\
    N_{\phi,f,\tilde\theta,y}^\mathtt{total} &:= N_{\phi,f,\tilde\theta,y}^{\mathtt{less}} + K_{\phi,f,\tilde\theta,y}.
\end{align*}    
For a fixed $\tilde\theta$, we define the \emph{$M$-sample rank CDF} as:
\begin{equation*}
    R_{\phi,f}(i | \tilde\theta, y) := \mbox{Pr} \left( N_{\phi,f,\tilde\theta,y}^\mathtt{total} \leq i \right).
\end{equation*}
Averaging over the correct posterior, we can define the \emph{$M$-sample $Q$} as:
\begin{equation*}
  Q_{\phi,f}(i | y) := \int_\Theta \mathrm{d}\tilde\theta \: \pi_\text{post}(\tilde\theta|y)  R_{\phi,f}(i| \tilde\theta, y)
\end{equation*}
We then say that \emph{$\phi$ passes} $M$-sample SBC w.r.t.\ $f$ if $\forall i \in 0, \dots, M - 1$ we have
\begin{equation*}
\int_Y \: \mathrm{d}y \: Q_{\phi,f}(i |y) \pi_\text{marg}(y) = \frac{i+1}{M + 1}.    
\end{equation*}
\end{definition}

This definition does not match immediately with the procedure we actually use to run SBC in practice but is more convenient for further analysis. We now prove the equivalence with the SBC procedure:

\begin{theorem}[Procedural definition of sample SBC]
\label{ath:procedural_sbc}
Given a data space $Y$, a parameter space $\Theta$, a test quantity $f$, $M \in \mathbf{N}$ and a posterior family $\phi$. $\phi$ passes $M$-sample SBC w.r.t.\ $f$ if and only if given $\randvar{\tilde\theta} \sim \pi(\tilde\theta), \randvar{y} \sim \pi_\text{obs}(y | \randvar{\tilde\theta}), N^\mathtt{total} = N_{\phi,f,\randvar{\tilde\theta},\randvar{y}}^\mathtt{total}$ we have $N^\mathtt{total} \sim \mathrm{uniform}(0, M)$.    
\end{theorem}

\begin{proof}
We show the equivalence of CDF for $N_\mathtt{total}$ with the formula in the definition of sample SBC. For all $i \in \{0, \dots, M\}$ we obtain
\begin{gather*}
\mbox{Pr}(N^\mathtt{total} \leq i) =  \int_\Theta \mathrm{d}\tilde\theta   \int_Y \: \mathrm{d}y \: \pi_\text{prior}(\tilde\theta)\pi_\text{obs}(y | \tilde\theta)  R_{\phi,f}(i | \tilde\theta, y) \\ =
\int_Y \: \mathrm{d}y  \int_\Theta \mathrm{d}\tilde\theta \: \pi_\text{post}(\tilde\theta|y)  \pi_\text{marg}(y) R_{\phi,f}(i| \tilde\theta, y)  = \int_Y \: \mathrm{d}y \: Q_{\phi,f}(i |y) \pi_\text{marg}(y)
\end{gather*} 
\end{proof}

However, the sample SBC is not particularly amenable to direct analysis as the choice of $M$ can matter. We will thus focus on a continuous case, which can be understood as the limit of the sample SBC as $M \to \infty$.

\begin{definition}[CDF, tie probability, quantile functions]
Given a model $\pi$, data space $Y$, a parameter space $\Theta$, a posterior family $\phi$, and a test quantity $f$, we define
\begin{itemize}
    \item fitted CDF: $C^\pi_{\phi,f}: \bar{\mathbb{R}}\times Y\to [0,1]$, \\ $C^\pi_{\phi,f}(s | y) := \int_{\Theta}\mathrm{d}\theta\,\mathbb{I}\left[f\left(\theta, y \right) \leq s\right]\phi\left(\theta| y\right)$
    \item true CDF: $C^\pi_{f}: \bar{\mathbb{R}} \times Y\to [0,1]$, \\
    $C^\pi_{f}\left( s | y\right) :=\int_{\Theta}\mathrm{d}\theta\,\mathbb{I}\left[f\left(\theta, y \right) \leq s\right]\pi_{\text{post}}(\theta | y)$
    \item    fitted quantile function: 
    $C^{\pi, -1}_{\phi,f} : [0,1] \times Y \to \bar{\mathbb{R}}$, \\ 
    $C^{\pi, -1}_{\phi,f}(x | y) := \inf\{s : x \leq C_{\phi, f}(s | y)\}$
    \item true quantile function: 
    $C^{\pi, -1}_{f} : [0,1] \times Y \to \bar{\mathbb{R}}$, \\
    $C^{\pi, -1}_{f}(x | y) := \inf\{s : x \leq C_{f}(s | y)\}$
    \item fitted tie probability: $D^\pi_{\phi,f}: \bar{\mathbb{R}}\times Y\to [0,1]$, \\
    $D^\pi_{\phi,f}(s | y) := \int_\Theta \mathrm{d}\theta \: \phi(\theta | y) \mathbb{I}\left[ f(\theta, y) = s \right]$ 
    \item true tie probability: $D^\pi_{f}: \bar{\mathbb{R}}\times Y\to [0,1]$, \\
    $D^\pi_{f}(s | y) := \int_\Theta \mathrm{d}\theta \: \pi_\text{post}(\theta | y) \mathbb{I}\left[ f(\theta, y) = s \right]$,
\end{itemize}
where $\bar{\mathbb{R}} := \mathbb{R} \cup \{-\infty, +\infty\}$ is the extended real number line. When no confusion arises, the model superscript will be omitted.
\end{definition}

\begin{definition}[Continuous rank CDF, continuous $q$, continuous SBC]
Assume a data space $Y$, a parameter space $\Theta$, a test quantity $f$, and a posterior family $\phi$. For fixed $\tilde\theta \in \Theta$ and $y \in Y$ we define the \emph{continuous rank CDF} $r_{\phi,f} : [0,1] \times \Theta \times Y \to [0, 1]$ as
\begin{equation*}
r_{\phi,f}(x | \tilde\theta, y) := 
\mbox{Pr}\left(C_{\phi,f} \left(f \left(\left.\tilde\theta, y \right) \right| y \right) - U D_{\phi,f} \left(f \left(\left.\tilde\theta, y \right) \right| y \right) \leq x\right) ,   
\end{equation*}
assuming $U$ is a random variable distributed uniformly over the $[0, 1]$ interval. If the fitted tie probability is $0$, then $r$ is a step function and the implied rank distribution is degenerate.

Averaging over the correct posterior, we can define the \emph{continuous } $q : [0,1] \times Y \to [0,1]$ as
\begin{equation*}
    q_{\phi,f}(x|y) := 
\int_\Theta \mathrm{d} \tilde\theta \: \pi_\text{post}(\tilde\theta | y) r_{\phi,f}(x | \tilde\theta, y).
\end{equation*}
We then say that \emph{$\phi$ passes} continuous SBC w.r.t.\ $f$ if
\begin{equation}
\forall x \in [0, 1]: \int_Y \mathrm{d}y \: q_{\phi,f}(x|y) \pi_\text{marg}(y)  = x. \label{eq:sbc_cdf}
\end{equation}
\end{definition}
%
%
%
In both the sample and the continuous case, the presence of ties introduces analytical difficulties. We thus start by a useful lemma that lets us avoid dealing with ties for many purposes.

\begin{definition}[Ties]
Given a test quantity $f: \Theta \times Y \to \mathbb{R}$
a model $\pi$ \emph{has ties} w.r.t.\ $f$ when there exists $y \in Y, \tilde\theta \in \Theta$ such that
\begin{equation*}
D_{f}(f(\tilde\theta,y) | y) > 0.
\end{equation*}
A posterior family $\phi$ \emph{has ties} w.r.t.\ $f$ when there exists $y \in Y, \tilde\theta \in \Theta$ such that
\begin{equation*}
D_{\phi,f}(f(\tilde\theta,y) | y) > 0.
\end{equation*}
\end{definition}

\begin{lemma}[Tie removal]
\label{le:tie_removal}
    Given a model $\pi$, posterior family $\phi$, and test quantity $f: \Theta \times Y \to \mathbb{R}$, where either $\pi$ has prior ties w.r.t.\ $f$ or $\phi$ has ties w.r.t.\ $f$, we can construct a model $\pi^\prime$, with parameter space $\Theta^\prime$ and data space $Y$, posterior family $\phi^\prime$ and test quantity $f^\prime$ such that (1) $\pi^\prime$ has no  ties w.r.t.\ $f^\prime$ and $\phi^\prime$ has no ties w.r.t.\  $f^\prime$, (2) $\forall \theta \in \Theta, y \in Y: \phi(\theta | y) = \pi_\text{post}(\theta | y)$ if and only if $\forall \theta \in \Theta^\prime, y \in Y:  \phi^\prime(\theta, y) = \pi^\prime_\text{post}(\theta| y)$ and (3) both sample $Q$ and continuous $q$ are maintained. That is,
\begin{gather*}
\forall y \in Y, x \in [0, 1]: q^\pi_{\phi,f}(x | y) = q^{\pi^\prime}_{\phi,f^\prime}(x | y) \\
\forall y \in Y, i \in 0, \dots, M - 1: Q^\pi_{\phi,f}(i | y) = Q^{\pi^\prime}_{\phi,f^\prime}(i | y).    
\end{gather*}
\end{lemma}
\begin{proof}
The overall idea is that we can always smooth the implied ties in the distributions of the test quantity by introducing a gap at each possibly tied point and adding a uniformly random value.  Specifically, for each $y \in Y, s \in \mathbb{R}$ we define $\hat{D}_{\phi,f}(v | y) \in \mathbb{R}$, the set $V_{\phi,f}(y) \subset \mathbb{R}$ and number $w_{\phi,f}(s, y) \in [0,1]$ :
%
\begin{align*}
\hat{D}_{\phi,f}(v | y) &:= D^\pi_{\phi,f}(v | y) + D^\pi_{f}(v | y) \\
V_{\phi,f}(y) &:= \{v : \hat{D}_{\phi,f}(v | y) > 0 \} \\
w_{\phi,f}(s, y) &:= \sum_{v \in V_{\phi,f}(y), v < s} \hat{D}_{\phi,f}(v | y).
\end{align*}
For each $y$ there can be at most a countable number of point masses in the implied distributions and the total mass is at most $1$, so $w_{\phi,f}(s, y)$ will always be defined and smaller than $2$.

We now construct the new model, posterior, and test quantity by adding a new parameter $p$ uniformly distributed over the $[0,1]$ interval, so $\Theta^\prime := \Theta \times [0, 1]$ and for all $y \in Y, p \in [0,1]$ we set:
\begin{align*}
\theta^\prime &:= (\theta, p) \\
\pi^\prime_\text{joint}(y, \theta^\prime) &:=  \pi_\text{joint}(y, \theta) \\
\phi^\prime(\theta^\prime | y) &:= \phi(\theta | y) \\
f^\prime(\theta^\prime, y) &:= \begin{cases}
f(\theta,y) + w_{\phi,f}(f(\theta,y),y) + p \hat{D}_{\phi,f}(f(\theta,y) | y)  & \text{if } f(\theta,y) \in V_{\phi,f}(y) \\
f(\theta,y) + w_{\phi,f}(f(\theta,y),y) & \text{otherwise}.
\end{cases}    
\end{align*}
Here, $w$ provides  gaps as it increases by $\hat{D}_{\phi,f}(f(\theta,y)) > 0$ for each tied value. Those gaps are then filled uniformly randomly by adding scaled  $p$.  By inspection, there are no ties in $f^\prime$. The ordering of previously non-tied elements for both continuous and sample ranks has not changed and the order among those previously tied is uniformly random so both sample $Q$ and continuous $q$ do not change. Since $\forall y \in Y, \theta \in \Theta, p \in [0,1]: \pi^\prime_\text{post}(\theta^\prime | y) = \pi_\text{post}(\theta | y)$, the construction also directly ensures that $\phi = \pi_\text{post}$ if and only if $\phi^\prime = \pi^\prime_\text{post}$.
\end{proof}

\begin{lemma}[No ties and inverting CDFs]
\label{le:no_ties_invert}
For all models $\pi$, posterior families $\phi$, and test quantities $f: \Theta \times Y$:
\begin{enumerate}
    \item If model $\pi$ has no ties w.r.t.\ $f$ then  for all $y$ the quantile function $C^{-1}_f(s|y)$ is the inverse of $C_f(s|y)$. Additionally, $C_f(s|y)$  is a surjection on the $[0, 1]$ interval. 
    \item If $\phi$ has no ties w.r.t.\ $f$ then for all $y$ the quantile function $C^{-1}_{\phi, f}(s|y)$ is the inverse of $C_{\phi, f}(s|y)$. Additionally, $C_{\phi, f}(s|y)$  is a surjection on the $[0, 1]$ interval. 
    \item If both $\pi$ and $\phi$  have no ties w.r.t.\ $f$ then 
    $q_{\phi,f}(x|y) = C_f \left(\left.C^{-1}_{\phi,f}(x | y) \right| y \right)$  and $q_{\phi,f}(x|y)$ is a bijection from $[0, 1]$ to $[0, 1]$.
\end{enumerate}
\end{lemma}

\begin{proof}
    Items 1 and 2 follow directly from the basic properties of CDFs. For item 3, when there are no ties, the definition of continuous rank simplifies to $r_{\phi,f}(x | \tilde\theta, y) = \mathbb{I}\left[x \geq C_{\phi,f} \left(\left. f \left(\tilde\theta, y \right) \right| y \right)  \right]$ and thus:
\begin{gather*}
q_{\phi,f}(x|y) = 
\int \mathrm{d} \tilde\theta \: \pi_\text{post}(\tilde\theta | y) \mathbb{I}\left[x \geq C_{\phi,f} \left(\left.f \left(\tilde\theta, y \right) \right| y \right)  \right] \\
= \int \mathrm{d} \tilde\theta \: \pi_\text{post}(\tilde\theta | y) \mathbb{I}\left[f \left(\tilde\theta, y \right) \leq C^{-1}_{\phi,f}(x | y)\right] =  C_f \left(\left.C^{-1}_{\phi,f}(x | y) \right| y \right),
\end{gather*}
where the second equality holds because $C^{-1}_{\phi,f}$ is the inverse of $C_{\phi,f}$ for all $y$.
\end{proof}

We now move to establishing a close correspondence between the sample SBC and continuous SBC.

\begin{theorem}[Sample SBC implies continuous SBC] 
\label{ath:sample_implies_continous}
Let $\phi$ be any posterior family over $Y$ and $\Theta$, and let $f$ be a test quantity:

\begin{enumerate}
    \item For any fixed $y \in Y$ if as the number of sample draws $M \to \infty$ we have $\forall i \in \{0, \dots, M \}: Q_{\phi,f}(i | y) \to \frac{i + 1}{M + 1}$ then $\forall x \in [0, 1]: q_{\phi,f}(x | y) = x$.
    \item If as $M \to \infty$ we have $\forall i \in \{0, \dots, M \}: \int_Y \mathrm{d}y \, Q_{\phi,f}(i | y) \pi_\text{marg}(y) \to \frac{i + 1}{M + 1}$ then $\phi$ passes continuous SBC for $f$.
\end{enumerate}

\end{theorem}

\begin{proof}

Using Lemma~\ref{le:tie_removal} it once again suffices to prove the case with no ties. We start by showing that as $M \to \infty$ normalized sample ranks converge almost everywhere to the continuous ranks.

Assuming no ties, $R_{\phi,f}(\lfloor xM \rfloor \,|\, \tilde\theta, y) = \text{Bin}(\lfloor xM \rfloor \,|\, M, C_{\phi,f}(f(\tilde\theta, y)))$ where \\ $\text{Bin}(K | M, p)$ is the CDF of a binomial distribution with $M$ trials and success probability $p$ evaluated at $K$.  Using the central limit theorem, we have  

\begin{equation*}
\lim_{M \to \infty} \left(\text{Bin}\Big(\lfloor xM \rfloor \,\Big|, M, p\Big) - \mbox{Nor}\left(\lfloor xM \rfloor \,\left|\, Mp,\sqrt{Mp(1-p)}\right.\right)\right) = 0,
\end{equation*}
where $\mbox{Nor}(y | \mu, \sigma)$ is the CDF of a normal distribution with mean $\mu$ and standard deviation $\sigma$ evaluated at $y$. We can now inspect the limiting behaviour of the $z$-score:

\begin{gather*}
\lim_{M \to \infty}\frac{\lfloor xM \rfloor - Mp}{\sqrt{Mp(1-p)}} =
\lim_{M \to \infty}\frac{\lfloor xM \rfloor -xM}{\sqrt{Mp(1-p)}} + \lim_{M \to \infty}\frac{xM- Mp}{\sqrt{Mp(1-p)}} \\
= 0 + \lim_{M \to \infty}\frac{(x - p)\sqrt{M}}{\sqrt{p(1-p)}} = 
\begin{cases}
  +\infty & \text{for } x > p \\
  0  & \text{for } x = p \\
  -\infty & \text{for } x < p.
\end{cases}
\end{gather*}
Thus,
\begin{gather*}
\lim_{M \to \infty}R_{\phi,f}(\lfloor xM \rfloor \,|\, \tilde\theta, y) = \left.\begin{cases}
  1 & \text{for } x > C_{\phi,f}(f(\tilde\theta, y)) \\
  \frac{1}{2}  & \text{for } x = C_{\phi,f}(f(\tilde\theta, y)) \\
  0 & \text{for } x < C_{\phi,f}(f(\tilde\theta, y))
\end{cases}\right. \\
= \begin{cases}
r_{\phi,f}(x | \tilde\theta, y) & \text{for } x \neq C_{\phi,f}(f(\tilde\theta, y))
\\
\frac{1}{2} & \text{for } x = C_{\phi,f}(f(\tilde\theta, y)).
\end{cases}
\end{gather*}

We have thus established pointwise convergence of $R_{\phi,f}(\lfloor xM \rfloor \,|\, \tilde\theta, y)$ to $r_{\phi,f}(x | \tilde\theta, y)$  almost everywhere w.r.t.\ $x$, because we assume there are no ties. This means that we can satisfy the conditions of the dominated convergence theorem.

We can now prove both claims in an analogous way. We start with  claim 2, which is more complex. For all $x \in [0, 1]$:
\begin{align}
\lim_{M \to \infty}\int_Y \: & \mathrm{d}y \: Q_{\phi,f}(\lfloor xM \rfloor \,|\,y) \pi_\text{marg}(y) \notag \notag \\
&= \lim_{M \to \infty}\int_Y \: \mathrm{d}y \: \int_\Theta \mathrm{d}\tilde\theta \: \pi_\text{post}(\tilde\theta|y)  R_{\phi,f}(\lfloor xM \rfloor\,|\, \tilde\theta , y) \pi_\text{marg}(y) \notag \\
&= \int_Y \: \mathrm{d}y \: \int_\Theta \mathrm{d}\tilde\theta \: \pi_\text{post}(\tilde\theta|y)  r_{\phi,f}(x| \tilde\theta | y) \pi_\text{marg}(y) \notag \\ 
&= \int_Y \mathrm{d}y \: q_{\phi,f}(x|y) \pi_\text{marg}(y).
\label{eq:sample_sbc_limit_lh}    
\end{align}
Equation (\ref{eq:sample_sbc_limit_lh}) holds regardless of the actual rank distributions. We now use the assumption that as $M \to \infty$ the data-averaged rank distribution becomes uniform, which implies $\forall x \in [0, 1]$: 
\begin{equation}
\lim_{M \to \infty}\int_Y \: \mathrm{d}y \: Q_{\phi,f}(\lfloor xM \rfloor \,|\,y) \pi_\text{marg}(y)  = 
\lim_{M \to \infty} \frac{\lfloor xM \rfloor+1}{M + 1} = x.
\label{eq:sample_sbc_limit_rh}
\end{equation}
Combining (\ref{eq:sample_sbc_limit_lh}) and (\ref{eq:sample_sbc_limit_rh}), we get that $\int_Y \mathrm{d}y \: q_{\phi,f}(x|y) \pi_\text{marg}(y) = x$ and therefore $\phi$ passes continuous SBC w.r.t.\ $f$. The reasoning for claim 1 is analogous, only omitting the integration over $Y$.
\end{proof}

\begin{theorem}[Continuous SBC implies sample SBC] 
\label{ath:continuous_implies_sample}
Let $\phi$ be any posterior family over $Y$, $\Theta$ then for all $M \in \mathbf{N}$ and any test quantity $f$:
\begin{enumerate}
    \item For any $y \in Y$, if $\forall x \in [0,1]: q_{\phi,f}(x|y) = x$ then $\forall i \in \{0, \dots, M - 1 \}: Q_{\phi,f}(i | y) = \frac{i + 1}{M + 1}$.
    \item If $\phi$ passes continuous SBC w.r.t.\ $f$, then $\phi$ passes $M$-sample SBC w.r.t.\ $f$. 
\end{enumerate}

\end{theorem}

\begin{proof}

If either $\phi$ or $\pi$ has ties w.r.t.\ $f$, we can use Lemma~\ref{le:tie_removal} to construct a model and posterior family with no ties but the same continuous $q$ and sample $Q$; therefore it suffices to prove the statement when there are no ties.

When there are no ties, then given $y$, ${q_{\phi,f}(x|y)}$ is a bijection and we can thus also define $q^{-1}$. Additionally, $\randvar{\tilde\theta} \sim \pi_\text{prior}(\tilde\theta)$ and $\theta_1, \dots \theta_M \sim \phi_y$ are all conditionally independent given $y$. 

This allows us to determine the probability of the first $r$ draws being smaller than the prior sample $\randvar{\tilde\theta} \sim \pi(\theta)$ and the remaining $M - r$ draws being larger conditional on a specific $y$:
\begin{gather*}
\hspace{-1.3in}
\mbox{Pr} \left( \left(f(\theta_1, y) < f(\randvar{\tilde\theta}, y)\right) \land \ldots \land \left(f(\theta_r, y) < f(\randvar{\tilde\theta}, y)\right) \right.\\
  \left.\left. \land \left(f(\theta_{r + 1}, y) \geq f(\randvar{\tilde\theta}, y)\right) \land \ldots \land \left(f(\theta_{M}, y) \geq f(\randvar{\tilde\theta}, y)\right) \right| y \right) \\ =
\int_{0}^{1}\mathrm{d}\tilde x \int_{\mathbf{x} \in [0, 1]^M}\mathrm{d}\mathbf{x}\: 
  \prod_{i = 1}^{r} \mathbb{I} \left[C^{-1}_{\phi,f}(x_i|y) < C_f^{-1}(\tilde x | y)\right]
  \prod_{i = r + 1}^{M} \mathbb{I}\left[C^{-1}_{\phi,f}(x_i | y) \geq C_f^{-1}(\tilde x | y)\right] \\ =
\int_{0}^{1}\mathrm{d}\tilde x \int_{\mathbf{x} \in [0, 1]^M}\mathrm{d}\mathbf{x}\: 
  \prod_{i = 1}^{r} \mathbb{I}\left[x_i < C_{\phi,f}\left( \left. C_f^{-1}(\tilde x | y) \right|  y\right)\right]
  \prod_{i = r + 1}^{M} \mathbb{I}\left[x_i \geq C_{\phi,f}\left( \left. C_f^{-1}(\tilde x|y) \right|y \right)\right] \\ 
= \int_{0}^{1}\mathrm{d}\tilde x \int_{0 < x_1, \ldots, x_k < q^{-1}_{\phi,f}(\tilde x | y)\leq x_{r+1} \ldots x_L < 1}^{}\mathrm{d}x_1\: 1 \\
 = \int_{0}^{1}\mathrm{d}x \: \left( q^{-1}_{\phi,f}(x | y)\right)^r \left(1 -  q^{-1}_{\phi,f}(x | y) \right)^{M-r}.
\end{gather*}
Since $\theta_1, \dots , \theta_L$ are independent given $y$, we can easily extend to all possible orderings and substitute $z = q_{\phi,f}^{-1}(x|y)$:
\begin{gather*}
\forall r \in {0, 1, \dots , M}: \mbox{Pr} \left(\left.\sum_{i=1}^M \mathbb{I} \left[f(\theta_i, y) < f\left(\randvar{\tilde\theta}, y\right)\right] = r \right| y \right)	\notag\\
 = \int_{0}^{1}\mathrm{d}x\,\binom{M}{r}\left(q_{\phi,f}^{-1}(x|y)\right)^{r}\left(1-q_{\phi,f}^{-1}(x|y)\right)^{M-r} \\
  =	\int_{0}^{1}\mathrm{d}z\,\binom{M}{r}z^{r}\left(1-z\right)^{M-r} \frac{\partial}{\partial z} q_{\phi,f}(z | y) \notag
\end{gather*}
From the precondition in claim 1, $q_{\phi,f}(x|y) = x$ and thus $\frac{\partial}{\partial z} q_{\phi,f}(z | y) = 1$. Then $\int_{0}^{1}\mathrm{d}z \, z^{r}\left(1-z\right)^{M-r} = B(r + 1, M - r + 1)$ is the beta integral. Thus,
\begin{gather*}
\forall r \in {0, 1, \dots , M}: \mbox{Pr} \left(\left.\sum_{i=1}^M \mathbb{I} \left[f(\theta_i, y) < f(\randvar{\tilde\theta}, y)\right] = r \right| y \right)	\\ =
    \binom{M}{r} B(r + 1, M - r + 1) = \frac{1}{M + 1}.
\end{gather*}
Therefore,
\begin{equation*}
    Q_{\phi,f}(j | y) = \sum_{r = 0}^{j} \mbox{Pr} \left(\left.\sum_{i=1}^M \mathbb{I} \left[f(\theta_i, y) < f(\randvar{\tilde\theta}, y)\right] = r \right| y \right) = \frac{j + 1}{M + 1},
\end{equation*}
which proves claim 1. For claim 2 we investigate the unconditional probability,
\begin{gather*}
\forall r \in {0, 1, \dots , M}: \mbox{Pr} \left(\sum_{i=1}^M \mathbb{I} \left[f(\theta_i, y) < f(\randvar{\tilde\theta}, y)\right] = r\right)	\notag\\
= \int_{Y}\mathrm{d}y\: \pi_\text{marg}\left(y\right) \mbox{Pr} \left(\left.\sum_{i=1}^M \mathbb{I}\left[f(\theta_i, y) < f(\randvar{\tilde\theta}, y)\right] = r\right | y \right) \\ 
  =	\int_{Y}\mathrm{d}y \: \pi_\text{marg}\left(y\right)   \int_{0}^{1}\mathrm{d}z\,\binom{M}{r}z^{r}\left(1-z\right)^{M-r} \frac{\partial}{\partial z} q_{\phi,f}(z | y) \notag \\
 =\int_{0}^{1}\mathrm{d}z\binom{M}{r}z^{r}\left(1-z\right)^{M-r}\frac{\partial}{\partial z}\int_{Y}\mathrm{d}y\,\pi_\text{marg}\left(y\right)q_{\phi,f}\left(z | y\right).
\end{gather*}
Since claim 2 assumes $\phi$ satisfies SBC w.r.t.\ $f$, we have $\frac{\partial}{\partial z}\int_{Y}\mathrm{d}y\,\pi_\text{marg}\left(y\right)q_{\phi,f}\left(z | y\right) = 1$ and thus we can proceed as in claim 1:
\begin{gather*}
\int_Y \mathrm{d}y \: Q_{\phi,f}(j | y) \pi_\text{marg}(y) = 
\sum_{r = 0}^j \mbox{Pr} \left(\sum_{i=1}^M \mathbb{I} \left[f(\theta_i, y) < f(\randvar{\tilde\theta}, y)\right] = r\right) \\
   = \sum_{r = 0}^j  \frac{1}{M + 1} = \frac{j + 1}{M + 1}.
\end{gather*}
\end{proof}

With those foundations ready we now focus on proving statements about continuous SBC with the understanding that thanks to Theorems~\ref{ath:procedural_sbc}-\ref{ath:continuous_implies_sample} they also shed light on SBC as is actually implemented in software.

\begin{theorem}[Correct posterior and $q$] 
\label{ath:SBC_correct}
For any $y \in Y$, if $\forall \theta \in \Theta: \phi(\theta | y) = \pi_\text{post}(\theta | y)$ then for any test quantity $f$ we have $\forall x \in [0, 1] : q_{\phi, f}(x | y) = x$.
\end{theorem}

\begin{proof}
When $\pi$ has ties w.r.t.\ $f$, we can use Lemma~\ref{le:tie_removal} to construct a model with no ties with the same continuous $q$, preserving correctness. So we only need to prove the case where there are no ties.

Without ties, we have $\forall s \in \bar{\mathbb{R}}, y \in Y: C_{\phi,f}(s | y) = C_f(s | y)$ and $\forall x \in [0,1], y \in Y: C_f \left(C^{-1}_f( x | y)\right) = x$ (Lemma~\ref{le:no_ties_invert}) and thus
\begin{gather*}
    q_{\phi,f}(x | y) = \int_Y \mathrm{d}y \: C_{f}\left(\left.C^{-1}_{\phi,f}(x | y) \right| y \right) \pi_\text{marg}(y) = 
    \int_Y \mathrm{d}y \: x \pi_\text{marg}(y) \\
    = x \int_Y \mathrm{d}y \: \pi_\text{marg}(y) = x.
\end{gather*}
\end{proof}

\begin{corollary*}
The correct posterior passes continuous SBC. Combining the result with Theorem~\ref{ath:continuous_implies_sample}, the correct posterior will produce uniform $Q$ and pass $M$-sample SBC for any $M$ and $f$.
\end{corollary*}

Passing SBC is a necessary condition to have a correct posterior distribution, but it is not a sufficient condition. The following theorem characterizes a sufficient condition for having a correct distribution for a given test quantity.

\begin{theorem}[Characterization of SBC failures]
\label{ath:characterization_failures}
Given a posterior family $\phi$ over $Y$, $\Theta$ and test quantity $f$, then for all $y \in Y$ and $s \in \mathbb{R} : C_{\phi,f}(s | y) = C_f(s | y)$ if and only if $\forall x \in [0, 1] : q_{\phi,f}(x | y) = x$.
\end{theorem}
\begin{proof}
Given a model $\pi$ and a test quantity $f$, we build a model $\pi^\prime$ over the same data space and a univariate parameter space $\Theta^\prime := \mathbb{R}$ by first setting:
\begin{align*}
\pi^\prime_\text{marg}(y) &:= \pi_\text{marg}(y).
\end{align*}
We can then define the true posterior and the posterior family $\phi^\prime$ via their univariate CDFs, for all $s \in \mathbb{R}, y \in Y$:
\begin{align*}
\int_{-\infty}^{s} \mathrm{d}\theta^\prime \: \pi^\prime_\text{post}(\theta^\prime | y) &:= C^\pi_f(s | y) \\
\int_{-\infty}^{s} \mathrm{d}\theta^\prime \: \phi^\prime(\theta^\prime | y) 
  &:= C_{\phi,f}(s | y).
\end{align*}
This always defines a valid joint distribution $\pi_\text{joint}^\prime(\theta^\prime, y) = \pi^\prime_\text{post}(\theta^\prime | y)\pi^\prime_\text{marg}(y)$ and thus a valid model.

We now consider the projection function $p : \mathbb{R} \times Y \to \mathbb{R}$, $p(\theta^\prime, y) = \theta^\prime$. By construction of $\pi^\prime$ and $\phi^\prime$ we have for all $ s \in \mathbb{R}, x \in [0,1], y \in Y$:
\begin{align*}
C^{\pi^\prime}_p(s | y) &= C^\pi_f(s | y)\\
C^{\pi^\prime}_{\phi^\prime,p}(s | y) &= C^\pi_{\phi, f}(s | y).
\end{align*}
This then implies $\forall x \in [0,1], y \in Y:$
$$
q^{\pi^\prime}_{\phi^\prime,p}(x | y)  = q^{\pi}_{\phi,f}(x | y).
$$
It therefore suffices to prove that for $\pi^\prime$,
$$
\forall \theta^\prime \in \mathbb{R}: \phi^\prime(\theta^\prime, y) =  \pi_\text{post}^\prime(\theta^\prime |y) \iff \forall x \in [0,1] : q^{\pi^\prime}_{\phi^\prime,p}(x | y) = x.
$$
The forward implication follows from Theorem~\ref{ath:SBC_correct} because the correct posterior always produces correct $q$.
According to Lemma~\ref{le:tie_removal} it suffices to prove the reverse implication only for posteriors without ties. In this case using Lemma~\ref{le:no_ties_invert} we have for all $x \in [0,1]$:
$$
x = q^{\pi^\prime}_{\phi^\prime,p}(x | y) = C^{\pi^\prime}_p(C^{\pi^\prime, -1}_{\phi,p}(x | y) | y).
$$
So for all $y \in Y$ we see that $C^{\pi^\prime, -1}$ is the inverse of $C^{\pi^\prime}_p$ and thus the CDFs have to be equal:
$$
\forall s \in \mathbb{R}: C^{\pi^\prime}_p(s | y) = C^{\pi^\prime}_{\phi,p}(s | y).
$$
Since there is only a single parameter $\theta^\prime$ in $\pi^\prime$, the equality of the CDFs for the projection function implies the equality of the densities, and
$$
\forall \theta^\prime \in \mathbb{R}: \phi^\prime(\theta^\prime| y)  =  \pi_\text{post}^\prime(\theta^\prime |y).
$$
This completes the proof.
\end{proof}

This characterizes SBC failures in the sense that SBC investigates the average over $Y$ of the $q_{\phi,f}$ functions, which can appear correct even if for some subset of $Y$ with non-zero marginal mass the $q_{\phi,f}$ functions do not yield expected results. The consequence is that whenever we have an incorrect distribution of a given test quantity, the failure can be isolated to some $\bar{Y} \subset Y$ and detected by separately considering $\int_{\bar{Y}} \mathrm{d}y \: q_{\phi,f}(x | y)$. While in practical uses of SBC inspecting rank distributions for subset of the simulations based on the realized $y \in Y$ leads to loss of power, we can often get similar results by using a suitable test quantity $g$ that depends on $f$ and $y$ in a way that causes the problems in individual $y$ to accumulate instead of canceling. This is one of the reasons why allowing test quantities to depend on data is useful in general. 

We will now show that our SBC variant is in some sense complete. Previous work was correct in noting that with test quantities that are only a function of the parameters some problematic posteriors (e.g., the unaltered prior distribution) can never fail SBC. However, when test quantities are allowed to depend on observed data we can in principle discover any problematic $\phi$. For this we will need a technical lemma.

\begin{lemma}[Value of q]
\label{le:value_of_q}
Let $\phi$ be any posterior family over $Y$, $\Theta$ and $f$ a test quantity and let $s = C^{-1}_{\phi, f}(x | y)$. Then for all $y \in Y, x \in [0,1]$,
\begin{equation*}
q_{\phi,f}\left(x\mid y\right) = \begin{cases}
C_{f}\left(s\mid y\right)+\frac{D_{f}\left(s\mid y\right)}{D_{\phi,f}\left(s\mid y\right)}\left(x-C_{\phi,f}\left(s\mid y\right)\right) & \text{for } D_{\phi,f}\left(s\mid y\right) > 0 \\
C_{f}\left(s\mid y\right) & \text{otherwise}.
\end{cases}    
\end{equation*}
\end{lemma}

\begin{proof}
For a given $x$, we split $\Theta$ into three disjoint sets:
\begin{align*}
\Theta_{1}^x &:=\left\{ \theta \in \Theta :x > C_{\phi,f}\left(f\left(\theta,y\right)\mid y\right)\right\} \\
\Theta_{2}^x &:=\left\{ \theta  \in \Theta :C_{\phi,f}\left(f\left(\theta,y\right)\mid y\right)-D_{\phi,f}\left(f\left(\theta,y\right)\mid y\right) < x \leq C_{\phi,f}\left(f\left(\theta,y\right)\mid y\right)\right\} \\
\Theta_{3}^x &:=\left\{ \theta \in \Theta :x \leq C_{\phi,f}\left(f\left(\theta,y\right)\mid y\right)-D_{\phi,f}\left(f\left(\theta,y\right)\mid y\right)\right\}.
\end{align*}   
Now using the definition of $q$
\begin{equation*}
q_{\phi,f}(x|y) = \sum_{i = 1}^3 \int_{\Theta_i^x} \mathrm{d}\theta \: \pi_\text{post}(\theta | y) r_{\phi,f}(x | \theta, y).
\end{equation*}
The definition of $r_{\phi,f}$ implies that $\theta \in \Theta_1^x \implies r_{\phi,f}(x | \theta, y) = 1$ and thus:
\begin{gather*}
\int_{\Theta_1^x} \mathrm{d} \tilde\theta \: \pi_\text{post}(\tilde\theta | y) r_{\phi,f}(x | \tilde\theta, y) = 
\int_{\Theta_1^x} \mathrm{d} \tilde\theta \: \pi_\text{post}(\tilde\theta | y) \\
= \int_\Theta \mathrm{d} \tilde\theta \: \pi_\text{post}(\tilde\theta | y) \mathbb{I}\left[x >  C_{\phi,f}\left(f\left(\theta,y\right)\mid y\right)\right] \\
= 1 - \int_\Theta \mathrm{d} \tilde\theta \: \pi_\text{post}(\tilde\theta | y) \mathbb{I}\left[x \leq  C_{\phi,f}\left(f\left(\theta,y\right)\mid y\right)\right] \\
= 1 - \int_\Theta \mathrm{d} \tilde\theta \: \pi_\text{post}(\tilde\theta | y) \mathbb{I}\left[C^{-1}_{
\phi,f}(x |y)\leq  f\left(\theta,y\right)\right] \\
= \int_\Theta \mathrm{d} \tilde\theta \: \pi_\text{post}(\tilde\theta | y) \mathbb{I}\left[f\left(\theta,y\right) < s  \right] = 
C_f(s | y) - D_f(s | y), 
\end{gather*}
using the Galois connection between quantile and CDF functions, so that \\ $x \leq C_{\phi,f}(s | y) \iff C^{-1}_{\phi,f}(x | y) \leq s$ even when $C_{\phi,f}$ is not invertible.

The definition also directly implies $\theta \in \Theta_3^x \implies r_{\phi,f}(x | \theta, y) = 0$, and thus
\begin{equation*}
\int_{\Theta_3^x} \mathrm{d}\theta \: \pi_\text{post}(\theta | y) r_{\phi,f}(x | \theta, y)  = 0.
\end{equation*}
If $D_{\phi,f}\left(s\mid y\right) = 0$ then $\Theta_2^x$ is empty, and we have $q_{\phi,f}(x | y) = C_f(s|y)$. 

If $D_{\phi,f}\left(s\mid y\right) > 0$, then we first note that $\theta \in \Theta_2^x \iff f(\theta, y) = s$. Then, straight from the definition of continuous rank, the rank distribution is uniform between $C_{\phi,f}(s | y) - D_{\phi,f}(s|y)$ and $C_{\phi,f}(s | y)$, and the corresponding CDF thus has a linear segment, specifically:
$$
\theta \in \Theta_2^x \implies r_{\phi,f}(x | \theta, y) = \frac{x - (C_{\phi,f}(s | y) - D_{\phi,f}(s|y)
)}{D_{\phi,f}(s | y)} = \frac{x - C_{\phi,f}(s | y)}{D_{\phi,f}(s | y)} + 1.
$$
This let's us evaluate the integral over $\Theta_2^x$ as
\begin{gather*}
\int_{\Theta_2^x} \mathrm{d}\theta \: \pi_\text{post}(\theta | y) r_{\phi,f}(x | \theta, y)  = 
\int_{\Theta_{2}^x} \mathrm{d}\theta \: \pi_\text{post}(\theta|y) \left( \frac{x - C_{\phi,f}(s | y)}{D_{\phi,f}(s | y)} + 1 \right) \\ 
= \left( \frac{x - C_{\phi,f}(s | y)}{D_{\phi,f}(s | y)} + 1 \right) \int_\Theta \mathrm{d}\theta \: \pi_\text{post}(\theta | y)\mathbb{I}\left[f(\theta, y) = s\right] \\
= \left( \frac{x - C_{\phi,f}(s | y)}{D_{\phi,f}(s | y)} + 1 \right) D_f(s|y).
\end{gather*}
Combining with the previous results yields,
\begin{equation*}
    q_{\phi,f}\left(x\mid y\right) = 
C_{f}\left(s\mid y\right)+\frac{D_{f}\left(s\mid y\right)}{D_{\phi,f}\left(s\mid y\right)}\left(x-C_{\phi,f}\left(s\mid y\right)\right)  \text{for } D_{\phi,f}\left(s\mid y\right) > 0.
\end{equation*}
\end{proof}


\begin{theorem}[Density ratio]
\label{ath:density_ratio}
For any posterior family $\phi$, take $g\left(\theta,y\right)=\frac{\pi_{\mathrm{post}}\left(\theta\mid y\right)}{\phi\left(\theta|y\right)}$. Then $\phi$ passes continuous SBC w.r.t.\ $g$ if and only if $\pi_{\mathrm{post}}$ and $\phi$ are equal except for a set of measure $0$:

\begin{equation}
\int_Y \mathrm{d}y \int_\Theta \mathrm{d}\theta \: \pi_\text{joint}(y, \theta) \mathbb{I}\left[ \pi_\text{post}(\theta | y) \neq \phi(\theta | y)\right] = 0. \label{eq:phi_correct_almost_everywhere} 
\end{equation}
\end{theorem}

\begin{proof}
First we establish that $\forall y \in Y, s \in \mathbb{R}: C_{\phi, g}(s | y) \geq C_g(s | y)$:
\begin{gather*}
C_{\phi,g}\left(s \mid y\right)=\int_\Theta\mathrm{d}\theta\,\phi\left(\theta|y\right)\mathbb{I}\left[g\left(\theta,y\right)\leq s\right] \\
C_{g}\left(s \mid y\right)=\int_\Theta\mathrm{d}\theta\,\pi_{\mathrm{post}}\left(\theta\mid y\right)\mathbb{I}\left[g\left(\theta,y\right)\leq s\right]=\int\mathrm{d}\theta\,\phi\left(\theta|y\right)g\left(\theta,y\right)\mathbb{I}\left[g\left(\theta,y\right)\leq s\right].
\end{gather*}
First consider $s \leq 1$. We can rearrange:
\begin{gather*}
C_{\phi,g}(s | y) - C_g(s|y) = \int_{\theta \in \Theta: g(\theta, y) \leq s} \mathrm{d}\theta \:  \phi(\theta | y) \left(1 -   g\left(\theta,y\right) \right).
\end{gather*}
Since $\forall \theta \in \Theta, y \in Y: \phi(\theta | y) \geq 0$, the integral only covers non-negative values, and so the overall difference has to be non-negative.

Now consider $s > 1$.  Because $\int_\Theta\mathrm{d}\theta\:\phi\left(\theta|y\right)g\left(\theta,y\right) = \int_\Theta\mathrm{d}\theta\, \pi_\text{post}(\theta | y) = 1$, we can again rearrange:
\begin{gather*}
C_{\phi,g}(s | y) - C_g(s|y) 
= \int_{\theta \in \Theta: g(\theta, y) \leq s} \mathrm{d}\theta \:  \phi(\theta | y)  -  \int_{\theta \in \Theta: g(\theta, y) \leq s} \mathrm{d}\theta \:   \phi(\theta | y) g\left(\theta,y\right)   \\
= 1 - \int_{\theta \in \Theta: g(\theta, y) > s} \mathrm{d}\theta \:  \phi(\theta | y)  - 1 +  \int_{\theta \in \Theta: g(\theta, y) > s} \mathrm{d}\theta \:   \phi(\theta | y) g\left(\theta,y\right) \\
= \int_{\theta \in \Theta: g(\theta, y) > s} \mathrm{d}\theta \:  \phi(\theta | y) \left(g\left(\theta,y\right) - 1 \right).
\end{gather*}
where we once again integrate over non-negative values. 

So the difference is non-negative for any $s$. Also the difference cannot be zero for all $s$ once there is region of nonzero mass where $g(\theta, y) \neq 1$. In other words $\forall s: C_{\phi, g}(s | y) = C_g(s | y) $ if and only if the condition in \eqref{eq:phi_correct_almost_everywhere} holds.

Since $C_{\phi, g}(s | y) \geq C_g(s | y)$ we have $C^{-1}_{\phi,g}(s | y) \leq C^{-1}_{g}(s | y)$. And so due to the monotonicity of CDF and quantile functions, $C_g(C^{-1}_{\phi,g}(x | y) | y) \leq  C_g(C^{-1}_{g}(x | y) | y)$. If there are no ties, we can apply  Lemma~\ref{le:no_ties_invert} and have for all $y \in Y, x \in [0,1]$:
\begin{gather}
q_{\phi,g}(x | y) = C_g(C^{-1}_{\phi,g}(x | y) | y) \leq  C_g(C^{-1}_{g}(x | y) | y) = x. \label{eq:density_ratio_inequality}
\end{gather}
Thus $\int_Y \mathrm{d}y \, q_{\phi,g}(x | y) \pi_\text{marg}(y) \leq x$, where equality holds only if Equation (\ref{eq:phi_correct_almost_everywhere}) is satisfied. When ties are present, we recall Lemma~\ref{le:value_of_q} and make several observations:

\begin{enumerate}
    \item Whenever $C_{\phi_g}(C^{-1}_{\phi,g}(x | y) | y) = x$ or $D_{\phi,g}(C^{-1}_{\phi,g}(x | y)) = 0$ we have $q_{\phi,g}(x | y) = C_g(C^{-1}_{\phi,g}(x | y) | y)$.
    \item Across any open interval $(a,b): 0 < a < b < 1$ such that for all $x \in (a,b)$ we have $C_{\phi_g}(C^{-1}_{\phi,g}(x | y) | y) \neq x$ and $D_{\phi,g}(C^{-1}_{\phi,g}(x | y)) > 0$ the $q$ function is a line segment, i.e. $x \in (a,b) \implies q_{\phi,g}(x | y) = cx + d$ for some $c,d \in \mathbb{R}$.
    \item At any point $0 < x < 1$ such that $C_{\phi_g}(C^{-1}_{\phi,g}(x | y) | y) = x$ and $D_{\phi,g}(C^{-1}_{\phi,g}(x | y)) > 0$ the $q$ function is right-continuous.
    \item Straight from its definition $r_{\phi,g}$ is a non-decreasing function of $x$ and thus $q_{\phi,g}$ is also a non-decreasing function of $x$.
\end{enumerate}

Putting those together we obtain that even if ties are present, \eqref{eq:density_ratio_inequality} holds for some $x$ directly (via point 1 above). At any other point, the $q$ function is linear and due to points 3 and 4 we also have $q_{\phi,g}(x | y) \leq x$ and equality cannot hold at those intermediate points unless it also holds at the edges of the linear segment. Therefore even when ties are present, $\int_Y \mathrm{d}y \, q_{\phi,g}(x | y) \pi_\text{marg}(y) \leq x$ and equality holds only if (\ref{eq:phi_correct_almost_everywhere}) is satisfied.

\end{proof}

\begin{corollary*}
For any posterior family $\phi$ that would provide different posterior inferences than the true posterior $\pi_\text{post}$, there is a test quantity $f$ such that $\phi$ does not pass SBC w.r.t.\ $f$.
\end{corollary*}

The density ratio is typically not available in practice---and if it were, it would likely make more sense to directly check whether $g(\theta, y) = 1$. The above theorem however provides some intuition why using the likelihood $f_\text{lik}: f_\text{lik}(\theta, y) = \pi_\text{obs}(y | \theta)$ is useful as it is related to the density ratio, but not dependent on particular $\phi$ and usually easily available in practice.

We now generalize the result that ignoring data in test quantities implies that some problems in the posterior family can never be detected. We use a technical lemma, a continuous analog of Theorem~\ref{ath:procedural_sbc}.

\begin{lemma}[Integral representation of SBC] 
\label{le:sbc_integral}
Given a test quantity $f$ and a posterior family $\phi$ such that both $\pi$ and $\phi$ have no ties w.r.t.\ $f$, $\phi$ satisfies continuous SBC w.r.t.\ $f$ if and only if for all $x \in [0,1]$:
\begin{gather}
  \int_Y \mathrm{d}y \int_\Theta \mathrm{d} \tilde\theta \:
  \mathbb{I}\left[
  \int_\Theta \mathrm{d}\theta \: \left(\mathbb{I}\left[f(\theta, y) < f(\tilde\theta, y)\right] \phi(\theta | y)\right)  \leq x
  \right] \pi_\text{obs}(y | \tilde\theta) \pi_\text{prior}(\tilde\theta) = x.
  \label{eq:sbc_integral}
\end{gather}
\end{lemma}

\begin{proof}

We directly apply the definition of $C_{\phi,f}$ and use the absence of ties to get the invertibility of $C_{\phi,f}$ results from Lemma~\ref{le:no_ties_invert}.

\begin{gather*}
 \int_Y \mathrm{d}y \int_\Theta \mathrm{d} \tilde\theta \:
  \mathbb{I}\left[
  \int_\Theta \mathrm{d}\theta \: \left(\mathbb{I}\left[f(\theta, y) < f(\tilde\theta, y)\right] \phi(\theta | y)\right)  \leq x
  \right]  \pi_\text{obs}(y | \tilde\theta) \pi_\text{prior}(\tilde\theta) \notag\\
  = 
  \int_Y \mathrm{d}y \int_\Theta \mathrm{d} \tilde\theta \:
  \mathbb{I}\left[
  C_{\phi,f}(f(\tilde\theta, y) | y)  \leq x
  \right]  \pi_\text{obs}(y | \tilde\theta) \pi_\text{prior}(\tilde\theta) \notag\\
  = \int_Y \mathrm{d}y \int_\Theta \mathrm{d} \tilde\theta \:
  \mathbb{I}\left[
  f(\tilde\theta, y) \leq C^{-1}_{\phi,f}(x | y)
  \right]   \pi_\text{obs}(y | \tilde\theta) \pi_\text{prior}(\tilde\theta)
  \\
  = \int_Y \mathrm{d}y \: C_{f}\left(C^{-1}_{\phi,f}(x|y)|y\right) \pi_\text{marg}(y) 
  = \int_Y \mathrm{d}y \: q_{\phi,f}(x|y) \pi_\text{marg}(y).
\end{gather*}
\end{proof}

\begin{theorem}[Incomplete use of data] 
\label{ath:incomplete_use_data}
Assume a model $\pi$ with observation space $Y$ and parameter space $\Theta$, a space $Y^\prime$, and a measurable function $t : Y \rightarrow Y^\prime$. Denote the set $t^{-1}(y^\prime) = \{y \in Y: t(y) = y^\prime\}$. Consider the model $\pi^\prime$ with parameter space $\Theta$ and observation space $Y^\prime$ such that for all $\theta \in \Theta, y^\prime \in Y^\prime$:
\begin{align*}
\pi^\prime_\text{prior}(\theta) &= \pi_\text{prior}(\theta) \\
\pi^\prime_\text{obs}(y^\prime |\theta) &= \int_{t^{-1}(y^\prime)} \mathrm{d}y \: \pi_\text{obs}(y | \theta).
\end{align*}

Assume a test quantity $f^\prime : Y^\prime \times \Theta \rightarrow \mathbb{R}$. 
If we have a posterior family $\phi^\prime$ on $Y^\prime, \Theta$ such that $\phi^\prime$ passes continuous SBC w.r.t.\ $f^\prime$ and set test quantity $f: Y \times \Theta \rightarrow \mathbb{R}, f(\theta, y) = f^\prime(\theta, t(y))$ and posterior family $\phi$ on $\Theta, Y$ such that $\phi(\theta | y) = \phi^\prime(\theta | t(y))$ then $\phi$ passes continuous SBC w.r.t.\ $f$.
\end{theorem}

Here, the choice of $t$ lets us choose which aspects of the data are ignored, if $\forall y \in Y: t(y) = 1$, we recover the case where all data are ignored:  $\pi^\prime_\text{post}(\theta | y) = \pi_\text{prior}(\theta)$ and thus $\phi(\theta |y) = \pi_\text{prior}(\theta)$ will pass SBC w.r.t.\ $f$. If $t$ is a bijection, no information is lost. Other choices of $t$ then let us interpolate between those two extremes, for example ignoring just a subset of the data points or treating some data points as censored.

\begin{proof}

Start with the integral representation of SBC for $\phi$ w.r.t.\ $f$ from Lemma~\ref{le:sbc_integral}.
\begin{gather*}
\int_Y \mathrm{d}y \int_\Theta \mathrm{d} \tilde\theta \:
  \mathbb{I}\left[
  \int_\Theta \mathrm{d}\theta \: \mathbb{I} \left[f(\theta, y) < f(\tilde\theta, y)\right]\phi(\theta | y)  \leq x
  \right]  \pi_\text{obs}(y | \tilde\theta) \pi_\text{prior}(\tilde\theta) \\ =
\int_Y \mathrm{d}y \int_\Theta \mathrm{d} \tilde\theta \:
  \mathbb{I}\left[
  \int_\Theta \mathrm{d}\theta \: \mathbb{I} \left[f^\prime(\theta, t(y)) < f^\prime(\tilde\theta, t(y))\right]\phi^\prime(\theta, t(y))  \leq x
  \right] \pi_\text{obs}(y | \tilde\theta) \pi_\text{prior}(\tilde\theta) \\ =
\int_{Y^\prime} \mathrm{d}y^\prime \int_{t^{-1}(\bar{y})} \mathrm{d} y \int_\Theta \mathrm{d} \tilde\theta \:
  \mathbb{I}\left[
  \int_\Theta \mathrm{d}\theta \: \mathbb{I} \left[f^\prime(\theta, y^\prime) < f^\prime( \tilde\theta, y^\prime)\right]\phi^\prime(\theta, y^\prime)  \leq x
  \right]  \pi_\text{obs}(y | \tilde\theta) \pi_\text{prior}(\tilde\theta) \\  =
\int_{Y^\prime} \mathrm{d}y^\prime \int_\Theta \mathrm{d} \tilde\theta \:
  \mathbb{I}\left[
  \int_\Theta \mathrm{d}\theta \: \mathbb{I} \left[f^\prime(\theta, y^\prime) < f^\prime( \tilde\theta, y^\prime)\right]\phi^\prime(\theta, y^\prime)  \leq x
  \right]  \int_{t^{-1}(y^\prime)} \mathrm{d} y \: \pi_\text{obs}(y | \tilde\theta) \pi_\text{prior}(\tilde\theta) \\ =
\int_{Y^\prime} \mathrm{d}y^\prime \int_\Theta \mathrm{d} \tilde\theta \:
  \mathbb{I}\left[
  \int_\Theta \mathrm{d}\theta \: \mathbb{I} \left[f^\prime(\theta, y^\prime) < f^\prime( \tilde\theta, y^\prime)\right]\phi^\prime(\theta, y^\prime)  \leq x
  \right]  \pi^\prime_\text{obs}(y | \tilde\theta) \pi^\prime_\text{prior}(\tilde\theta).
\end{gather*}
The last step is an integral representation of SBC for $\phi^\prime$ with respect to $f^\prime$ from Lemma~\ref{le:sbc_integral}.
\end{proof}

\begin{corollary*}
If we perform SBC with $f_1, \dots, f_n$ such that there exist $P_1,P_2 \subset Y$, $\int_{P_k} \mathrm{d}y \: \pi_\text{marg}(y) > 0$ where $\forall i, \forall y_1 \in P_1, y_2 \in P_2, \forall \theta \in \Theta: f_i(y_1, \theta) = f_i(y_2, \theta)$ and $\int_{P_1} \mathrm{d} y_1 \int_{P_2} \mathrm{d} y_2 \int_\Theta \mathrm{d}\theta \: \mathbb{I}\left[\pi_
\text{obs}(y_1 | \theta) \neq \pi_\text{obs}(y_2 | \theta)\right]$, then there will exist $\phi$ passing SBC w.r.t.\ $f_1, \dots, f_n$ that is distinct from the correct posterior $\pi_\text{post}(\theta | y)$.
\end{corollary*}

\begin{proof}
If the conditions hold, then one can construct $t$ that merges $P_1$ and $P_2$ while preserving the values of $f_i$. The correct posterior for the implied model $\pi^\prime$ will imply an incorrect posterior for $\pi$ that passes SBC.
\end{proof}

\begin{lemma}[Order-preserving transformations]
\label{le:order_preserving_transformations}
Given test quantities $f,g: Y \times \Theta   \rightarrow  \mathbb{R}$ such that $\forall y \in Y ,\theta_1, \theta_2 \in \Theta: \: g(\theta_1, y) < g(\theta_2, y) \iff f(\theta_1, y) < f(\theta_2, y)$ and $g(\theta_1, y) = g(\theta_2, y) \iff f(\theta_1, y) = f(\theta_2, y)$ and a posterior family $\phi$ then for all $y \in Y,  \tilde\theta \in \Theta, x \in [0, 1], M \in \mathbf{N}, i \in \{0, \dots, M \}$: 
\begin{gather*}
r_{\phi, f}(x | \tilde\theta, y) = r_{\phi, g}(x | \tilde\theta, y) \\
R_{\phi, f}(i | \tilde\theta, y) = R_{\phi, g}(i | \tilde\theta, y).
\end{gather*}
\end{lemma}

\begin{proof}
The result follows directly from definitions of $r$ and $R$ as they depend only on ordering of the values.
\end{proof}

\begin{lemma}[Reverse transformation]
\label{le:reverse_transformation}
Given test quantities $f,g: Y \times \Theta   \rightarrow  \mathbb{R}$ such that $\forall y \in Y,\theta \in \Theta: g(\theta, y) = -f(\theta, y)$, then for all $y \in Y,  \tilde\theta \in \Theta, x \in [0, 1], M \in \mathbf{N}, i \in \{0, \dots, M \}$: 
\begin{gather*}
r_{\phi, g}(x | \tilde\theta, y) = 1 - r_{\phi, f}(1 - x | \tilde\theta, y) \\
R_{\phi, g}(i | \tilde\theta, y) = 1 - R_{\phi, f}(M - i - 1 | \tilde\theta, y).
\end{gather*}
\end{lemma}

\begin{proof}
First consider the continuous rank CDF $r$. We first establish relations between the CDFs and tie probabilities:
\begin{gather*}
C_{\phi,g}\left( s | y\right)=\int_{\Theta}\mathrm{d}\theta \, \mathbb{I}\left[g\left(\theta, y \right) \leq s\right]\phi(\theta | y) = 
\int_{\Theta}\mathrm{d}\theta \,\mathbb{I}\left[-f\left(\theta, y \right) \leq s\right]\phi(\theta | y) = \\
= \int_{\Theta}\mathrm{d}\theta\,\mathbb{I}\left[f\left(\theta, y \right) \geq -s\right]\phi(\theta | y) = 1 - C_{\phi,f}(-s | y) + D_{\phi,f}(-s|y) \\
D_{\phi,g}(s|y) = D_{\phi,f}(-s|y).
\end{gather*}
We can then proceed to directly evaluate $r_{\phi, g}$:

\begin{gather*}
r_{\phi, g}(x | \tilde\theta, y) = \mbox{Pr}\left(C_{\phi,g} \left(g \left(\left.\tilde\theta, y \right) \right| y \right) - U D_{\phi,g} \left(g \left(\left.\tilde\theta, y \right) \right| y \right) \leq x\right) \\ =
\mbox{Pr}\left(C_{\phi,g} \left(-f \left(\left.\tilde\theta, y \right) \right| y \right) - U D_{\phi,g} \left(\left.-f \left(\tilde\theta, y \right) \right| y \right) \leq x\right)  \\ =
\mbox{Pr}\left(1 - C_{\phi,f} \left(f \left(\left.\tilde\theta, y \right) \right| y \right) + (1 - U) D_{\phi,f} \left(\left.f \left(\tilde\theta, y \right) \right| y \right) \leq x\right) \\ =
\mbox{Pr}\left( C_{\phi,f} \left(\left.f \left(\tilde\theta, y \right) \right| y \right) - (1 - U) D_{\phi,f} \left(\left.f \left(\tilde\theta, y \right) \right| y \right) \geq  1 - x\right) \\ =
1 - \mbox{Pr}\left( C_{\phi,f} \left(\left.f \left(\tilde\theta, y \right) \right| y \right) - U D_{\phi,f} \left(\left.f \left(\tilde\theta, y \right) \right| y \right) \leq  1 - x\right) \\ =
1 - r_{\phi, f}(1 -x | \tilde\theta, y),  
\end{gather*}
where $U$ is uniformly distributed on $[0,1]$ and the second-to-last identity holds, because the probability of exact equality is $0$ and $1 - U$ has the same distribution as $U$.

Now we can focus on the sample rank CDF. We first distinguish how the behaviour of the $N_\text{equals}$ and $N_\text{less}$ random variables (from Definition~\ref{def:sample_SBC}) relates between $f$ and $g$:
\begin{gather*}
N_{\phi,g,\tilde\theta,y}^{\mathtt{equals}} = N_{\phi,f,\tilde\theta,y}^{\mathtt{equals}}
\\
N_{\phi,g,\tilde\theta,y}^{\mathtt{less}} = M - N_{\phi,f,\tilde\theta,y}^{\mathtt{less}} - N_{\phi,f,\tilde\theta,y}^{\mathtt{equals}}.
\end{gather*}
We can then directly evaluate $R_{\phi,g}$:
\begin{gather*}
R_{\phi,g}(i | \tilde\theta, y) 
= \mbox{Pr} \left(N_{\phi,g,\tilde\theta,y}^{\mathtt{less}} + K_{\phi,g,\tilde\theta,y} \leq i \right) \\=
  \mbox{Pr} \left(M - N_{\phi,f,\tilde\theta,y}^{\mathtt{less}} - N_{\phi,f,\tilde\theta,y}^{\mathtt{equals}} + K_{\phi,f,\tilde\theta,y} \leq i \right) \\=
  \mbox{Pr} \left(N_{\phi,f,\tilde\theta,y}^{\mathtt{less}} +N_{\phi,f,\tilde\theta,y}^{\mathtt{equals}} - K_{\phi,f,\tilde\theta,y} \geq  M  -i \right) \\ =
  1 - \mbox{Pr} \left(N_{\phi,f,\tilde\theta,y}^{\mathtt{less}} + K_{\phi,f,\tilde\theta,y} <   M - i \right) \\ =
  1 - \mbox{Pr} \left(N_{\phi,f,\tilde\theta,y}^{\mathtt{less}} + K_{\phi,f,\tilde\theta,y} \leq   M - i - 1 \right) = 
  1 - R_{\phi,f}(M - i - 1 | \tilde\theta, y).
\end{gather*}
\end{proof}

\begin{theorem}[Monotonic transformations]
\label{ath:monotonous_transformations}
Assume test quantities $f,g$ and a set of measurable functions $h_y: \mathbb{R} \to \mathbb{R}$ such that $\forall y \in Y ,\theta \in \Theta: f(\theta, y) = h_y(g(\theta_1, y))$ and a posterior family $\phi$. If either for all $y \in Y: h_y$ is strictly increasing or  for all $y \in Y: h_y$ is strictly decreasing then 1) $\phi$ passes continuous SBC w.r.t.\ $f$ if and only if $\phi$ passes continuous SBC w.r.t.\ $g$ and 2) $\phi$ passes $M$-sample SBC w.r.t.\ $f$ if and only if $\phi$ passes $M$-sample SBC w.r.t.\ $g$.    
\end{theorem}

\begin{proof}
This is a direct consequence of Lemmas \ref{le:order_preserving_transformations} and \ref{le:reverse_transformation}. The test quantity $g$ has to be either an order-preserving transform or a sequence of an order-preserving transform and a reverse transform. Since order-preserving transforms do not change the rank CDFs at all, SBC results will be exactly identical. The only remaining thing to prove is that a reverse transform maintains uniform rank distribution in both the sample and continuous case.

Assuming $\forall y \in Y: h_y$ is strictly decreasing, we get for the sample case:
\begin{gather*}
 Q_{\phi,g}(i | y) = \int_\Theta \mathrm{d}\tilde\theta \: \pi_\text{post}(\tilde\theta|y)  R_{\phi,g}( i| \tilde\theta, y) \\ =
\int_\Theta \mathrm{d}\tilde\theta \: \pi_\text{post}(\tilde\theta|y) - \int_\Theta \mathrm{d}\tilde\theta \: \pi_\text{post}(\tilde\theta|y) R_{\phi,f}(M - i - 1| \tilde\theta, y) \\ =
1 - Q_{\phi,f}(M - i - 1|y),   
\end{gather*}
which then implies,
\begin{gather*}
    \int_Y \mathrm{d}y\: Q_{\phi,g}(i | y) \pi_\text{marg}(y) = 
\int_Y \mathrm{d}y\: \pi_\text{marg}(y) - \int_Y \mathrm{d}y\: Q_{\phi,f}(M - i - 1 | y) \pi_\text{marg}(y) \\ =
1 - \int_Y \mathrm{d}y\: Q_{\phi,f}(M - i - 1 | y).
\end{gather*}
So if $\phi$ passes SBC w.r.t.\ $f$, then,
\begin{equation*}
    \int_Y \mathrm{d}y\: Q_{\phi,g}(i | y) \pi_\text{marg}(y) = 
1 - \frac{M - i}{M + 1} = \frac{i + 1}{M + 1},
\end{equation*}
and thus $\phi$ passes SBC w.r.t.\ $g$. If instead we first assume $\phi$ passes SBC w.r.t.\ $g$, then,
\begin{gather*}
 \int_Y \mathrm{d}y\: Q_{\phi,f}(M - i - 1 | y) =  1 - \int_Y \mathrm{d}y\: Q_{\phi,g}(i | y) \pi_\text{marg}(y) \\ = 1 - \frac{i + 1}{M + 1} = \frac{(M - i - 1) + 1}{M + 1} .  
\end{gather*}
After substituting $j = M - i - 1$ this shows $\phi$ passes SBC w.r.t.\ $f$.

For the continuous case,
\begin{gather*}
    q_{\phi,g}(x|y) = 
\int_\Theta \mathrm{d} \tilde\theta \: \pi_\text{post}(\tilde\theta | y) r_{\phi,g}(x | \tilde\theta, y) \\ 
= \int_\Theta \mathrm{d} \tilde\theta \: \pi_\text{post}(\tilde\theta | y) - \int_\Theta \mathrm{d} \tilde\theta \: \pi_\text{post}(\tilde\theta | y) r_{\phi,f}(1 - x | \tilde\theta, y) =
1 - q_{\phi,f}(1 - x | y),
\end{gather*}
which then implies,
\begin{gather*}
    \int_Y \mathrm{d}y \: q_{\phi,g}(x|y) \pi_\text{marg}(y) = \int_Y \mathrm{d}y \:  \pi_\text{marg}(y) - \int_Y \mathrm{d}y \: q_{\phi,f}(1 - x|y) \pi_\text{marg}(y) \\ =
    1 -  \int_Y \mathrm{d}y \: q_{\phi,f}(1 - x|y) \pi_\text{marg}(y), 
\end{gather*}
where both directions of implication for passing SBC w.r.t.\ $f$ and $g$ follow immediately.
\end{proof}

Theorem~\ref{ath:monotonous_transformations} shows that certain simple transformations of test quantities lead to equivalent behaviour for SBC. The result cannot be easily strengthened as many transformations in fact can lead to different behaviours for SBC. Appendix B provides several counterexamples: Example~\ref{ex:bernoulli_likelihood} shows that if we allow $h_y$ to be strictly increasing for some $y \subset Y$ while other $y$ will have $h_y$ strictly decreasing, we obtain a different check. Example~\ref{ex:non_monotonoues_bijection} shows that we can obtain different behaviour if we combine a test quantity with a discontinuous non-monotonic bijection, and Example~\ref{ex:ties_continous} shows the same when a whole range of values is projected onto a single point. In all those examples, the transformed test quantities rule out some posterior families that pass SBC for the original, but there also are posterior families not passing SBC for the original that pass SBC for the transformed quantity.

In some discussions, SBC is linked to the concept of data-averaged posterior, which equals the prior:

\begin{equation}
\forall \theta \in \Theta: \pi_\text{prior}(\theta) = \int_Y \mathrm{d} y \int_\Theta \mathrm{d}\tilde\theta  \: \pi_\text{post}(\theta |y) \pi_\text{obs}(y | \tilde\theta) \pi_\text{prior}(\tilde \theta). \label{eq:data_averaged_posterior}
\end{equation}
As with SBC, one can use this condition to compare the prior and the data-averaged posterior for some test quantity derived from the parameters and (optionally) data. 

\begin{definition}[Data-averaged posterior] A posterior family $\phi$ has the correct data-averaged posterior for a test quantity $f$ if for all $s \in \mathbb{R}$
\begin{gather*}
\int_Y \mathrm{d}y \int_\Theta \mathrm{d}\theta \: \mathbb{I}\left[ f(\theta, y) < s \right] \pi_\text{obs}(y|\theta) \pi_\text{prior}(\theta) \\
= \int_\Theta \mathrm{d}\theta \int_Y \mathrm{d} y \int_\Theta \mathrm{d}\tilde\theta \: \mathbb{I}\left[f(\theta, y) < s\right] \phi(\theta | y) \pi_\text{obs}(y|\tilde\theta) \pi_\text{prior}(\tilde\theta).
\end{gather*}
\end{definition}
Even restricting to test quantities that do not depend on data, there will be posterior families with correct data-averaged posterior for a test quantity but not passing SBC as well as posterior families with incorrect data-averaged posterior that pass SBC. Example~\ref{ex:bernoulli_projection_dap} in Appendix B shows both cases.  

The space of posterior families with correct data-averaged posteriors is closed under convex combinations: if $\phi_1$ and $\phi_2$ have the correct data-averaged posterior for $f$, then for any $0 \leq a \leq 1, \bar{\phi} : \bar{\phi}(\theta | y) = a \phi_1(\theta |y) + (1 - a)\phi_2(\theta|y)$ will also have correct data averaged posterior for $f$. However a convex combination of two posterior families passing SBC w.r.t.\ $f$ will not in general pass SBC; again, see Example~\ref{ex:bernoulli_projection_dap} for a counterexample.

One intuitive way to understand the difference between SBC and data-averaged posterior is that following Lemma~\ref{le:value_of_q} and examples in Appendix B we see that passing SBC is naturally understood as a constraint on the quantile function of the test quantity while a correct data-averaged posterior is naturally expressed as a constraint on the density of the test quantity.

This difference is relevant to broader discussions of SBC and related diagnostics, as other works sometimes characterize SBC as comparing prior to data-averaged posterior (e.g., \cite{lee_calibration_2019}, Appendix M.2 of \cite{pmlr-v130-lueckmann21a}, Equation 17 of \cite{schad_sbc_bf}, \cite{grinsztajn_bayesian_2021}, \cite{ramesh_gatsbi_2022}, \cite{saad_family_2019}).

There is however a sense in which SBC is a stronger condition than data-averaged posterior---as proved in Lemma 2.19 of \citet{cockayne_testing_2022}, a posterior family that passes SBC for \emph{all} test quantities not depending on data will have a correct data-averaged posterior for all test quantities not depending on data.

\section*{Appendix B: Formal analysis of SBC for a simple model}

Here we take a simple model and characterize the posterior families that pass SBC against several test quantities. This is intended to help build intuition about what SBC is actually doing. Additionally some of the solutions serve as counterexamples for some discussions in Appendix A. We will work with the model,
\begin{align}
\Theta &:= \mathbb{R} \notag\\
Y &:= \{0,1\} \notag\\
\randvar{\theta} &\sim \mathrm{uniform}(0,1) \notag\\
\randvar{y}  &\sim \mathrm{Bernoulli}(\randvar{\theta}) .\label{eq:bernoulli_model}
\end{align}
We know the correct posterior and marginal distributions are 
\begin{align}
    \pi_\text{marg}(0) &= \pi_\text{marg}(1) = \frac{1}{2} \notag\\
    \pi_\text{post}(\theta | 0) &= \mathrm{beta}(\theta | 1, 2) = 2(1-\theta) \notag\\
    \pi_\text{post}(\theta | 1) &= \mathrm{beta}(\theta | 2, 1) = 2\theta.
    \label{eq:bernoulli_correct}
\end{align}
We denote $\Phi(p | y) := \int_0^p \mathrm{d}\theta \phi(\theta | y)$, the CDF of the posterior given by $\phi$, and $\Phi^{-1}(x | y)$ is the associated inverse CDF.

Since $Y$ has just two elements, passing SBC for this model translates to:
\begin{gather}
    \forall x \in [0,1] : \frac{1}{2} (q_{\phi,f}(x | 0) + q_{\phi,f}(x | 1)) = x. \label{eq:bernoulli_sbc} 
\end{gather}

Further, since for continuous models there are typically no ties we have $q_{\phi,f}(x | y) = C_f(C^{-1}_{\phi,f}(x | y) | y)$ (Lemma~\ref{le:no_ties_invert}). All the following examples thus have the same structure:

\begin{enumerate}
    \item Determine $C_f$;
    \item Express $C^{-1}_{\phi,f}$ in terms of $\Phi^{-1}$;
    \item Express $q_{\phi,f}$ in terms of $\Phi^{-1}$;
    \item Solve (\ref{eq:bernoulli_sbc}) for $\Phi^{-1}$.
\end{enumerate}
We focus on $\Phi^{-1}$ instead of $\phi$ to make the computation easier as $\Phi^{-1}$ occurs naturally in expressions for $C^{-1}_{\phi,f}$. It however has the side effect that the inverse CDF can be immediately used to sample from the posterior family and thus compare the theoretical results to simulations of SBC.

All the examples are accompanied by simulations that show how specific incorrect posteriors manifest with regards to a set of test quantities the simulation results can be found at \url{https://martinmodrak.github.io/sbc_test_quantities_paper}, and the underlying code at \url{https://github.com/martinmodrak/sbc_test_quantities_paper}.

\begin{example}[Projection]
\label{ex:bernoulli_projection}

First consider the projection $f_1(\theta, y) := \theta$, so that,
\begin{gather}
C_{f_1}(s | 0) = \int_0^s \mathrm{d}\theta \: 2(1 - \theta)  =  2s - s^2 \notag \\
C_{f_1}(s | 1) = \int_0^s \mathrm{d}\theta \: 2\theta = s^2 \notag \\
C_{\phi, f_1}^{-1}(x | y) = \Phi^{-1}(x | y).
\end{gather}
Plugging this into (\ref{eq:bernoulli_sbc}) yields, for all $x \in [0,1]$,
\begin{equation}
\Phi^{-1}(x | 0) - \frac{(\Phi^{-1}(x| 0))^2}{2} + \frac{(\Phi^{-1}(x | 1))^2}{2}  = x.
\label{eq:bernoulli_sbc_condition}
\end{equation}
This means we are relatively free to choose one of the inverse CDFs and solve for the other, for all $x \in [0,1]$:
\begin{equation}
  \Phi^{-1}(x | 1) = \sqrt{2x + (\Phi^{-1}(x | 0) - 1)^2 - 1}. \label{eq:bernoulli_sbc_solve}
\end{equation}
Given a valid inverse CDF $\Phi^{-1}(x | 1)$ on $[0, 1]$, Equation (\ref{eq:bernoulli_sbc_solve}) defines a valid $\phi$ as long as: 
\begin{align}
    \forall x, 0 \leq x \leq \frac{1}{2} : \Phi^{-1}(x | 0) &\leq 1 - \sqrt{1 - 2x} \label{eq:sbc_solve_cond1} \\
    \forall x, \frac{1}{2} \leq x \leq 1 : \Phi^{-1}(x | 0) &\geq 1 - \sqrt{2 - 2x} \label{eq:sbc_solve_cond2} \\
    \frac{\partial \Phi^{-1}(x | 0)}{\partial x} (\Phi^{-1}(x | 0) - 1) &\geq -1, \label{eq:sbc_solve_cond3}
\end{align}
where (\ref{eq:sbc_solve_cond1}) and (\ref{eq:sbc_solve_cond2}) ensure a non-negative value inside the square root in (\ref{eq:bernoulli_sbc_solve}), and (\ref{eq:sbc_solve_cond3}) ensures $\Phi^{-1}(x | 1)$ is nondecreasing.

While this allows for a wide array of incorrect solutions, passing SBC for this condition already avoids some pathological behavior, e.g. the low quantiles of the posterior when $0$ is observed cannot be too high. If we instead solve (\ref{eq:bernoulli_sbc_condition}) for $\Phi^{-1}(x | 1)$ we obtain that when $1$ is observed, high quantiles cannot be to low; specifically, $\forall \frac{1}{2} < x < 1: \Phi^{-1}(x | 1) > \sqrt{2x - 1})$.
\end{example}

\begin{example}[Projection and data-averaged posterior]
\label{ex:bernoulli_projection_dap}
We can further build upon Example~\ref{ex:bernoulli_projection} to illustrate the differences between passing SBC and having correct data-averaged posterior. Correctness of the data-averaged posterior for the projection $f_1(\theta, y) = \theta$ simply entails for all $\theta \in [0,1]$:
\begin{equation}
\phi(\theta | 0) + \phi(\theta | 1) = 2.  
\end{equation}
For example, flipping the correct posterior to
\begin{gather*}
  \phi_A(\theta, 0) := \pi_\text{post}(\theta | 1) = 2 \theta \\
  \phi_A(\theta, 1) := \pi_\text{post}(\theta | 0) = 2 - 2 \theta
\end{gather*}
satisfies the data-averaged posterior criterion. This however does not pass SBC w.r.t.\ $f_1$ as the implied inverse CDF $\Phi^{-1}_A(x | 0) = \sqrt{x}$ violates all the conditions in (\ref{eq:sbc_solve_cond1})--(\ref{eq:sbc_solve_cond3}).

On the other hand, we can take a simple posterior family satisfying SBC for $f_1$ (via Equation~\ref{eq:bernoulli_sbc_solve}):
\begin{gather*}
\Phi^{-1}_B(x | 0) := \begin{cases}
   \frac{2}{3}x & x < \frac{3}{4} \\
   \frac{1}{2} + 2(x - \frac{3}{4}) & x \geq \frac{3}{4} \\
\end{cases} \\
\Phi^{-1}_B(x | 1) := \begin{cases}
   \frac{1}{3}\sqrt{6x + 4 x^2} & x < \frac{3}{4} \\
   \sqrt{3 - 6x + 4 x ^2} & x \geq \frac{3}{4}.
\end{cases}    
\end{gather*}
The implied density has an incorrect data-averaged posterior:
\begin{gather*}
\phi_B(\theta | 0) = \begin{cases}
 \frac{3}{2} & \theta \leq \frac{1}{2} \\
 \frac{1}{2} & \theta > \frac{1}{2}
\end{cases}\\
\phi_B(\theta | 1) = \begin{cases}
 \frac{3s}{\sqrt{1 + 4 s ^ 2}} & \theta \leq \frac{\sqrt{3}}{2} \\
 \frac{s}{\sqrt{4s ^2  - 3}} & \theta > \frac{\sqrt{3}}{2}.
\end{cases}
\end{gather*}
Finally, \citet{lee_calibration_2019} in Equation 1.3 claim that several methods including SBC cannot notice when $\phi$ is a convex combination of the prior and the correct posterior. This is not true; the authors derive it from the incorrect assumption that SBC checks for the correct data-averaged posterior. We can readily build a counterexample: consider $\phi_C(\theta| y) := \frac{1}{2}\left(\pi_\text{post}(\theta| y) + \pi_\text{prior}(\theta)\right)$. We then have:
\begin{align}
    \phi_C(\theta, 0) &= \frac{3}{2} - \theta \notag\\
    \phi_C(\theta, 1) &= \frac{1}{2} + \theta \notag\\
    \Phi_C(s | 0) &= \frac{1}{2}(3s - s^2) \notag\\
    \Phi_C(s | 1) &= \frac{1}{2}(s + s^2) \notag\\    
    \Phi^{-1}_C(x | 0) &= \frac{3}{2} - \frac{1}{2}\sqrt{9 - 8x} \notag\\
    \Phi^{-1}_C(x | 1) &= -\frac{1}{2} + \frac{1}{2}\sqrt{1 + 8x} \label{eq:phi_c_1}.
\end{align}
However, plugging $\Phi^{-1}_C(x | 0)$ into (\ref{eq:bernoulli_sbc_solve}) which is a necessary condition for passing SBC w.r.t\ $f_1$ yields,
\begin{equation*}
\Phi^{-1}_{C^\prime}(x | 1) = \frac{\sqrt{3 - \sqrt{9 - 8x}}}{\sqrt{2}},
\end{equation*}
which is in conflict with what we derived in \eqref{eq:phi_c_1}  and $\phi_C$ thus cannot pass SBC w.r.t.\ $f_1$ (although it has correct data-averaged posterior).

\end{example}

\begin{example}[Likelihood]
\label{ex:bernoulli_likelihood}
Here we evaluate SBC for the likelihood,
\begin{equation*}
f_2(\theta, y) := \mathrm{Bernoulli}(y | \theta) = \begin{cases} 1 - \theta & y = 0 \\ \theta & y = 1. \end{cases}    
\end{equation*}
In this case,
\begin{align*}
C_{f_1}(s | 0) &= \int_0^1 \mathrm{d}\theta \: \mathbb{I}[1 - \theta < s]2(1 - \theta)  =  
  2 \int_{1-s}^1 \mathrm{d}\theta \: (1 - \theta) = 1 - 2s + s^2 \\
C_{f_2}(s | 1) &= C_{f_1}(s | 1) = \int_0^s \mathrm{d}\theta \: 2\theta = s^2 \\
C_{\phi,f_2}(s | 0) &= \int_0^1 \mathrm{d}\theta \: \mathbb{I}[1 - \theta < s]\phi(\theta | 0) = \int_{1-s}^1 \mathrm{d}\theta \: \phi(\theta |0) = 1 - \Phi(\theta | 0) \\
C^{-1}_{\phi,f_2}(x | 0) &= \Phi^{-1}(1 - x | 0)  \\
C^{-1}_{\phi,f_2}(x | 1) &= C^{-1}_{\phi,f_1}(x | 1) =  \Phi^{-1}(x | 1).
\end{align*}
Plugging this into (\ref{eq:bernoulli_sbc}) gives, $\forall x \in [0,1]$,
\begin{equation}
\frac{1}{2} - \Phi^{-1}(1 - x | 0) + \frac{(\Phi^{-1}(1 - x| 0))^2}{2} + \frac{(\Phi^{-1}(x | 1))^2}{2}  = x,
\label{eq:bernoulli_sbc_lik_condition}
\end{equation}
which we can solve for $\Phi^{-1}(1 - x | 0)$:
\begin{equation}
  \Phi^{-1}(1 - x | 0) = 1 - \sqrt{2x - (\Phi^{-1}(x|1) ^ 2)}.
\label{eq:bernoulli_sbc_loglik_solve}
\end{equation}
This defines a valid $\phi$ under the same conditions as SBC for the projection $f_1$; see (\ref{eq:sbc_solve_cond1})--(\ref{eq:sbc_solve_cond3}).

To combine this with the solution in \ref{eq:bernoulli_sbc_solve}, we additionally need to ensure continuity of $\Phi^{-1}(\frac{1}{2} | 0)$, so that computing the number via \ref{eq:bernoulli_sbc_solve} followed by \ref{eq:bernoulli_sbc_loglik_solve} matches the original value. This translates to:
\begin{equation}
\Phi^{-1}\left(\left.\frac{1}{2}\right| 0 \right) = 1 - \frac{\sqrt{2}}{2}.
\label{eq:sbc_solve_continuity_cond}
\end{equation}
Combining the projection and likelihood therefore fixes the midpoint of the quantile function to its true value for both observations. However, we are still able to choose $\Phi^{-1}(0, x)$ for $x \in [0, \frac{1}{2}]$ freely as long as the conditions \ref{eq:sbc_solve_cond1} - \ref{eq:sbc_solve_cond3} and \ref{eq:sbc_solve_continuity_cond} are met, and this then defines a $\phi$ that passes SBC w.r.t.\ both $f_1$ and $f_2$.

\end{example}

\begin{example}[Non-monotonous bijection]
\label{ex:non_monotonoues_bijection}
Consider the test quantity
\begin{equation*}
    f_3(y, \theta) := \begin{cases}\theta&  \text{for } \theta < \frac{1}{2} \\ \theta - 1,&  \text{otherwise.} \end{cases}
\end{equation*}
It might seem intuitively plausible that since $f_3$ captures the same information as $f_1$ that passing SBC w.r.t.\ $f_1$ will entail passing SBC w.r.t.\ $f_3$. This intuition proves to be incorrect.

Assuming $-\frac{1}{2} \leq s \leq \frac{1}{2}$, then the true CDF is,
\begin{align*}
    &C_{f_3}(s | 0) \\ 
    &= \int_0^1 \mathrm{d}\theta \: \mathbb{I}\left[f_3(\theta) \leq s\right] (2 - 2\theta) =
\begin{cases}
 \int_{\frac{1}{2}}^{s + 1} \mathrm{d}\theta \: (2 - 2\theta)  & -\frac{1}{2} < s < 0 \\
 \int_{\frac{1}{2}}^{1} \mathrm{d}\theta (2 - 2\theta) + \int_{0}^{s} \mathrm{d}\theta (2 - 2\theta) & 0 \leq s < \frac{1}{2}
\end{cases} \\
&= \begin{cases}
 \frac{1}{4} - s^2& -\frac{1}{2} < s < 0 \\
 \frac{1}{4} + 2 s - s^2  & 0 \leq s < \frac{1}{2}
\end{cases}
\\
&C_{f_3}(s | 1) =
\begin{cases}
 \int_{\frac{1}{2}}^{s + 1} \mathrm{d}\theta \: 2\theta  & -\frac{1}{2} < s < 0 \\
 \int_{\frac{1}{2}}^{1} \mathrm{d}\theta \: 2\theta + \int_{0}^{s} \mathrm{d}\theta \: 2\theta & 0 \leq s < \frac{1}{2}
\end{cases}
= \begin{cases}
 \frac{3}{4} + 2s + s^2& -\frac{1}{2} < s < 0 \\
 \frac{3}{4} + s^2  & 0 \leq s < \frac{1}{2}.
\end{cases}
\end{align*}
Let us further set $h_y := \Phi \left(\frac{1}{2} | y\right)$; then we can then express the fitted CDF as:
\begin{gather*}
    C_{\phi,f_3}(s | y) =
\begin{cases}
 \int_{\frac{1}{2}}^{s + 1} \mathrm{d}\theta \: \phi(\theta, y)  & -\frac{1}{2} < s < 0 \\
 \int_{\frac{1}{2}}^{1} \mathrm{d}\theta \: \phi(\theta, y) + \int_{0}^{s} \mathrm{d}\theta \: \phi(\theta, y) & 0 \leq s < \frac{1}{2}
\end{cases}  \\
= \begin{cases}
  \Phi(s + 1| y) - h_y & \text{for } -\frac{1}{2} <s < 0 \\
1 -   h_y + \Phi(s | y)& 0 \leq s < \frac{1}{2}
\end{cases}
\end{gather*}
and invert it to obtain:
\begin{gather*}
    C^{-1}_{\phi,f_3}(x | y) = \begin{cases}
\Phi^{-1}(x + h_y | y) - 1  & \text{ for } x < 1 - h_y \\
\Phi^{-1}(x - 1 + h_y |y) & \text{ otherwise.}
\end{cases}
\end{gather*}
We can now evaluate $q_{\phi,f_3}$:
\begin{align*}
q_{\phi,f_3}(x | 0) &= \begin{cases}
  \frac{1}{4} - \left(\Phi^{-1}(x + h_0 | 0) - 1 \right)^2 & \text{for }  x < 1 - h_0 \\
  \frac{1}{4} + 2 \Phi^{-1}(x - 1 + h_0|0) - \left( \Phi^{-1}(x - 1 + h_0 | 0)\right)^2  & \text{ otherwise}
\end{cases}
\\
q_{\phi,f_3}(x | 1) &= \begin{cases}
  -\frac{1}{4} + \left(\Phi^{-1}(x + h_1 | 1)\right)^2 & \text{for }  x < 1 - h_1 \\
  \frac{3}{4} + \left( \Phi^{-1}(x - 1 + h_1|1)\right)^2  & \text{ otherwise.}
\end{cases}    
\end{align*}
So the SBC condition $x = \frac{1}{2}(q_{\phi,f_3}(x | 0) + q_{\phi,f_3}(x | 1))$ resolves to four cases. First, when $x \leq \min\{1 - h_0, 1 - h_1\}$,
\begin{align*}    
2x &= - \left(\Phi^{-1}(x + h_0 | 0) - 1 \right)^2 + \left(\Phi^{-1}(x + h_1 | 1)\right)^2  \\
\Phi^{-1}(x + h_1 | 1) &= \sqrt{\left(\Phi^{-1}(x + h_0 | 0) - 1 \right)^2 + 2x}.
\end{align*}
Second, when $1 - h_0 \leq x \leq 1 - h_1$,
\begin{align*}    
2x &=  - \left( \Phi^{-1}(x - 1 + h_0 | 0) - 1\right)^2 + 1 + \left(\Phi^{-1}(x + h_1 | 1)\right)^2 \\
\Phi^{-1}(x + h_1 | 1) &= \sqrt{\left( \Phi^{-1}(x - 1 + h_0 | 0) - 1\right)^2 + 2x - 1}.
\end{align*}
Third, when $1 - h_1 \leq x \leq 1 - h_0$.
\begin{align*}    
2x &=-\left(\Phi^{-1}(x + h_0 | 0) - 1\right)^2 + 1 + \left( \Phi^{-1}(x - 1 + h_1|1)\right)^2  \\
\Phi^{-1}(x - 1 + h_1|1) &= \sqrt{\left(\Phi^{-1}(x + h_0 | 0) - 1\right)^2 + 2x - 1}.
\end{align*}
Fourth, when $x \geq \max\{1 - h_0, 1 - h_1\}$,
\begin{align*}
2x &= 1 - \left( \Phi^{-1}(x - 1 + h_0 | 0) - 1\right)^2 + 1 + \left( \Phi^{-1}(x - 1 + h_1|1)\right)^2 \\
\Phi^{-1}(x - 1 + h_1|1) &= \sqrt{\left( \Phi^{-1}(x - 1 + h_0 | 0) - 1\right)^2 + 2x - 2}.
\end{align*}
Substituting $\bar{h} := h_1 - h_0$ and $y := x + \bar{h} + h_0$ in the first two cases and $y := x - 1 + \bar{h} + h_0$ into the latter two cases, we obtain:
\begin{gather}    
\Phi^{-1}(y | 1) = 
  \begin{cases}
    \sqrt{\left(\Phi^{-1}(y - \bar{h}  | 0) - 1 \right)^2 + 2(y - \bar{h} - h_0)}
    & \text{for } \bar{h} < y \leq 1 + \bar{h} \\
    \sqrt{\left( \Phi^{-1}(y - 1 -\bar{h} | 0) - 1\right)^2 + 2(y -  \bar{h} - h_0) - 1}
    & \text{for } 1 + \bar{h} \leq y  \\
    \sqrt{\left(\Phi^{-1}(y + 1 - \bar{h} | 0) - 1\right)^2 + 2(y -  \bar{h} - h_0) + 1}
    & \text{for } y  \leq \bar{h}.
  \end{cases}
  \label{eq:bernoulli_swap_sbc_solve}
\end{gather}
We have one less case, after the substitution the first and fourth cases are identical.

There are a few conditions to make Equation~\ref{eq:bernoulli_swap_sbc_solve} define a valid posterior family. For brevity we do not evaluate all, but we  show the relation between $h_0$ and $h_1$. We need to ensure $\Phi^{-1}(0 | 1) = 0$, and so:
\begin{gather*}
\begin{cases}
    0 = \left(\Phi^{-1}(1 - h_1 + h_0 | 0) - 1\right)^2  - h_1 + 1
    & \text{for } 0  \leq \bar{h} \\
    0 = \left(\Phi^{-1}(h_0 - h_1  | 0) - 1 \right)^2 -2h_1
    & \text{for } \bar{h} < 0.
  \end{cases} 
\end{gather*}
This can be further rearranged to:
\begin{gather}
\begin{cases}
    \Phi^{-1}(1 - h_1 + h_0 | 0) =  1 - \sqrt{2h_1 - 1}
    & \text{for } h_0 \leq h_1 \\
    \Phi^{-1}(h_0 - h_1  | 0) = 1 - \sqrt{2h_1}
    & \text{for } h_0 > h_1.
  \end{cases}
  \label{eq:bernoulli_swap_h0_h1}
\end{gather}
So to construct a posterior family satisfying SBC w.r.t.\ $f_3$, we can choose almost any $\Phi^{-1}(x | 0)$, calculate $h_0$, then use Equation~\ref{eq:bernoulli_swap_h0_h1} to solve for $h_1$ and finally compute $\Phi^{-1}(x | 1)$ via (\ref{eq:bernoulli_swap_sbc_solve}).

Now we can show that this is indeed a different condition than the one derived in Example~\ref{ex:bernoulli_projection}. Let us take $\phi_B$ from Example~\ref{ex:bernoulli_projection_dap}---as discussed there, it passes SBC w.r.t.\ $f_1$: 

\begin{gather*}
\Phi^{-1}_B(x | 0) = \begin{cases}
   \frac{2}{3}x & x < \frac{3}{4} \\
   \frac{1}{2} + 2(x - \frac{3}{4}) & x \geq \frac{3}{4} \\
\end{cases} \\
\Phi^{-1}_B(x | 1) = \begin{cases}
   \frac{1}{3}\sqrt{6x + 4 x^2} & x < \frac{3}{4} \\
   \sqrt{3 - 6x + 4 x ^2} & x \geq \frac{3}{4}.
\end{cases}    
\end{gather*}
From $\Phi_B$,we obtain $h_0 = \frac{3}{4} > \frac{3}{4}(\sqrt{2} - 1) = h_1$, while 
\begin{gather*}
\Phi^{-1}_B(h_0 - h_1 | 0) =  \Phi^{-1}_B \left(\frac{3}{4}(2 - \sqrt{2}) | 0\right) = \frac{1}{2}(2 - \sqrt{2}) = 1 - \frac{1}{\sqrt{2}} \\
1 - \sqrt{2h_1} = 1 - \sqrt{\frac{3}{2}(\sqrt{2} - 1)} = 1 - \frac{\sqrt{3(\sqrt{2} - 1)}}{\sqrt{2}},
\end{gather*}
so the condition in (\ref{eq:bernoulli_swap_h0_h1}) does not hold and $\phi_B$ cannot pass SBC w.r.t.\ $f_3$.
\end{example}

\begin{example}[Ties - continuous]
\label{ex:ties_continous}
We obtain a yet different check when we deliberately choose to ignore some information. Take:
\begin{equation*}
    f_4(\theta, y) := \begin{cases}
\theta & \theta < \frac{1}{2} \\
\frac{1}{2} & \theta \geq \frac{1}{2}.
\end{cases}
\end{equation*}
We now have ties, and so we must evaluate the tie probabilities:
\begin{align*}
    C_{f_4}(s | 0) &= \begin{cases}
   \int_0^s \mathrm{d}\theta \: 2(1 - \theta)  =  2s - s^2  & s < \frac{1}{2}\\
   1 & s \geq \frac{1}{2}
\end{cases}\\
D_{f_4}(s | 0) &= \begin{cases}
    0 & s \neq \frac{1}{2}\\
   \int_{\frac{1}{2}}^1 \mathrm{d}\theta \: 2(1 - \theta)  =  \frac{1}{4} & s = \frac{1}{2}
\end{cases}\\
C_{f_4}(s | 1) &= \begin{cases}
  \int_0^s \mathrm{d}\theta \: 2\theta = s^2 & s < \frac{1}{2}\\
  1 & s \geq \frac{1}{2}
\end{cases} \\
D_{f_4}(s | 1) &= \begin{cases}
    0 & s \neq \frac{1}{2}\\
   \int_{\frac{1}{2}}^1 \mathrm{d}\theta \: 2\theta  =  \frac{3}{4} & s = \frac{1}{2}.
\end{cases}
\end{align*}
Reusing the notation $h_y := \Phi \left(\frac{1}{2}|y \right)$ from previous example, we get:
\begin{align*}
    C_{\phi, f_4}(s | y) &= \begin{cases}
   \Phi(s | y) & s < \frac{1}{2} \\
   1 & s \geq \frac{1}{2}
\end{cases}\\
    C^{-1}_{\phi, f_4}(x | y) &= \begin{cases}
   \Phi^{-1}(x | y) & x < h_y \\
   \frac{1}{2} & x \geq h_y
\end{cases}\\
D_{\phi, f_4}(s | y) &= \begin{cases}
  \int_{\frac{1}{2}}^1 \mathrm{d}\theta \: \phi (\theta | y) = 1 - h_y & s = \frac{1}{2} \\
  0 & s \neq \frac{1}{2}.
\end{cases}
\end{align*}
We can now use Lemma~\ref{le:value_of_q}:
\begin{gather*}   
q_{\phi,f_4}(x | y) = \begin{cases}
  C_{f_4}(\Phi^{-1}(x | y) | y)  & x < h_y \\
  C_{f_4}\left(\frac{1}{2}| y \right) + \frac{D_{f_4}\left(\frac{1}{2}| y \right)}{D_{\phi,f_4}\left(\frac{1}{2}| y \right)} \left(x - C_{\phi,f_4}\left(\frac{1}{2}|y \right) \right)& x \geq h_y
\end{cases} \\
= \begin{cases}
  C_{f_4}(\Phi^{-1}(x | y) | y)  & x < h_y \\
  1 + \frac{D_{f_4}\left(\frac{1}{2}| y \right)}{1 - h_y} \left(x - 1 \right)& x \geq h_y
\end{cases}\\
q_{\phi,f_4}(x | 0) = \begin{cases}
1 - \left(\Phi^{-1}(x | 0) - 1\right)^2 & x < h_0 \\
1 + \frac{x - 1}{4(1 - h_0)} & x \geq h_0
\end{cases} \\
q_{\phi,f_4}(x | 1) = \begin{cases}
\left(\Phi^{-1}(x | 1)\right)^2 & x < h_1 \\
1 + \frac{3(x - 1)}{4(1 - h_1)} & x \geq h_1.
\end{cases}
\end{gather*}
And the SBC condition becomes:
\begin{equation}
2x = \begin{cases}
  - \left(\Phi^{-1}(x | 0) - 1\right)^2 + \left(\Phi^{-1}(x | 1)\right)^2 + 1 & x< \min\{h_0, h_1\} \\
  - \left(\Phi^{-1}(x | 0) - 1\right)^2 + 2 + \frac{3(x - 1)}{4(1 - h_1)}& h_1 \leq x < h_0 \\
 1 + \frac{x - 1}{4(1 - h_0)} + \left(\Phi^{-1}(x | 1)\right)^2 & h_0 \leq x < h_1 \\
 2 + \frac{x - 1}{4(1 - h_0)} + \frac{3(x - 1)}{4(1 - h_1)} & \max\{h_0, h_1\} \leq x.
\end{cases}
\label{eq:bernoulli_clamp_sbc}
\end{equation}
If $\max\{h_0, h_1\} = 1$, then SBC fails as values of $f_4(\theta, y) = \frac{1}{2}$ are never generated in the posterior while appearing in the prior. If not, then the last branch implies,
\begin{equation}
h_1 = \frac{5 h_0 - 4}{8 h_0 - 7}, h_0 < \frac{4}{5}.
\label{eq:bernoulli_clamp_h1}    
\end{equation}
Substituting this into the second branch in \ref{eq:bernoulli_clamp_sbc}, we obtain a set of constraints on $\Phi^{-1}$:
\begin{align}
x< \min\{h_0, h_1\} & \implies \Phi^{-1}(x | 1) = \sqrt{2x + (\Phi^{-1}(x | 0) - 1)^2 - 1} \label{eq:bernoulli_clamp_cond_phiboth} \\
h_1 \leq x < h_0 & \implies \Phi^{-1}(x | 0) = 1 - \frac{1}{2}\sqrt{\frac{x - 1}{h_0 - 1}} \label{eq:bernoulli_clamp_cond_phi0}\\
h_0 \leq x < h_1 & \implies \Phi^{-1}(x | 1) = \frac{1}{2}\sqrt{\frac{3 - 4h_0 + x(8 h_0 - 7)}{h_0 - 1}} \label{eq:bernoulli_clamp_cond_phi1}
\end{align}
Take $\Phi_D^{-1}(x | 0) := x^2, \Phi_D^{-1}(x | 1) := \sqrt{2x - 2x^2 + x^4}$. This will pass SBC w.r.t.\ $f_1$, as it is derived directly via (\ref{eq:bernoulli_sbc_solve}). However, it will not pass SBC w.r.t.\ $f_4$ for multiple reasons. Probably the easiest to see is that in this case we have $h_0 = \sqrt{\frac{1}{2}}$ and $h_1 \approx 0.15$, so over a large range we would need $\Phi_D^{-1}(x | 0)$ proportional to square root of $x$ due to \eqref{eq:bernoulli_clamp_cond_phi0}. Additionally, condition  \eqref{eq:bernoulli_clamp_h1} entails $h_0 = \frac{1}{2} \implies h_1 = \frac{1}{2}$ and is therefore also violated. 

Passing SBC for $f_4$ restricts the functional form of a potentially big segment of the inverse CDFs via \eqref{eq:bernoulli_clamp_cond_phi0} or \eqref{eq:bernoulli_clamp_cond_phi1}, which makes it strict in this area while providing no constraints on the distribution of values above $\frac{1}{2}$.

Constructing a posterior family satisfying SBC for $f_4$ can then proceed as follows: Pick $0 < h_0 < \frac{4}{5}$, calculate $h_1$ via \eqref{eq:bernoulli_clamp_h1}. For $x > h_y$ the values of $\Phi(x | y)$ can be arbitrary (as long as they define valid quantile functions). We can also freely choose $\Phi^{-1}(x | 0)$ for $x < \min\{h_0, h_1\}$, but we need to ensure that (a) $\Phi^{-1}(x | 1)$ is a valid quantile function---the conditions derived for $f_1$ in (\ref{eq:sbc_solve_cond1})--(\ref{eq:sbc_solve_cond3}) restricted to $x < \min\{h_0, h_1\}$ are sufficient---and (b) $\Phi^{-1}(h_y | y) = \frac{1}{2}$. The latter condition implies that  if $h_0 < h_1$ then $\Phi(h_0 | 1) < \frac{1}{2}$ which we can combine with \eqref{eq:bernoulli_clamp_cond_phiboth} to get
$$
\sqrt{2h_0 + (\Phi^{-1}(h_0 | 0) - 1)^2 - 1} < \frac{1}{2},
$$ 
which reduces to $\frac{3}{8} \leq h_0 < \frac{1}{2}$. The condition $\Phi(h_1 | 1) = \frac{1}{2}$ is then already ensured by \eqref{eq:bernoulli_clamp_cond_phi1}. If on the other hand $h_0 \geq h_1$ then substituting $x = h_1$ into \eqref{eq:bernoulli_clamp_cond_phiboth} implies
$$
\sqrt{2h_1 + (\Phi^{-1}(h_1 | 0) - 1)^2 - 1} = \frac{1}{2},
$$
which reduces to $\Phi^{-1}(h_1|0) = 1 - \frac{1}{2}\sqrt{\frac{9}{7 - 8 h_0}}$ and $\frac{1}{2} \leq h_0 < \frac{25}{32}$.
The condition $\Phi(h_0 | 0) = \frac{1}{2}$ is then already ensured by \eqref{eq:bernoulli_clamp_cond_phi0}.

\end{example}

\begin{example}[Ties---discrete]

To get some intuition on the behaviour of SBC for discrete parameters, we further simplify the model into,
\begin{align}
\Theta &:= \left\{\frac{1}{3}, \frac{2}{3}\right\} \notag\\
Y &:= \{0,1\} \notag\\
\forall \theta \in \Theta: \pi_\text{prior}(\theta) &:= \frac{1}{2}  \notag\\
\randvar{y}  &\sim \mathrm{Bernoulli}(\randvar{\theta}). \label{eq:bernoulli_model_ties}
\end{align}
We know the correct posterior and marginal distributions are 
\begin{align}
    \pi_\text{marg}(0) &= \pi_\text{marg}(1) = \frac{1}{2} \notag\\
    \pi_\text{post}\left(\left.\frac{1}{3} \right| 0\right) &= \pi_\text{post}\left(\left.\frac{2}{3} \right| 1\right) = \frac{2}{3}\notag\\
    \pi_\text{post}\left(\left.\frac{1}{3} \right| 1\right) &= \pi_\text{post}\left(\left.\frac{2}{3} \right| 0\right) = \frac{1}{3}.
    \label{eq:bernoulli_ties_correct}
\end{align}
The posterior family is fully defined by just two numbers: $a = \phi\left(\frac{1}{3}|0\right)$ and $b = \phi\left(\frac{1}{3}|1\right)$. We can directly use the definitions to obtain for $f_1$ (i.e., just $\theta$):
\begin{gather*}
q_{\phi,f_1}(x | 0) = \begin{cases}
   \frac{2x}{3 a} & x \leq a \\
   \frac{2}{3} + \frac{x - a}{3\left(1 - a\right)} & \text{otherwise}
\end{cases} \\
q_{\phi,f_1}(x | 1) = \begin{cases}
   \frac{x}{3 b} & x \leq b \\
   \frac{1}{3} + \frac{2\left(x - b\right)}{3\left(1 - b\right)} & \text{otherwise.}
\end{cases}
\end{gather*}
For the case $ x \leq \min\left\{a,b\right\}$ the SBC condition becomes:
$$
2x = \frac{2x}{3a} + \frac{x}{3b}.
$$
this can hold for all $x$ in the range only if
\begin{gather}
2 = \frac{2}{3a} + \frac{1}{3b} \notag \\
a b - \frac{1}{3}b - \frac{1}{6}a = 0 \label{eq:disc_cond_1}.
\end{gather}
Now we can inspect the case $ x > \max\left\{a,b\right\}$ where the SBC condition becomes.
$$
2x = 1 + \frac{x-a}{1-a} +\frac{2\left(x - b\right)}{3\left(1 - b\right)}.
$$
Focusing on the coefficient for $x$ in the above equation yields
\begin{gather}
1 - \frac{1}{6(1-a)} - \frac{1}{3(1-b)} = 0 \notag \\
6ab - 4a - 5b + 3  = 0 \label{eq:disc_cond_2}.
\end{gather}
Combining \eqref{eq:disc_cond_1} and \eqref{eq:disc_cond_2} already leaves only two solution: either the posterior equals the prior ($\phi\left(\frac{1}{3}|0\right) = \phi\left(\frac{1}{3}|1\right) = \frac{1}{2}$) or $\phi = \pi_\text{post}$. Those solutions also satisfy all the other cases.

This example shows that when the underlying parameter space is discrete, there is less flexibility to craft $\phi$ to pass SBC and the space of posterior families passing SBC can have more structure than when no ties are present.

\end{example}

\end{document}